\documentclass[11pt]{article}
\usepackage[utf8]{inputenc}
\usepackage{amsfonts,latexsym,amsthm,amssymb,amsmath,amscd,euscript,mathrsfs,graphicx}
\usepackage{framed}
\usepackage{fullpage}
\usepackage{color}
\usepackage{enumitem}
\usepackage{authblk}
\usepackage{adjustbox}

\usepackage{algorithm}
\usepackage[noend]{algpseudocode}
\usepackage{todonotes}

\newtheorem{theorem}{Theorem}[section]
\newtheorem{proposition}[theorem]{Proposition}
\newtheorem{lemma}[theorem]{Lemma}

\theoremstyle{definition}
\newtheorem{defn}[theorem]{Definition}

\theoremstyle{remark}
\newtheorem*{remark}{Remark}

\newcommand{\BE}{\mathbb E}

\newcommand{\BI}{\mathbb I}

\newcommand{\BP}{\mathbb P}

\newcommand{\eps}{\varepsilon}
\newcommand{\TV}{d_{\text{TV}}}
\newcommand{\poly}{\text{poly}}
\newcommand{\Tr}{\text{Tr}}

\title{On Distribution Testing in the Conditional Sampling Model}
\author{Shyam Narayanan\thanks{Email: \texttt{shyamsn@mit.edu}. Research supported by the MIT Akamai Fellowship, the NSF Graduate Fellowship, and the Simons Investigator award.}}
\affil{Massachusetts Institute of Technology}
\date{\today}

\begin{document}

\maketitle

\begin{abstract}
    Recently, there has been significant work studying  distribution testing under the Conditional Sampling model. In this model, a query specifies a subset $S$ of the domain, and the output received is a sample drawn from the distribution conditioned on being in $S$. In this paper, we improve query complexity bounds for several classic distribution testing problems in this model.
    
    First, we prove that tolerant uniformity testing in the conditional sampling model can be solved using $\tilde{O}(\eps^{-2})$ queries, which is optimal and improves upon the $\tilde{O}(\eps^{-20})$-query algorithm of Canonne et al. \cite{CanonneRS15}. This bound even holds under a restricted version of the conditional sampling model called the Pair Conditional Sampling model. Next, we prove that tolerant identity testing in the conditional sampling model can be solved in $\tilde{O}(\eps^{-4})$ queries, which is the first known bound independent of the support size of the distribution for this problem. Next, we use our algorithm for tolerant uniformity testing to get an $\tilde{O}(\eps^{-4})$-query algorithm for monotonicity testing in the conditional sampling model, improving on the $\tilde{O}(\eps^{-22})$-query algorithm of Canonne \cite{Canonne15}. Finally, we study (non-tolerant) identity testing under the pair conditional sampling model, and provide a tight bound of $\tilde{\Theta}(\sqrt{\log N} \cdot \eps^{-2})$ for the query complexity, where the domain of the distribution has size $N$. This improves upon both the known upper and lower bounds in \cite{CanonneRS15}.
\end{abstract}

\newpage

\section{Introduction}

\subsection{Distribution Testing in the Sampling Model}

\textbf{Distribution testing} is a fundamental problem in statistics where the goal is to learn properties of a distribution $\mathcal{D}$ over an $N$-element set $[N]$, given several independent samples from $\mathcal{D}$. Given an arbitrary number of samples, it is quite simple to learn the distribution fully, but we wish to learn properties of $\mathcal{D}$ using a \emph{sublinear} number of samples, and preferably as few samples as possible. Unfortunately, one cannot learn the distribution of $\mathcal{D}$, even up to total variation distance $0.1$, without using at least $\Omega(N)$ samples (see, e.g., \cite[Theorem 25]{BatuFRSW00}). Instead, as in many problems in the field of property testing, one may require that the algorithm output ACCEPT with high probability if the distribution $\mathcal{D}$ has some desired property, but output REJECT with high probability if $\mathcal{D}$ has total variation distance at least $\eps$ from every distribution with the desired property. However, if $\mathcal{D}$ neither has the desired property nor has total variation distance at least $\eps$ from every distribution with the desired property, the algorithm is allowed to output either ACCEPT or REJECT. This freedom is what makes it possible for such an algorithm to determine properties of the distribution while only receiving a sublinear number of samples.

The study of distribution testing in the framework of property testing began nearly two decades ago with \cite{BatuFRSW00}, which initiated a long line of study in this area (see the surveys \cite{RubinfeldSurvey, BalakrishnanWSurvey, KamathThesis, CanonneSurvey} for a comprehensive list of results in the field). Problems which have been studied in the distribution testing framework have included \emph{uniformity testing}, or testing whether an unknown distribution $\mathcal{D}$ is uniform; \emph{identity testing}, or testing whether $\mathcal{D}$ is identical to a known distribution $\mathcal{D}^*$; \emph{equivalence testing}, or testing whether $\mathcal{D}_1$ and $\mathcal{D}_2$ are identical where we have sample access to both $\mathcal{D}_1$ and $\mathcal{D}_2$; and \emph{monotonicity testing}. Problems such as estimating the entropy of $\mathcal{D}$ and the support size of $\mathcal{D}$ have also been well-studied. A final important class of studied problems are the \emph{tolerant} versions of uniformity, identity, and equivalence testing, which mean trying to approximate the total variation distance between $\mathcal{D}$ and the Uniform distribution, $\mathcal{D}$ and $\mathcal{D}^*$, and $\mathcal{D}_1$ and $\mathcal{D}_2$, respectively. Tolerant testing provides much stronger functionality than basic hypothesis testing, as it provides meaningful guarantees even if the underlying distribution is known not to satisfy the hypothesis exactly.

While there exist sublinear algorithms for all of the problems listed above, the optimal algorithms are often only slightly sublinear, so they are not significantly more efficient than naive algorithms. For instance, estimating the entropy, the support size, and the distance from the uniform distribution or a fixed distribution $\mathcal{D}^*$ all have optimal sample complexities of $\Theta(N/\log N)$. Even for the simple problem of uniformity testing, the optimal algorithm requires $\Theta(\sqrt{N})$ samples. This motivated different models where one is allowed additional information besides simply sampling from the distribution. Some models involve sampling from distributions with additional structure. For instance, the distributions may be promised to be monotone \cite{RubinfeldS09}, to be $k$-modal \cite{DaskalakisDSVV13}, to be a low-degree Bayesian network \cite{CanonneDKS17, DaskalakisP17, AcharyaBDK18}, or to have some other property \cite{DaskalakisDK18, GheissariLP18, BezakovaBCSV19, DiakonikolasKP19}. In other settings, one is allowed additional types of queries, such as to either the Probability Mass Function (PMF) or the Cumulative Distribution Function (CDF) of the distribution \cite{BatuDKR05, GuhaMV06, RubinfeldS09, CanonneR14}. A related model is the probability-revealing samples model, where one is allowed samples $(x, p(x))$ where $x \sim \mathcal{D}$ and $p(x) = \BP_{y \sim \mathcal{D}}(y = x)$ \cite{OnakS18}. In this paper, however, we primarily study the \emph{Conditional Sampling Model}, which allows for sampling conditioned on some query set $S \subset [N]$. This model was first studied independently by Canonne et al. \cite{CanonneRS15} and Chakraborty et al. \cite{ChakrabortyFGM16}.

\subsection{The Conditional Sampling Model}

In the conditional sampling model, first developed by \cite{CanonneRS15, ChakrabortyFGM16}, the goal is still to test properties of a distribution $\mathcal{D}$ supported over $[N]$, but now we are given a stronger sampling oracle. This time, for each query, we are allowed to choose a subset $S \subseteq [N]$, and the oracle draws from the distribution $\mathcal{D}$ conditioned on $S$. Formally, we have the following.

\begin{defn} \cite{CanonneRS15} \label{CondDef}
    Fix a distribution $\mathcal{D}$ over $[N].$ A \Call{Cond}{} oracle for $\mathcal{D}$ is defined as follows. The oracle is given as input a query set $S \subseteq [N]$ and outputs an element $i \in S,$ where the probability that element $i$ is returned is $D(i)/D(S),$ where $D(i) = \BP_{x \sim \mathcal{D}} (x = i)$ and $D(S) = \BP_{x \sim \mathcal{D}} (x \in S).$ Moreover, each output of the oracle is independent of all previous calls.
\end{defn}

We note that the above definition only makes sense if $D(S) > 0.$ One way to fix this is to make sure that all probabilities are nonzero by slightly modifying any $i$ such that $D(i) = 0$ to be barely positive. In our case, this will not even matter, because every time that we call $\Call{Cond}{S}$, we will have previously sampled at least one element in $S$, so we know that $D(S) > 0$.

\bigskip

The simplest motivation for the conditional sampling model is that for many traditional distribution testing problems, the standard sampling model cannot provide a constant or poly-logarithmic query algorithm, and sometimes can only provide an $N^{1-o(1)}$ query algorithm. However, as we will see, for many standard distribution testing problems, only a poly-logarithmic or often even a constant number of queries to the conditional sampling oracle is required (assuming a fixed error parameter $\eps$). This makes the conditional sampling model very powerful even though at a first glance it may not look significantly more powerful than the regular sampling oracle. Moreover, for many problems, the full extent of the COND sampling power is not even needed. As we will discuss later, some problems can be solved using an oracle which only samples from either all of $[N]$ or pairs of elements in $[N]$ (PAIRCOND queries); or from other restricted versions of COND.

At the same time, various forms of conditional sampling are supported in multiple applied scenarios. For instance, the BlinkDB database system \cite{AgrawalMPMMS13} uses stratified random sampling to provide approximate answers to SQL aggregation queries over large volumes of data. Their system enables generating random samples of the data satisfying user-specified predicates, which are then used to approximate the aggregates of interests (counts, sums, etc). In particular, PAIRCOND queries considered here correspond to simple disjunctive predicates, where an attribute can have one of two fixed values.

Another appealing aspect of the conditional sampling model, as noted in \cite{CanonneRS15}, is that unlike in the standard sampling model, we are now able to deal with adaptive queries, since we can choose at each step which set $S$ to sample from. This leads to a richer class of algorithms than in the standard sampling model, where the only queries allowed are samples from the full data set $\mathcal{D}$. Therefore, the conditional sampling model leads to a much broader range of potential algorithms.

\subsection{Prior Work in the Conditional Sampling Model}

The conditional sampling model was initially studied by Chakraborty et al. \cite{ChakrabortyFGM16} and Canonne et al. \cite{CanonneRS15} independently. Since then, several works have established fast sublinear algorithms for a variety of distribution testing problems. Chakraborty et al. \cite{ChakrabortyFGM16} proved that uniformity testing with error $\eps$ could be done in $\poly(\eps^{-1})$ queries and that identity testing with error $\eps$ could be done in $\poly(\eps^{-1}, \log^*(N))$ queries. They also proved that computing entropy could be done in $\poly(\eps^{-1}, \log N)$ queries, by providing an algorithm for testing any label-invariant property of a distribution, i.e., a property which was invariant under a permutation of the elements of $[N]$. Canonne et al. \cite{CanonneRS15} gave nearly tight bounds for uniformity testing, by providing an $\tilde{O}(\eps^{-2})$-query algorithm and a nearly matching lower bound of $\Omega(\eps^{-2}).$ \cite{CanonneRS15} also provided an $\tilde{O}(\eps^{-4})$-query algorithm for the identity testing problem, which was improved to a nearly optimal $\tilde{O}(\eps^{-2})$-query algorithm by Falahatgar et al. \cite{FalahatgarJOPS15}. \cite{CanonneRS15} also provided an $\tilde{O}(\eps^{-4} \log^5 N)$-query algorithm for equivalence testing, which was also improved by \cite{FalahatgarJOPS15} to $\tilde{O}(\eps^{-5} \cdot \log \log N)$-queries. The best known lower bound for equivalence testing, however, is $\Omega(\sqrt{\log \log N})$, due to \cite{AcharyaCK18}. \cite{CanonneRS15} also provided an $\tilde{O}(\eps^{-20})$-query algorithm for tolerant uniformity testing. Canonne \cite{Canonne15} also studied monotonicity testing in the conditional sampling model, providing an $\tilde{O}(\eps^{-22})$-query algorithm for testing monotonicity. For an in-depth summary of results in the standard sampling model, conditional sampling model, and other related models, we point the interested reader towards Canonne's survey paper \cite{CanonneSurvey}.

We note that many variants of the conditional sampling model have been studied, most of which are more restrictive versions of the standard conditional sampling model. For instance, \cite{ChakrabortyFGM16, AcharyaCK18, KamathT19} also looked at testing in the nonadaptive conditional sampling model, where queries are not allowed to depend on the previous outputs of the oracle. Another variant is the subcube conditioning problem, where $N = 2^n$ and we treat $[N]$ as the set of vertices of an $n$-dimensional cube, but we are only allowed to conditionally sample from subcubes of the $n$-dimensional cube \cite{BlaisCG17, BhattacharyyaC18, CanonneCKLW19, ChenJLW20}. Two more variants, also investigated in \cite{CanonneRS15, ChakrabortyFGM16}, are the PAIRCOND and INTCOND sampling models. In the PAIRCOND model, all samples must either be sampled from the entire set $[N]$ or from the oracle $\Call{Pcond}{x, y},$ which samples from the conditional distribution of $\{x, y\} \subseteq [N].$ In other words, if we ever conditionally sample from a set $S \subseteq [N],$ either $S = [N]$ or $|S| = 2.$ In the INTCOND model, all conditional samples must be conditionally sampled on some interval $S = [a, b] = \{a, a+1, \dots, b\},$ i.e., we can only sample in intervals. In this work, we will provide new upper bounds in the COND model, as well as new upper and lower bounds in the PAIRCOND model.

\subsection{Our Results}

For a fixed $\mathcal{D}$, the \Call{Samp}{} oracle simply draws a random element $i$ from the distribution $\mathcal{D}.$ Recall the \Call{Cond}{} oracle from Definition \ref{CondDef}, and finally, define the \Call{Pcond}{} oracle as follows. \Call{Pcond}{} takes as input $x, y \in [N],$ and returns $x$ with probability $\frac{D(x)}{D(x)+D(y)}$ and returns $y$ with probability $\frac{D(y)}{D(x)+D(y)}.$ Equivalently, the outputs of \Call{Pcond}{$x, y$} and \Call{Cond}{$\{x, y\}$} have the same distributions. In the PAIRCOND model, one is allowed query access to \Call{Samp}{} and \Call{Pcond}{}.

We now describe the main results of this paper. See Table \ref{Table} for a summary of our main results as well as comparison to previous results (both upper and lower bounds). We primarily focus on tolerant testing, though we provide improved bounds on non-tolerant problems as well.

\begin{table}[ht]
    \centering
    \begin{adjustbox}{max width=\textwidth}
    \begin{tabular}{|c|c|c|c|}
    \hline
        Problem & COND & PAIRCOND & SAMP \\
    \hline
        Uniformity Testing & $\tilde{\Theta}(\eps^{-2})$ \cite{CanonneRS15} & $\tilde{\Theta}(\eps^{-2})$ \cite{CanonneRS15} & $\Theta\left(\frac{\sqrt{N}}{\eps^2}\right)$ \cite{Paninski08} \\
        Is $\mathcal{D}$ uniform? & & & \\ \hline
        Identity Testing & $\tilde{O}(\eps^{-2})$ \cite{FalahatgarJOPS15} & $\tilde{O}\left(\frac{\log^4 N}{\eps^4}\right)$ \cite{CanonneRS15}& $\Theta\left(\frac{\sqrt{N}}{\eps^2}\right)$ \cite{ValiantV17}\\
        Does $\mathcal{D} = \mathcal{D}^*$? & & $\tilde{\Omega}(\sqrt{\log N})$ \cite{CanonneRS15} & \\ 
        & & $\boldsymbol{\tilde{\Theta}\left(\frac{\sqrt{\log N}}{\eps^{2}}\right)}$ & \\ \hline
        Tolerant Uniformity & $\tilde{O}(\eps^{-20})$ \cite{CanonneRS15} & $\tilde{O}(\eps^{-20})$ \cite{CanonneRS15} & $O\left(\frac{1}{\eps^2} \cdot \frac{N}{\log N}\right)$ \cite{ValiantV11} \\
        What is $\TV(\mathcal{D}, \mathcal{U})$? & $\boldsymbol{\tilde{O}(\eps^{-2})}$ & $\boldsymbol{\tilde{O}(\eps^{-2})}$ & $\Omega\left(\frac{1}{\eps^2} \cdot \frac{N}{\log N}\right)$ \cite{JiaoHW16}\\ \hline
        Tolerant Identity & $\tilde{O}\left(\frac{\log^5 N}{\eps^5}\right)$ \cite{CanonneRS15} & & $O\left(\frac{1}{\eps^2} \cdot \frac{N}{\log N}\right)$ \cite{ValiantV11} \\
        What is $\TV(\mathcal{D}, \mathcal{D}^*)$? & $\boldsymbol{\tilde{O}(\eps^{-4})}$ & & $\Omega\left(\frac{1}{\eps^2} \cdot \frac{N}{\log N}\right)$ \cite{JiaoHW16}\\ \hline
        Monotonicity & $\tilde{O}(\eps^{-22})$ \cite{Canonne15} & $\tilde{O}\left(\frac{\log^2 N}{\eps^3} + \frac{\log N}{\eps^4}\right)$ \cite{CanonneSurvey} & $O\left(\frac{\sqrt{N}}{\eps^2}\right)$ \cite{AcharyaDK15} \\
        Is $D(1) \ge \cdots \ge D(N)$? & $\boldsymbol{\tilde{O}(\eps^{-4}})$ & & $\Omega\left(\frac{\sqrt{N}}{\eps^2}\right)$ \cite{Paninski08}\\ \hline
    \end{tabular}
    \end{adjustbox}
    \caption{List of query complexity bounds in the COND, PAIRCOND, and standard sampling (SAMP) models. Our new results are in bold. The distribution is always over some finite domain $[N] = \{1, 2, \dots, N\}$. We note that the $\tilde{O}(\eps^{-2})$ upper bound for Tolerant Uniformity in the COND model is directly implied by our $\tilde{O}(\eps^{-2})$ upper bound for Tolerant Uniformity in the PAIRCOND model, since PAIRCOND is a more restrictive model. Also, even if not stated, we note that all of the problems in the COND and PAIRCOND models are known to have a $\Omega(\eps^{-2})$ query lower bound that follows from the Uniformity Testing lower bound in COND due to \cite{CanonneRS15}.}
    \label{Table}
\end{table}

\medskip

\textbf{Tolerant Uniformity Testing:} The first main result we prove is a nearly optimal query complexity algorithm for tolerant uniformity testing, improving on the $\tilde{O}(\eps^{-20})$-query algorithm of \cite{CanonneRS15}. Like the result in \cite{CanonneRS15}, we only need the weaker PAIRCOND model.

\begin{theorem} \label{TolerantUnif}
    Let $\mathcal{U}$ be the uniform distribution over $[N]$. There is an algorithm \Call{TolerantUnif}{} that, given any distribution $\mathcal{D}$ and access to \Call{Samp}{} and \Call{Pcond}{}, uses $\tilde{O}(\eps^{-2})$ queries and determines the total variation distance $\TV(\mathcal{D}, \mathcal{U})$ up to an additive error of $\eps$ with probability at least $2/3$.
\end{theorem}

The above result is known to be nearly optimal in both the COND and PAIRCOND models, since even the standard uniformity testing problem, i.e., distinguishing between $\mathcal{D} = \mathcal{U}$ and $\TV(\mathcal{D}, \mathcal{U}) > \eps$, requires at least $\Omega(\eps^{-2})$ queries, even in the COND model.

\medskip

\textbf{Tolerant Identity Testing:} Our next main result is an algorithm for tolerant identity testing in the conditional sampling model, which we describe now.

\begin{theorem} \label{TolerantId}
    Let $\mathcal{D}^*$ be some fixed distribution over $[N]$. There is an algorithm \Call{TolerantId}{} that, given any distribution $\mathcal{D}$ and access to \Call{Cond}{}, uses $\tilde{O}(\eps^{-4})$ queries and determines the total variation distance $\TV(\mathcal{D}, \mathcal{D}^*)$ up to an additive error of $\eps$ with probability at least $2/3$.
\end{theorem}

While tolerant identity testing in the conditional sampling model has never been \emph{directly} addressed in the previous literature, we note that a $\tilde{O}(\log^5 N/\eps^5)$-query algorithm is implied by work of \cite{CanonneRS15, CanonneR14}. To briefly describe why, \cite{CanonneRS15} provides a $\tilde{O}(\log^5 N/\eps^3)$-query algorithm that they call APPROX-EVAL that closely approximates the probability of a randomly chosen element in the distribution with very high probability. The work of \cite{CanonneR14} is able to solve tolerant identity testing in $\eps^{-2}$ queries to EVAL, and this can be modified to just needing APPROX-EVAL, so the total query complexity is $\tilde{O}(\log^5 N/\eps^5).$ 

While our bounds are not optimal, we provide the first algorithm with query complexity that does not depend on the support size $N$ at all (even logarithmically) and grows only polynomially in $\eps^{-1}$. Moreover, our polynomial dependence on $\eps$ is quite reasonable. However, unlike the tolerant uniformity case, we require the full power of the COND model: indeed, there is a $\sqrt{\log N}$ lower bound even in the standard identity testing for PAIRCOND \cite{CanonneRS15}, which we will improve to $\sqrt{\log N}/\eps^2$. A natural open question is whether the bound of $\tilde{O}(\eps^{-4})$ for tolerant identity testing in COND can be improved to $\tilde{O}(\eps^{-2})$.

\medskip

\textbf{Monotonicity testing:} Next, we improve the best known algorithm for monotonicity testing in COND, using our algorithm for Theorem \ref{TolerantUnif}. In monotonicity testing, we require that algorithm accepts distributions $\mathcal{D}$ that are monotone, i.e., distributions with $D(1) \ge D(2) \ge \cdots \ge D(N)$, with high probability. However, they must reject distributions that are far from monotone, i.e., distributions $\mathcal{D}$ such that $\TV(\mathcal{D}, \mathcal{D}^*) \ge \eps$ for all monotone distributions $\mathcal{D}^*$, with high probability.

The best known algorithm for monotonicity testing in COND, due to Canonne \cite[Theorem 4.1]{Canonne15}, is divided into two subroutines. The first subroutine makes $O(\eps^{-2})$ calls to computing the total variation distance, up to error $\eps$, between some conditional subset and the uniform distribution on that subset. This previously required $O(\eps^{-22})$ queries (since previously each call to tolerant uniformity took $O(\eps^{-20})$ queries) but now only requires $O(\eps^{-4})$ queries. The second subroutine does not involve any calls to tolerant uniformity testing, and uses $\tilde{O}(\eps^{-8})$ queries. Thus, we already have an improved $\tilde{O}(\eps^{-8})$ query algorithm. But with a few additional insights, we can improve the second subroutine to a $\tilde{O}(\eps^{-4})$ query complexity, giving us the following theorem.

\begin{theorem} \label{TestMonotone}
    There is an algorithm \Call{TestMonotone}{} that, given any distribution $\mathcal{D}$ over $[N]$ and access to \Call{Cond}{}, if $\mathcal{D}$ is monotone, the algorithm outputs ACCEPT with probability at least $2/3,$ and if $\mathcal{D}$ is $\eps$-far from all monotone distributions, the algorithm outputs REJECT with probability at least $2/3.$ Moreover, the algorithm uses $\tilde{O}(\eps^{-4})$ queries.
\end{theorem}

A natural open problem is whether this upper bound can be improved to $\tilde{O}(\eps^{-2})$, which matches the known lower bound for monotonicity testing \cite{CanonneRS15, Canonne15}.

\medskip

\textbf{Identity/Uniformity Testing in PAIRCOND:} Our final result is a nearly optimal query complexity algorithm for identity testing in the PAIRCOND model, as well as a nearly matching lower bound, which we now state formally.

\begin{theorem} \label{IdPcondUpper}
    Let $\mathcal{D}^*$ be some fixed distribution over $[N]$. There is an algorithm \Call{PcondId}{} that, given any distribution $\mathcal{D}$ and access to \Call{Samp}{} and \Call{Pcond}{}, if $\mathcal{D} = \mathcal{D}^*$, the algorithm outputs ACCEPT with probability at least $2/3,$ and if $\TV(\mathcal{D}, \mathcal{D}^*) \ge \eps$, the algorithm outputs REJECT with probability at least $2/3.$ Moreover, \Call{PcondId}{} uses $\tilde{O}(\sqrt{\log N} \cdot \eps^{-2})$ queries.
\end{theorem}

\begin{theorem} \label{IdPcondLower}
    There exists a distribution $\mathcal{D}^*$ with the following property. If any algorithm that, given access to \Call{Samp}{} and \Call{Pcond}{}, outputs ACCEPT with probability at least $2/3$ if $\mathcal{D} = \mathcal{D}^*$ and outputs REJECT with probability at least $2/3$ if $\TV(\mathcal{D}, \mathcal{D}^*) \ge \eps$, then the algorithm must make at least $\Omega\left(\sqrt{\frac{\log N}{\log (\eps^{-1} \cdot \log N)}} \cdot \eps^{-2} \right)$ queries.
\end{theorem}

While there exists a previously known $\tilde{O}(\eps^{-2})$-query algorithm in the COND model \cite{FalahatgarJOPS15}, their algorithm requires COND instead of just PAIRCOND. Our upper and lower bounds improve upon the results of \cite{CanonneRS15}, which provides an $O(\log^4 N \cdot \eps^{-4})$-query algorithm and a $\Omega\left(\sqrt{\frac{\log N}{\log \log N}}\right)$-query lower bound. Importantly, our upper and lower bounds are now tight, up to a $\poly(\log \eps^{-1}, \log \log N)$ multiplicative factor.

\subsection{A Failed Reduction from Tolerant Uniformity to Tolerant Identity Testing}

In the standard sampling model, uniformity and identity testing can both be solved in $\Theta(\eps^{-2} \cdot \sqrt{N})$ samples, and tolerant uniformity and tolerant identity testing can both be solved in $\Theta(\eps^{-2} \cdot \frac{N}{\log N})$ samples. In fact, work by \cite{DiakonikolasK16, Goldreich16} showed a reduction from uniformity testing to identity testing in the standard model: given an algorithm for uniformity testing, one could do identity testing using an asymptotically equivalent number of queries. Unfortunately, their approach does not appear to extend to the conditional sampling model, so we do not get a $\tilde{O}(\eps^{-2})$-query algorithm for tolerant identity testing in COND for free, given such an algorithm for tolerant uniformity testing.

We briefly outline the ideas of \cite{DiakonikolasK16, Goldreich16} and how they fail to extend to the Conditional Sampling model. For simplicity, we assume we are trying to determine whether $\mathcal{D} = \mathcal{D}^*$ or $\TV(\mathcal{D}, \mathcal{D}^*) \ge \eps$, where $\mathcal{D}^*$ is a known distribution such that for all $1 \le x \le N,$ $D^*(x)$ is a multiple of $\frac{1}{N}.$ In this case, the solution is the following. Consider the distribution $\mathcal{P}^*$ as follows. Sample $x \sim \mathcal{D}^*$, and if $D^*(x) = \frac{k}{N},$ we choose $i$ uniformly between $1$ and $k$. Finally, output $(x, i).$ Note that $\mathcal{P}^*$ is uniform on the set $\{(x, i): 1 \le i \le N \cdot D^*(x)\}.$ However, if we modified the distribution so that $x \sim \mathcal{D}$ but $i$ is still randomly chosen in $[k]$ where $D^*(x) = \frac{k}{n}$, one can verify that this modified distribution $\mathcal{P}$ satisfies $\TV(\mathcal{P}^*, \mathcal{P}) = \TV(\mathcal{D}^*, \mathcal{D})$. Moreover, given sample access to $\mathcal{D},$ one can simulate samples from $\mathcal{P}$ because we just need to generate $i$ randomly from $[N \cdot D^*(x)]$, but we know the value of $D^*(x).$ Therefore, identity and tolerant identity testing of $\mathcal{D}$ reduce to uniformity and tolerant uniformity testing of $\mathcal{P}$, respectively.

Unfortunately, this sort of reduction will not work in the conditional sampling model. While we can simulate samples of $\mathcal{P}$ with $\mathcal{D}$, we have no way of simulating conditional samples of $\mathcal{P}$ with $\mathcal{D}$. For instance, if we wished to compare the probabilities of $(x, i)$ and $(x', i')$ in $\mathcal{P}$ using pair conditional samples, we have no way of doing so besides doing pair conditional samples from $\{x, x'\}$. If $D^*(x)$ is much larger than $D^*(x'),$ then to properly simulate this conditional sampling, we will need to draw a very large number of samples until we even get a single sample $x'$.

\subsection{Outline}

We briefly outline the rest of the paper. In Section \ref{Preliminaries}, we go over some definitions and preliminary results. In Section \ref{ProofOutline}, we outline the ideas of our main results. In Section \ref{TolerantUniformity}, we prove Theorem \ref{TolerantUnif}. In Section \ref{TolerantIdentity}, we prove Theorem \ref{TolerantId}. In Section \ref{PcondIdentity}, we prove Theorem \ref{IdPcondUpper}. We defer the proof of Theorem \ref{IdPcondLower} to Appendix \ref{Lower}, as the proof is very similar to \cite[Theorem 8]{CanonneRS15}. We also include pseudocode for the algorithms (divided into subroutines corresponding to lemmas) in Appendix \ref{Pseudocode}. We remark that all our proofs are self contained and do not require looking at the pseudocode, but we include the pseudocode for the potential convenience of the reader.

\section{Preliminaries} \label{Preliminaries}

First, for any integer $N \ge 1,$ we use $[N]$ to denote the set $\{1, 2, \dots, N\},$ and for integers $b \ge a \ge 1,$ we use $[a:b]$ to denote the set $\{a, a+1, \dots, b\}.$ In this paper, we will usually be working with an unknown distribution over the set $[N],$ unless specified otherwise. For a distribution $\mathcal{D}$ over $[N],$ we write $x \sim \mathcal{D}$ to mean that $x$ was drawn from the distribution $\mathcal{D}.$ For any element $i \in [N],$ we let $D(i) := \BP_{x \sim \mathcal{D}} (x = i),$ and for any subset $S \subseteq [N],$ we let $D(S) := \BP_{x \sim \mathcal{D}} (x \in S) = \sum_{i \in S} D(i).$ Likewise, we define $D^*(i) := \BP_{x \sim \mathcal{D}^*} (x = i)$ for a distribution $\mathcal{D}^*,$ and so on. We also let $\mathcal{U}$ denote the uniform distribution over $[N]$.

Recall that the Total Variation Distance $\TV$ between two distributions $\mathcal{D}_1$ and $\mathcal{D}_2$ is defined as
\[\TV(\mathcal{D}_1, \mathcal{D}_2) := \frac{1}{2} \|\mathcal{D}_1-\mathcal{D}_2\|_1 = \frac{1}{2} \sum\limits_{i = 1}^{N} |D_1(i)-D_2(i)|.\]

Regarding definitions, we recall that in the COND model, we are allowed queries to $\Call{Cond}{S}$, which draws from $x \sim \mathcal{D}$ conditional on $x \in S$, and in the PAIRCOND model, we are allowed samples $\Call{Samp}{},$ which draws from $x \sim \mathcal{D}$ and queries to $\Call{Pcond}{x, y} = \Call{Cond}{\{x, y\}},$ where $x, y$ are distinct elements in $[N].$ If dealing with a distribution $\mathcal{Q} \neq \mathcal{D},$ we will write $\Call{Cond}{}_{\mathcal{Q}}(S)$ to denote sampling from the distribution $\mathcal{Q}$, conditioned on being in $S$.

We note a simple result about total variation distance, which is straightforward to verify.

\begin{proposition} \cite{CanonneR14, CanonneSurvey} \label{BasicTV}
    For distributions $\mathcal{D}, \mathcal{D}^*$ on $[N]$, we have that
\[\TV(\mathcal{D},\mathcal{D}^*) = \sum_{i = 1}^{N} \max(0, D(i)-D^*(i)) = \sum_{i = 1}^{N} \max(0, D^*(i)-D(i)).\]
    We also have that
\[1 - \TV(\mathcal{D},\mathcal{D}^*) = \sum_{i = 1}^{N} \min(D(i), D^*(i)).\]
\end{proposition}

We next note the following proposition, proven as a part of \cite[Theorem 14]{CanonneRS15}.

\begin{proposition}  \label{HighEltMass} \cite{CanonneRS15}
    Suppose that $\BP_{x \sim \mathcal{D}}\left[D(x) \ge \frac{1}{\kappa \cdot N}\right] \ge 1-\kappa.$ Then, $\TV(\mathcal{D}, \mathcal{U}) \ge 1-2 \kappa.$
\end{proposition}

We will also be using the Chernoff bound numerous times. We state it formally here.

\begin{theorem} \cite{DubhashiP09}
    Let $X_1, X_2, \dots, X_n$ be independent random variables bounded between $0$ and some value $A$. Let $X = \sum_{i = 1}^{n} X_i$ and let $\mu = \BE[X].$ Then, for any $\delta \le 1,$
\[\BP\left(X \not\in [(1-\delta) \mu, (1+\delta) \mu]\right) \le 2 \exp\left(-\frac{\delta^2 \mu}{3 A}\right).\]
\end{theorem}

Finally, we note a simple primitive algorithm called \Call{Compare}{}, used in \cite{CanonneRS15}, that will be very useful for us. The algorithm explains how to get a good approximation to the ratio of $D(x)$ and $D(y)$ using \Call{Pcond}{}. We slightly restate the guarantees of \Call{Compare}{} for our convenience (see \cite[Lemma 2]{CanonneRS15} for comparison).

\begin{proposition} \label{Compare}
    For any $\gamma < 1$ and elements $x, y \in [N],$ there is an algorithm \Call{Compare}{$x, y, \gamma$} that uses $O(\gamma^{-2} \log \eps^{-1})$ calls to \Call{Pcond}{} and returns $\alpha$ (which may equal $\infty$) such that
\begin{enumerate}
    \item if $\frac{D(x)}{D(y)} \le 20$, then with probability at least $1-\eps^{10}$, $\left|\alpha - \frac{D(x)}{D(y)}\right| \le \gamma$.
    \item if $\frac{D(x)}{D(y)} \ge \frac{1}{20},$ then with probability at least $1-\eps^{10}$, $\left|\frac{1}{\alpha} - \frac{D(y)}{D(x)}\right| \le \gamma$ (where $1/\infty = 0$).
\end{enumerate}
\end{proposition}

As a general note, all of our success probabilities in the above result and future results will be of the form $1 - \eps^C$ for some constant $C$. This is because it is easy to get these success probabilities using $O(\log \eps^{-1})$ repetitions, and so that we can union bound over $\poly(\eps^{-1})$ potential failures and still have a high success probability. These exponents may often be larger than necessary (for instance, $1-\eps^5$ in Proposition \ref{Compare} is likely sufficient), though we keep them high just to be safe.


\section{Proof Overview} \label{ProofOutline}

In this section, we provide the general proof outlines for Theorems \ref{TolerantUnif}, \ref{TolerantId}, \ref{TestMonotone}, and \ref{IdPcondUpper}. We do not give an outline of Theorem \ref{IdPcondLower} here, due to its similarity to \cite[Theorem 8]{CanonneRS15}. However, we explain the distribution $\mathcal{D}^*$ used in Theorem \ref{IdPcondLower} and give intuition for the proof in Subsection \ref{LowerOverview}.

\subsection{Overview of Theorem \ref{TolerantUnif}}

Canonne et al. \cite{CanonneRS15} provided an $\tilde{O}(\eps^{-20})$-query algorithm for estimating distance to uniformity. Their algorithm can be broken into two steps. The first step is to find a pair $(x, \hat{D}(x))$ for $x \in [N]$ such that $\hat{D}(x) = (1 \pm \eps) \cdot D(x)$ and $D(x) \in \left[\frac{\eps}{N}, \frac{\eps^{-1}}{N}\right]$ - they use $\tilde{O}(\eps^{-20})$ queries to achieve this. The second step is to use $x$ and the estimate for $D(x)$ to estimate $D(y)$ for a randomly sampled $y$, using \Call{Pcond}{}. However, they need to sample up to $O(\eps^{-2})$ elements $y_1, \dots, y_{O(1/\eps^2)}$ from the distribution, and for each sample $y_i$, they obtain a $1 \pm \eps$ multiplicative estimate for the ratio $\frac{D(x)}{D(y_i)}$, so that they can estimate $D(y_i)$ properly. By ignoring $y$ such that $\frac{D(x)}{D(y)}$ is too small or too large, they show that only $O(\eps^{-6})$ queries are needed for the second step in the worst case.

We therefore need improved algorithms for both steps. To do this, we first deal with the case that all elements $i \in [N]$ satisfy $D(i) \in \left[\frac{1}{2N}, \frac{2}{N}\right].$ For the first step, given some $x$, we first create an unbiased estimator for $\frac{D(y)}{D(x)}$ using an average of $O(1)$ queries to \Call{Pcond}{}. The idea is a ``Geometric Distribution trick''. If we call \Call{Pcond}{$x, y$} many times until we see $x$, the expected number of times we saw $y$ before seeing $x$ equals $\frac{D(y)}{D(x)}$: this can be easily proven using geometric random variables. If $y$ were sampled uniformly from $[N]$, then this is actually an unbiased estimator for $\frac{1}{N \cdot D(x)},$ since the average value of $D(y)$ is $1/N$. We will repeat this procedure $O(\eps^{-2})$ times with a different, randomly chosen $y$ each time to estimate $D(x)$ up to a $1 \pm \eps$ multiplicative factor.

We now want to estimate $\TV(\mathcal{D}, \mathcal{U})$ given a very good estimate for a single item $x$. The natural intuition is to pick $\eps^{-2}$ samples $i$ and estimate $D(i)$ for each $i$, and use these estimates to approximate the total variation distance. However, each $i$ would need $\eps^{-2}$ queries to estimate, for a total query complexity of $O(\eps^{-4})$. To improve this to $\tilde{O}(\eps^{-2}),$ we again try another geometric distribution trick to give an unbiased estimator of $\frac{D(i)}{D(x)}$, where $D(x)$ is already approximately known. However, we need some way of getting an unbiased estimator of $|D(i)-1/N|$ rather than just $D(i)$, since $\TV(\mathcal{D}, \mathcal{U})$ equals $\frac{1}{2} \cdot \sum_i |D(i)-1/N|.$ It suffices to find an unbiased estimator of $D(i)$ that is always on the same side of $1/N$ as $D(i)$. We note that if $D(i) \approx \frac{1}{N} \cdot (1 \pm \delta),$ then we'll need $\delta^{-2}$ queries to $\Call{Pcond}{}$ to know if $D(i) > 1/N$ or $D(i) < 1/N.$ So, it may seem that if $\delta$ is often close to $\eps,$ we'll need $\eps^{-2}$ queries to \Call{Pcond}{} for each of $\eps^{-2}$ samples, meaning an $O(\eps^{-4})$ query complexity again. But if $\delta$ were $O(\eps)$ for all $i$, we actually only need to sample $O(1)$ different $i$'s, since $D(i)$ is already known to be in the tight range $[1 \pm O(\eps)]/N$ for all $i$. Generally, if $D(i) \approx (1 \pm \delta)/N,$ we'll need $\tilde{O}(\delta^{-2})$ queries to $\Call{Pcond}{x, i}$ but only $O(\delta^2/\eps^2)$ samples of such $i$. We look at $\delta = 2^{-j}$ for each $j = 1, 2, \dots, \log \eps^{-1},$ and in total we only use $\tilde{O}(\eps^{-2})$ queries.

One issue, however, is that we are not guaranteed that $D(i) \in \left[\frac{1}{2N}, \frac{2}{N}\right]$ for all $i$. We briefly explain the ideas needed to fix this issue. We create an oracle which essentially only accepts elements $i$ with $D(i)$ close to $\frac{1}{N}.$ The general approach is to sample random elements, and for any random element $x$ picked, we consider the probability that $D(x) \approx D(y)$ if $y$ is drawn from the uniform distribution, and the probability that $D(x) \approx D(z)$ if $z$ is drawn from $\mathcal{D}.$ Checking if $D(x) \approx D(y)$ or $D(x) \approx D(z)$ is easy with \Call{Pcond}{}. Note that if $D(x)$ is much larger than $1/N,$ then if $D(x') \approx D(x),$ $x'$ is much more likely to be drawn from $\mathcal{D}$ than from $\mathcal{U}.$ Likewise, if $D(x)$ is much smaller than $1/N$, then $x'$ is much more likely to be drawn from $\mathcal{U}$ than from $\mathcal{D}.$ So, we will choose some value $x$ if these two probabilities are equal. However, we only need to do this until we have found a single $x$ with $D(x) \approx 1/N$. The oracle can then just accept some $i$ if $\frac{D(x)}{D(i)} = \Theta(1),$ which is easily checkable with \Call{Pcond}{}. We note that there may not always be a good value of $x$ (for instance if $D(x)$ is always outside the range $[1/N, 2/N]$) but we will see when proving this theorem that these extreme cases are not too difficult to fix. We also remark that in this case, we are already at an $\Omega(1)$ distance from uniform.

\subsection{Overview of Theorem \ref{TolerantId}}

For identity testing in the PAIRCOND model, there must be some dependence on $N$. The rough reason for why is if we are trying to determine if $\mathcal{D} = \mathcal{D}^*,$ some values $\mathcal{D}^*(i)$ can be much bigger than other values $\mathcal{D}^*(j)$, so there are groups of elements which can never be compared to each other. So ``now, you will experience the full power of the COND model.''

As done in \cite{CanonneRS15, FalahatgarJOPS15}, we assume WLOG that $D^*(1) \le D^*(2) \le \cdots \le D^*(N)$. Since $\mathcal{D}^*$ is known, we can permute the elements accordingly. Our goal will be to determine
\[\TV(\mathcal{D}, \mathcal{D}^*) = 1 - \sum_{i = 1}^{N} \min(D(i), D^*(i)) = 1 - \BE_{i \sim \mathcal{D}^*} \left[\min\left(1, \frac{D(i)}{D^*(i)}\right)\right].\]
    To actually make use of this observation, we choose some $i \leftarrow \mathcal{D}^*$ and attempt to calculate $\frac{D(i)}{D^*(i)}$. If we do this for several samples of $i$, we can get a good approximation for $\TV(\mathcal{D}, \mathcal{D}^*)$
    
    The first step for approximating $\frac{D(i)}{D^*(i)}$ is to note a crucial idea used for (non-tolerant) identity testing in the COND model \cite{CanonneRS15, FalahatgarJOPS15}: rather than compare $D(i)$ with $D(j)$ when $D^*(i)$ and $D^*(j)$ may be vastly different, we compare $D(i)$ with $D(S)$ for some subset $S$ with $D^*(S) \approx D^*(i)$. We can partition $[i-1]$ into sets $S_1, \dots, S_{k-1}$ so that $0.5 \cdot D^*(i) \le D^*(S_i) \le D^*(i),$ unless $D^*([i-1]) \le 0.5 \cdot D^*(i).$ Ignoring the latter case, we can approximate $\frac{D(i)}{D([i])}$ simulating conditional samples from the set $\{S_1, S_2, \dots, S_{k-1}, S_k\}$ (where $S_k = \{i\}$) using the ideas from our tolerant uniformity algorithm. Then, we approximate $D([i])$ by sampling from $S$, so intuitively, we already have a close multiplicative approximation to $D(i)$ (and thus $\frac{D(i)}{D^*(i)}$ since $D^*(i)$ is known).
    
    
    However, there are actually a few problems with this approach. We'll just focus on the biggest issue, which is that our algorithm for approximating $\frac{D(i)}{D([i])}$ only works if $\frac{D(i)}{D([i])}$ is sufficiently close to the expectation $\frac{D^*(i)}{D^*([i])}$. Note that in non-tolerant identity testing, this is not an issue, since we'll just find out that $\frac{D(i)}{D([i])} \neq \frac{D^*(i)}{D([i])}$ so $\mathcal{D} \neq \mathcal{D}^*$. Even in most scenarios for the tolerant case, this isn't a problem, as it would imply $\frac{D(i)}{D^*(i)}$ is either very large or very small. But we may run into an issue if $\frac{D(i)}{D([i])}$ is much larger than expected, but $D([i])$ is much smaller than expected, causing $D(i)$ to actually be close to $D^*(i)$. In this case, since $1 - \TV(\mathcal{D}, \mathcal{D}^*) = \sum_{j = 1}^{N} \min(D^*(j), D(j)),$ if $D([i]) = \sum_{j = 1}^{N} D(j)$ is very small, we can hopefully ignore the first $i$ terms and compute $\sum_{j = i+1}^{N} \min(D^*(j), D(j))$ instead.
    
    
    
    In the case where we have removed some of the elements, and just have to approximate $\sum_{i \in S} \min(D(i), D^*(i))$ for some $S \subset [N],$ we condition on the samples coming from $S$, but this may cause the values $D^*(i)$ to scale differently from $D(i),$ since $D(S)$ and $D^*(S)$ may differ by much more than a $1 \pm \eps$ factor. As a result, we will actually attempt to solve a slightly more general question of estimating $\BE_{i \sim \mathcal{P}^*} \min(c_1 \frac{P(i)}{P^*(i)}, c_2) = \sum \min(c_1 P(i), c_2 P^*(i))$, for some constants $\eps \le c_1, c_2 \le 1,$ where $\mathcal{P}, \mathcal{P}^*$ are the conditional distributions of $\mathcal{D}, \mathcal{D}^*$ conditioned on $i \in S.$ This will end up being quite similar to estimating $\sum \min(D(i), D^*(i))$: we still sample $i$ from $\mathcal{P}^*$ and approximate $\frac{P(i)}{P^*(i)},$ though there will be some additional technical details to tackle.
    
\subsection{Overview of Theorem \ref{TestMonotone}}

    Our monotonicity test follows the same structure as Canonne's monotonicity test in COND \cite{Canonne15}. We first describe the idea behind Canonne's algorithm, and then explain our improvements.
    
    Canonne's algorithm crucially depends on a result known as the Birg\'{e} decomposition theorem, which is commonly used for monotonicity testing and related problems. The idea is that if $\mathcal{D}$ is monotone, i.e., $D(1) \ge D(2) \ge \cdots \ge D(N),$ we can partition the data $[N]$ into consecutive intervals $I_1, I_2, \dots, I_\ell$, with $I_k$ of size approximately $(1+\eps)^k$. (In distribution testing, this type of partitioning is often called \emph{bucketing}.) If $\mathcal{D}$ is monotone, the Birg\'{e} decomposition theorem tells us that if $\mathcal{D}'$ is the ``flattened'' distribution where the probabilities $D(i)$ have been replaced with $\frac{D(I_k)}{|I_k|},$ then $\TV(\mathcal{D}, \mathcal{D'}) \le \eps.$ So, the first step is to approximate this distance. However, for any fixed $k$, we can determine the distance between $\mathcal{D}$ restricted to $I_k$ and the uniform distribution on $I_k$ using $O(\eps^{-2})$ samples. From this, Canonne notes that if we sample $O(\eps^{-2})$ intervals (where we can sample an interval by drawing $i \sim \mathcal{D}$ and finding the interval $I_k \ni i$) and averaging their total variation distances from uniform, we obtain an $\eps$-approximation to $\TV(\mathcal{D}, \mathcal{D}').$ If this value is much more than $\eps$, then we know $\mathcal{D}$ is not monotone.
    
    But even if $\TV(\mathcal{D}, \mathcal{D}')$ is small, it is still possible that $\mathcal{D}$ is far from monotone. But this implies that $\mathcal{D}'$ is far from monotone. The second step of Canonne's algorithm is to test whether this is the case. (One can verify that if $\mathcal{D}$ is monotone, then so is $\mathcal{D}'$). However, we'll instead look at a slightly different distribution $\mathcal{Q}$ called the \emph{reduced distribution}, where we say the probability of sampling $k$ is $D(I_k).$ Equivalently, we are combining all terms in a bucket $I_k$ into a single element $k$. Note that conditional samples from this distribution can be easily simulated. Moreover, if $\mathcal{D}'$ is monotone, $\mathcal{Q}$ now satisfies what's called the \emph{exponential property}, meaning that $Q(i) \le (1+\eps) Q(i-1)$ for all $i$. But if $\mathcal{D}'$ is far from monotone, $\mathcal{Q}$ is far from satisfying the exponential property.
    
    Thus, for the second step, the main task is an algorithm for testing whether a distribution has the exponential property or is far from it. Canonne notes a very natural algorithm of comparing $\frac{Q(i)}{Q(i-1)}$ with queries to $\Call{Cond}{\{i-1, i\}}$: the difficult part is proving such an algorithm works. Canonne proves that if $\mathcal{Q}$ is $\eps$-far from having the exponential property, then with probability at least $\eps^2$ for a randomly chosen $i$ sampled from $\mathcal{Q},$ that $\frac{Q(i)}{Q(i-1)} > 1 + \eps + \Theta(\eps^3)$. But if $\mathcal{Q}$ had the exponential property, we would always have that $\frac{Q(i)}{Q(i-1)} \le 1 + \eps.$ We can distinguish between these two cases using $\tilde{O}(\eps^{-6})$ queries using $\Call{Compare}{i, i-1, \Theta(\eps^3)}$, and repeating this for $\eps^{-2}$ random samples $i$ gives us an algorithm only needing $\tilde{O}(\eps^{-8})$ queries.
    
    For the first step, we now only need $\tilde{O}(\eps^{-4})$ queries to \Call{Cond}{}, since we can do tolerant uniformity in $\tilde{O}(\eps^{-2})$ queries as opposed to $\tilde{O}(\eps^{-20})$ queries. For the second step, we need to prove a stronger technical lemma. Specifically, we show that there exist some $\alpha, \beta$ such that $\BP\left(\frac{Q(i)}{Q(i-1)} > 1 + \eps + \alpha\right) \ge \beta,$ and $\alpha \cdot \beta \ge \tilde{\Omega}(\eps^2).$ Given such an $\alpha, \beta,$ we can find some $i$ with $\frac{Q(i)}{Q(i-1)} > 1 + \eps,$ by sampling $O(\beta^{-1})$ random $i$, and for each one, using \Call{Compare}{$i, i-1, \frac{\alpha}{10}$} to determine $\frac{Q(i)}{Q(i-1)}$ up to an $\frac{\alpha}{10}$ error. With high probability, we'll find at least one violating $i$. The query complexity is $O(\beta^{-1})$ per $i$ and $\tilde{O}(\alpha^{-2})$ (since \Call{Compare}{} requires this many queries). But since $\beta^{-1} \cdot \alpha^{-2} \le (\beta \cdot \alpha)^{-2} = \tilde{O}(\eps^{-4}),$ the overall algorithm just requires $\tilde{O}(\eps^{-4})$ queries.
    
    To briefly describe the main lemma's proof, the argument constructs a distribution $\mathcal{Q}'$ satisfying the exponential property, but is very close to $\mathcal{Q}$ if we don't have the desired property. The construction of $\mathcal{Q}'$ works as follows: we just find the smallest ``distribution'' $\mathcal{Q}' = \{q_1', q_2', \cdots\}$ that satisfies the exponential property such that each $q_i'$ is at least $q_i$: this is not a real distribution since $\sum q_i' > 1.$ However, we show that if we do not satisfy the property $\BP\left(\frac{Q(i)}{Q(i-1)} > 1 + \eps + \alpha\right) \ge \beta$ for any $\alpha \cdot \beta \ge \tilde{\Omega}(\eps^2),$ then we show that $\|\mathcal{Q}-\mathcal{Q}'\|_1 \lesssim \eps$ using some simple telescoping arguments. But since $q_i \le q_i'$ for all $i$, one can easily show that $\|\mathcal{Q}-\mathcal{Q}'\|_1 = \|\mathcal{Q}''-\mathcal{Q}'\|_1,$ where $\mathcal{Q}''$ is the normalized version of $\mathcal{Q}'$ so that $\sum q_i'' = 1.$ Then, $\mathcal{Q}''$ satisfies the exponential property, whereas $\TV(\mathcal{Q}, \mathcal{Q}'') = o(\eps)$, so $\mathcal{Q}$ is \emph{close} to satisfying the exponential property.
    
\subsection{Overview of Theorem \ref{IdPcondUpper}}

    A crucial observation motivating this proof, also noted in \cite[Theorem 6]{FalahatgarJOPS15}, is that if $\mathcal{D} \neq \mathcal{D}^*,$ then if we select $i$ from $\mathcal{D}$ and $j$ from $\mathcal{U},$ we should expect $\frac{D(i)}{D(j)}$ to be larger than $\frac{D^*(i)}{D^*(j)}$. Intuitively, this is true since drawing $i \leftarrow \mathcal{D}$ is biased in favor of elements with high values of $D(i)$. We formalize this by proving that if all expected probabilities $D^*(i)$ are in the range $\left[\frac{1}{2N}, \frac{2}{N}\right]$, and if we draw $i \leftarrow \mathcal{D}$ and $j \leftarrow \mathcal{U}$, then $\BE\left|\frac{D(i)}{D(i)+D(j)} - \frac{D^*(i)}{D^*(i)+D^*(j)}\right| = \Omega(\eps)$ assuming that $\TV(\mathcal{D}, \mathcal{D}^*) \ge \eps.$
    From here, we sample several pairs $(i, j)$ with $i \leftarrow \mathcal{D}, j \leftarrow \mathcal{U}$ and use \Call{Pcond}{$i, j$} to determine if $\frac{D(i)}{D(i)+D(j)}$ and $\frac{D^*(i)}{D^*(i)+D^*(j)}$ differ substantially. By considering various levels of how accurately we determine $\frac{D(i)}{D(i)+D(j)}$, we show that we only need $\tilde{O}(\eps^{-2})$ total calls to \Call{Pcond}{} over all pairs $(i, j)$.
    
    However, to generalize to arbitrary distributions $\mathcal{D}$ and $\mathcal{D}^*$ using only the power of PAIRCOND, we need a different approach from \cite{FalahatgarJOPS15}. Notably, \cite{FalahatgarJOPS15} was able to utilize the COND model by comparing $D(i)$ with $D(S)$ for some set $S$ with $D^*(S) \approx D^*(i)$, as noted in the overview of Theorem \ref{TolerantId}, but in the PAIRCOND model, we cannot do this. Instead, we bucket $[N]$ into sets $S_1, S_2, \dots, S_{O(\log N)}$ where elements $i$ in $S_k$ have $D^*(i) \approx 2^{-k}$. We similarly can show that if we draw $i$ from $\mathcal{D}$, and then draw $j$ uniformly from the set $S_k$ containing $i$, then one of the following two things occurs. Either $\BE \left|\frac{D(i)}{D(i)+D(j)} - \frac{D^*(i)}{D^*(i)+D^*(j)}\right| = \Omega(\eps)$, or $\TV(\mathcal{S}, \mathcal{S}^*) \ge \eps$, where the distributions $\mathcal{S}, \mathcal{S}^*$ over $[O(\log N)]$ are defined such that $\BP(k \leftarrow \mathcal{S}) = D(S_k)$ and $\BP(k \leftarrow \mathcal{S}^*) = D^*(S_k)$. We verify the first possibility using the argument in the previous paragraph, with $\tilde{O}(\eps^{-2})$ queries. We verify the second possibility using sampling from $\mathcal{S},$ which can be simulated by sampling from $\mathcal{D},$ which is doable with $O\left(\frac{\sqrt{\log N}}{\eps^2}\right)$ queries by \cite{ValiantV17}, since the support size of $\mathcal{S}$ is $O(\log N)$.
    
\section{An $\tilde{O}(\eps^{-2})$-Query Algorithm for Tolerant Uniformity Testing} \label{TolerantUniformity}

In this section, we present an algorithm that makes $\tilde{O}(\eps^{-2})$ queries to \Call{Samp}{} and \Call{Pcond}{} and determines $\TV(\mathcal{D}, \mathcal{U})$ up to an $O(\eps)$ additive error. This algorithm is known to be optimal even in the stronger COND model and even for the standard uniformity testing problem.

The algorithm can be broadly divided into three steps. The first step of the proof, done in subsection \ref{OracleGiven}, determines the distance from uniformity when all probabilities are known to be close to $\frac{1}{N}.$ To get this to work for general distributions, we will assume an oracle which discards elements with probabilities too far away from $\frac{1}{N}.$ In subsection \ref{OracleCreate}, we show how to generate the oracle, which roughly works by first finding a single element $x$ and an approximation $\hat{D}(x) \approx D(x)$, which will not be a $1+\eps$ approximation but will be an $O(1)$ approximation. We then combine the two steps together and finish the proof in subsection \ref{FinishAlgorithm}.

\subsection{An Algorithm with Access to an Oracle} \label{OracleGiven}

In this subsection, our goal is to determine $\sum_{i = 1}^{N} |D(i) - \frac{1}{N}|$ but only among the $D(i)$'s which are within a constant factor of $\frac{1}{N}.$ To make this usable for the general algorithm, we assume we have an oracle that roughly accepts no elements that are not close to $\frac{1}{N}$ but accepts elements that are very close to $\frac{1}{N}$ with some probability. In the case where $D(i)$ is within a constant factor of $\frac{1}{N}$ for all $i$, this subsection immediately implies an algorithm for estimating Total Variation Distance from uniform.

First, we show how to get a good approximation for the probability of a single element $x$.

\begin{lemma} \label{SingleElement}
    Let $s: [N] \to [0, 1]$ be an (unknown) function so that $\sum_{i = 1}^{N} \frac{s(i)}{N} = \gamma_1 \ge 10 \eps$, and $s(i) = 0$ whenever $D(i) \not\in \left[\frac{1}{5N}, \frac{5}{N}\right]$, where $\gamma_1, \gamma_2$ are also unknown. Then, given an oracle $\mathcal{O}$ which accepts an element $i$ with probability $s(i)$, as well as an element $x$ with probability between $\frac{1}{2N}$ and $\frac{2}{N}$, there is an algorithm \Call{SingleElement}{} that can determine $D(x)$ up to a $1 \pm c \frac{\eps}{\gamma_1}$ multiplicative factor with probability at least $0.98$ for some small constant $c$, using an expected $O(\eps^{-2})$ queries to \Call{Pcond}{} and $\mathcal{O}$.
\end{lemma}

\begin{proof}
    First, let $K = O(\eps^{-2})$ and sample elements $y_1, \dots, y_K \leftarrow \mathcal{U}$ and $z_1, \dots, z_K \leftarrow \mathcal{D}.$ Note that for $y \leftarrow \mathcal{U},$ the probability of $\mathcal{O}(y)$ accepting is $\sum_{i = 1}^{N} \frac{s(i)}{N} = \gamma_1$, and for $z \leftarrow \mathcal{D},$ the probability of $\mathcal{O}(z)$ accepting is $\gamma_2 := \sum_{i = 1}^{N} s(i) \cdot D(i).$ Thus, by checking if the oracle $\mathcal{O}$ accepts the $y_i$'s and $z_i$'s for each $1 \le i \le K$, we can determine $\tilde{\gamma}_1 \approx \gamma_1$ and $\tilde{\gamma}_2 \approx \gamma_2$, which are accurate up to a $c_1 \cdot \eps$ additive error with probability at least $0.99$ for some small constant $c_1$. Also, recall that $\gamma_1 \ge 10 \eps,$ and as $s(i)$ is nonzero only when $D(i) \in \left[\frac{1}{5N}, \frac{5}{N}\right],$ we have that $\gamma_2 \in \left[\frac{\gamma_1}{5}, 5 \gamma_1\right]$, so $\gamma_2 \ge 2 \eps.$ Therefore, $\tilde{\gamma}_1 \approx \gamma_1$ and $\tilde{\gamma}_2 \approx \gamma_2$ up to a multiplicative factor of $2$ as well.
    
    Now, we attempt to determine $D(x).$ To do so, for each $y_i$ accepted by $\mathcal{O},$ we run $\Call{Pcond}{x, y_i}$, which outputs $y_i$ with probability $\frac{D(y_i)}{D(y_i)+D(x)}$. This probability will be in the range $[1/11, 10/11]$, as $D(y_i)\in \left[\frac{1}{5N}, \frac{5}{N}\right]$ and $D(x) \in \left[\frac{1}{2N}, \frac{2}{N}\right].$ Keep running $\Call{Pcond}{x, y_i}$ until it outputs $x$ and consider the number of times $y_i$ is returned. This is some random variable $V$, which, conditioned on $y_i$, is a Geometric Random variable with mean $\sum_{j = 1}^{\infty} \left(\frac{D(y_i)}{D(y_i)+D(x)}\right)^j = \frac{D(y_i)}{D(x)}$ and variance $O(1).$ Moreover, in expectation we only call $\Call{Pcond}{}$ $O(1)$ times for each $i$.
    
    Next, let $A$ be the indicator random variable of $\mathcal{O}$ accepting $y$. In this case, we output the random variable $V$ computed above, and we know that $\BE\left[V|y, A = 1\right] = \frac{D(y)}{D(x)}$ and $Var(V|y, A = 1) = O(1)$ for all $y$. Otherwise, if $A = 0$, we output $V = 0$. Therefore, $V$ has mean $\frac{1}{D(x)} \cdot \sum_{y = 1}^{N} \frac{1}{N} s(y) \cdot D(y) = \frac{\gamma_2}{N \cdot D(x)}$ (since $y$ is chosen from $\mathcal{U}$) and variance
\[Var(V) = \BE[Var(V|y, A)] + Var(\BE[V|y, A]) = O(1).\]
    The first equality in the above equation follows from the Law of Total Variance. The second inequality is true since if $A = 0,$ then $\BE[V|y, A = 0] = Var(V|y, A = 0) = 0,$ and if $A = 1,$ then as $\frac{D(y)}{D(y)+D(x)} \in \left[\frac{1}{11}, \frac{10}{11}\right]$ for all $y$ accepted by $\mathcal{O},$ $\BE[V|y, A = 1], Var(V|y, A = 1) = O(1).$ By averaging over $K$ samples $V_1, \dots, V_K$, where each $V_i$ depends on $y_i \leftarrow \mathcal{U},$ we get a random variable with mean $\frac{\gamma_2}{N \cdot D(x)}$ and standard deviation $O(\eps)$, so we get some number $\gamma_3$ which equals $\frac{\gamma_2}{N \cdot D(x)} \pm c_2 \eps$ with probability at least $0.98$ for some small constant $c_2$.
    
    Noting that $N \cdot D(x) \in [1/2, 2]$ and $\gamma_2 \ge 2 \eps$, we have that with probability at least $0.98$,
\[\frac{\tilde{\gamma}_2}{\gamma_3} = \frac{\gamma_2 \pm c_1 \eps}{\frac{\gamma_2}{N \cdot D(x)} \pm c_2 \eps} = N \cdot D(x) \cdot \left(1 \pm \frac{c}{5} \cdot \frac{\eps}{\gamma_2}\right).\]
    Finally, $\gamma_2 \ge \frac{1}{5} \cdot \gamma_1,$ so we thus get a $1 \pm c \frac{\eps}{\gamma_1}$ multiplicative approximation to $D(x)$.
\end{proof}

Next, we give an algorithm that estimates $D(z)$ for a general element $z$, and also approximately determines whether $D(z) > D(x)$ or $D(z) < D(x)$.

\begin{lemma} \label{ZEstimate}
    Let $z$ be in $[N]$ such that we know $D(z) \in \left[\frac{1}{5N}, \frac{5}{N}\right].$ Suppose we also know some element $x$ with $D(x) \in \left[\frac{1}{2N}, \frac{2}{N}\right],$ along with an estimate $\tilde{D}(x)$ such that $\frac{\tilde{D}(x)}{D(x)} \in [1-\beta, 1+\beta]$ for some known $\eps \le \beta \le \frac{1}{1600}.$ Then, there is an algorithm \Call{ZEstimate}{} that gives an estimate $\tilde{D}(z)$ of $D(z)$ with the following properties:
\begin{enumerate}
    \item If $|N \cdot D(z) - 1| \ge 160 \cdot \beta,$ then with probability at least $1-\eps^9$, $\frac{N \cdot \tilde{D}(z) - 1}{N \cdot D(z) - 1} \in \left[\frac{1}{3}, 2\right]$.
    \item If $|N \cdot D(z) - 1| \le 160 \cdot \beta,$ then with probability at least $1-\eps^9$, $|N \cdot \tilde{D}(z) - 1| \le 320 \cdot \beta$.
    \item The algorithm uses $O\left(\min\left((N \cdot D(z) - 1)^{-2}, \beta^{-2}\right) \cdot \log \eps^{-1}\right)$ queries to \Call{Pcond}{}.
\end{enumerate}
\end{lemma}

\begin{remark}
    By setting $c = \frac{1}{160}$ in Lemma \ref{SingleElement}, we have that $\beta \le \frac{1}{1600}$ is an acceptable assumption.
\end{remark}

\begin{proof}
    Fix some integer $i \ge 1$ with $2^{-i} \ge 40 \cdot \beta$, and suppose we run $\Call{Compare}{z, x, 2^{-i}/20}$ to get $\alpha$, which is a $2^{-i}/20$ additive approximation to $\frac{D(z)}{D(x)}.$ Moreover, we know $\tilde{D}(x)$ and we know that $\frac{\tilde{D}(x)}{D(x)} \in [1-\beta, 1+\beta],$ so 
\[\alpha \cdot \tilde{D}(x) = \left(\frac{D(z)}{D(x)} \pm \frac{2^{-i}}{20}\right) \cdot D(x) \cdot \frac{\tilde{D}(x)}{D(x)} = \left(D(z) \pm \frac{2^{-i}}{20} \cdot D(x)\right) \cdot (1 \pm \beta).\]
    Now, noting that $\frac{2^{-i}}{20} \cdot D(x) \le 2^{-i} \cdot \frac{1}{10 N}$ and that $\beta \le 0.025 \cdot 2^{-i},$ we have that $|\alpha \cdot \tilde{D}(x) - D(z)| \le 2^{-i} \cdot \frac{1.025}{10N} + 0.025 \cdot 2^{-i} \cdot \frac{5}{N} \le \frac{2^{-(i+2)}}{N}.$ For $D(z) < \frac{2}{N},$ we can improve this to $2^{-i} \cdot \frac{1.025}{10N} + 0.025 \cdot \frac{2}{N} \le \frac{2^{-i} \cdot (1/6)}{N}.$
    
    Therefore, our algorithm will work as follows. First, we set $i = 1.$ For each $i$, we will run $\Call{Compare}{z, x, 2^{-i}/20}$ and multiply the returned value by $\tilde{D}(x)$ to get an approximation to $D(z).$ If our approximation is not in the range $\left[\frac{1}{N} \cdot (1-2^{-i}), \frac{1}{N} \cdot (1 + 2^{-i})\right],$ we output the approximation as our estimate $\tilde{D}(z).$ Otherwise, we increment $i$, until we reach $2^{-i} < 40 \beta.$ In this case, we simply output $\tilde{D}(z) = \frac{1}{N}.$
    
    Note that for $D(z) > \frac{2}{N},$ with probability at least $1 - \eps^{10},$ we will just output $\tilde{D}(z)$ on the iteration $i = 1$ and that $|\tilde{D}(z)-D(z)| \le \frac{1}{8N},$ which means that in fact $\frac{N \cdot \tilde{D}(z) - 1}{N \cdot D(z) - 1} \in \left[\frac{7}{8}, \frac{9}{8}\right].$ The total number of calls to \Call{Pcond}{} in this case is just $O(\log \eps^{-1}).$ Otherwise, if $80 \beta \le 2^{-(j+1)} \le |N \cdot D(z) - 1| \le 2^{-j},$ with probability at least $1-\eps^9,$ for all $i < j-1$ we will not output an estimate $\tilde{D}(z)$ and we will output an estimate $\tilde{D}(z)$ at least by the time of reaching $i = j+2.$ Note that $40 \beta \le 2^{-(j+2)}$. The estimate will be off by at most $\frac{2^{-(j-1)}}{6 N}$ but $|N \cdot D(z) - 1| \ge 2^{-(j+1)} = \frac{2^{-(j-1)}}{4},$ so $\frac{N \cdot \tilde{D}(z)-1}{N \cdot D(z)-1} \in \left[\frac{1}{3}, \frac{5}{3}\right].$ The number of queries in this case is $O(2^{2j} \cdot \log \eps^{-1}).$ Otherwise, we must have that $|N \cdot D(z)-1| \le 160 \beta,$ so we will either never output an estimate until the end (so $\tilde{D}(z) = \frac{1}{N}$), or we output $\tilde{D}(z)$ at some step $i$, so $2^{-(i+1)} < |N \cdot D(z)-1|$. In this case, $|\tilde{D}(z)-D(z)| \le \frac{2^{-i}}{6 N} \le \frac{|N \cdot D(z)-1|}{3 N},$ so $|N \cdot \tilde{D}(z)-1| \le 2 \cdot |N \cdot D(z)-1| \le 320 \beta.$ The number of queries in this case is $O(\beta^{-2} \log \eps^{-1}).$
\end{proof}

Finally, we give an estimate for $\TV(\mathcal{D}, \mathcal{U}).$ In reality, we estimate $\sum |D(i) - \frac{1}{N}|$ just over elements accepted by the oracle (weighted by the probability of being accepted by the oracle).

\begin{lemma} \label{EstimateCloseTerms}
    Let $x, s, \mathcal{O}, \gamma_1, \gamma_2$ be as in Lemma \ref{SingleElement}. Then, using an expected $O(\eps^{-2} \log^{2} \eps^{-1})$ queries to $\Call{Pcond}{}$ and $\mathcal{O},$ there is an algorithm \Call{EstimateCloseTerms}{} that can determine $\sum_{i = 1}^{N} s(i) \cdot |D(i) - \frac{1}{N}|$ up to an additive $O(\eps)$ error with probability at least $0.97$.
\end{lemma}

\begin{proof}
    Let $\tilde{D}(x)$ be our estimate of $D(x)$ based on Lemma \ref{SingleElement}, where we know that $D(x) \in \left[\frac{1}{2N}, \frac{2}{N}\right]$ and $\frac{\tilde{D}(x)}{D(x)} = 1 + O(\eps/\gamma_1)$ (with probability at least $0.98$). For any $z \sim \mathcal{U}$ with $D(z) \in \left[\frac{1}{5N}, \frac{5}{N}\right],$ let $\tilde{D}(z)$ be our guess for $D(z)$ based on Algorithm \ref{ZEstimate} with $\beta = \frac{\eps}{160 \gamma_1}$. Note that $\tilde{D}(z)$ is a random variable even if we fix $z$. However, for any $z$ with $|N \cdot D(z) - 1| > \frac{\eps}{\gamma_1},$ with $1-\eps^{9}$ probability, $\frac{N \cdot \tilde{D}(z) - 1}{N \cdot D(z)-1}$ is in the range $\left[\frac{1}{3}, 2\right].$

    Now, we choose $T$ so that $2^{-T} = \Theta(\eps/\gamma_1)$. While we don't know $\gamma_1,$ we know $\tilde{\gamma}_1 = \Theta(\gamma_1),$ so we can choose $T$. For each $0 \le t \le T-1,$ let $\BI_{+, t}(z)$ be the indicator event that $\mathcal{O}$ accepts $z$ and $2^{-(t+1)} < N \cdot \tilde{D}(z)-1 \le 2^{-t}$. (For $t = 1,$ we also let $\BI_{+, 1} (z)$ indicate when $\tilde{D}(z) > \frac{2}{N}$, i.e., $N \cdot \tilde{D}(z) - 1 > 1$ also means $\BI_{+, 1}(z) = 1$.) Likewise, let $\BI_{-, t}(z)$ be the indicator event that $\mathcal{O}$ accepts $z$ and $2^{-(t+1)} < 1-N \cdot \tilde{D}(z) < 2^{-t}$ (we also let $\BI_{-, 1} (z)$ indicate when $\tilde{D}(z) < \frac{1}{2N}$), and let $\BI_T(z)$ be the indicator event that $\mathcal{O}$ accepts $z$ and $|N \cdot \tilde{D}(z)-1| \le 2^{-T}.$ Next, let $q_{+, t}(z) = \BP(\BI_{+, t}(z))$, $q_{-, t}(z) = \BP(\BI_{-, t}(z))$, and $q_T(z) = \BP(\BI_T(z)).$ Also, let $q_{+, t} = \BP(\BI_{+, t}),$ $q_{-, t} = \BP(\BI_{-, t}),$ and $q_T = \BP(\BI_T),$ where the probability is now also over $z \leftarrow \mathcal{U}$. Finally, let $g_{+, t} = \BE[\BI_{+, t}(z) \cdot (N \cdot D(z)-1)],$ $g_{-, t} = \BE[\BI_{-, t}(z) \cdot (1-N \cdot D(z))],$ and $g_T = \BE[\BI_T(z) \cdot |N \cdot D(z)-1|],$ where the expectations again are also over $z \leftarrow \mathcal{U}.$
    
    Now, fix some $0 \le t \le T-1$ and consider $\delta = 2^{-(t+1)}.$ For a sample $z \leftarrow \mathcal{U}$, we create a random variable $Z_{+, t}$ as follows. First, we determine if $\mathcal{O}$ accepts $z$ and, if so, we use the algorithm of Lemma \ref{ZEstimate} with $\beta = 2^{-t + O(1)}$ to sample the indicator variable $\BI_{+, t}(z)$ and determine if $\BI_{+, t}(z) = 1.$ While a sample from $\BI_{+, t}(z)$ technically requires setting $\beta = \frac{\eps}{160 \gamma_1},$ we note that setting $\beta = 2^{-t+O(1)} > \frac{\eps}{160 \gamma_1}$ to sample the indicator variable $\BI_{+, t}(z)$ has the same distribution as sampling $\BI_{+, t}(z)$ with $\beta = \frac{\eps}{160 \gamma_1}$, with at most $O(\eps^9)$ error. If either $\mathcal{O}$ rejects $z$ or $\BI_{+, t}(z) = 0,$ we set $Z_{+, t} = 0.$ 
    Else, if $\mathcal{O}$ accepts $z$ and $\BI_{+, t}(z) = 1$, we saw in Lemma \ref{SingleElement} that using an expected $O(1)$ queries to \Call{Pcond}{}, we could create a random variable $V$ with expectation $\frac{D(z)}{D(x)}$ and variance $O(1)$ (since $\frac{D(z)}{D(x)} = \Theta(1)$). Thus, by averaging $O\left(\frac{1}{\delta^2}\right)$ copies of the random variable $V$, we get a random variable $\bar{V}$, which when conditioned on $z$, has expectation $\frac{D(z)}{D(x)}$ and variance $O(\delta^2)$. We will finally return $Z_{+, t} = N \cdot \tilde{D}(x) \cdot \bar{V} - 1$ in this case. In total, we use an expected $O(\delta^{-2} \log \eps^{-1})$ queries to \Call{Pcond}{} and $\mathcal{O}$ to generate $Z_{+, t}.$
    
    Conditioned on $z$, we have that 
\begin{align*}
    \BE[Z_{+, t}|z] &= \BP(\BI_{+, t}(z) = 1) \cdot  \BE[Z_{+, t}|z, \BI_{+, t}(z) = 1] \\
    &= q_{+, t}(z) \cdot \left(N \cdot \tilde{D}(x) \cdot \frac{D(z)}{D(x)} - 1\right) \\
    &= q_{+, t}(z) \cdot \left((N \cdot D(z) - 1) \cdot \frac{\tilde{D}(x)}{D(x)} + \frac{\tilde{D}(x)-D(x)}{D(x)}\right).
\end{align*}
    Taking the expectation over $z$, we have that
\begin{align*}
    \BE[Z_{+, t}] &= \frac{\tilde{D}(x)}{D(x)} \cdot \BE[q_{+, t}(z) \cdot (N \cdot D(z)-1)] + \frac{\tilde{D}(x) - D(x)}{D(x)} \cdot \BE[q_{+, t}(z)] \\
    &= \frac{\tilde{D}(x)}{D(x)} \cdot g_{+, t} + \frac{\tilde{D}(x)-D(x)}{D(x)} \cdot q_{+, t} \\
    &= g_{+, t} + \frac{\tilde{D}(x)-D(x)}{D(x)} \cdot \left(g_{+, t}+q_{+, t}\right).
\end{align*}

    Next, note that since $Var(\bar{V}|z, \BI_{+, t} = 1) = O(\delta^{2}),$ and since $N \cdot \tilde{D}(x) = O(1),$ $Var(Z_{+, t}|z, \BI_{+, t} = 1) = O(\delta^{2})$. We also have $Var(Z_{+, t}|z, \BI_{+, t} = 0) = 0.$ We also have that $\BE[Z_{+, t}|z, \BI_{+, t} = 1] = N \cdot D(z) \cdot \frac{\tilde{D}(x)}{D(x)} - 1 = O(\delta)$ for all $z$ that can allow $\BI_{+, t} = 1,$ since $\delta > \frac{\eps}{\gamma_1}$ and $\BE[Z_{+, t}|z, \BI_{+, t} = 0] = 0.$ Therefore, by the Law of Total Variance,
\[Var(Z_{+, t}) = \BE[Var(Z_{+, t}|z, \BI_{+, t} = 1)] + Var(\BE[Z_{+, t}|z, \BI_{+, t} = 0]) = O(\delta^{2}).\]

    Therefore, by using $O(\delta^{-2} \log \eps^{-1})$ queries to \Call{Pcond}{} and $\mathcal{O},$ we can generate a random variable $Z_{+, t}$ with mean $g_{+, t} + O\left(\frac{\eps}{\gamma_1}\right) \cdot (g_{+, t}+q_{+, t})$ and variance $O(\delta^2)$ for $\delta = 2^{-(t+1)}.$ By the same argument, we can also generate a random variable $Z_{-, t}$ with mean $g_{-, t} + O\left(\frac{\eps}{\gamma_1}\right) \cdot (g_{-, t}+q_{-, t})$ and variance $O(\delta^2)$. Finally, we can also generate $Z_T$ with mean $g_T + O\left(\frac{\eps}{\gamma_1}\right) \cdot (g_T+q_T)$ and variance $O\left(\delta^2\right)$ for $\delta = 2^{-T}.$ By generating $O((\delta/\eps)^2 \cdot \log \eps^{-1})$ repetitions of $Z_{+, t}$ and averaging them, we can get a random variable $W_{+, t}$ with the same mean but variance $O(\eps^2/(\log \eps^{-1})),$ using $O(\eps^{-2} \log^2 \eps^{-1})$ queries to \Call{Pcond}{} and $\mathcal{O}.$ The same is true for $W_{-, t}$ and $W_T$.
    
    Our final estimate will be
\[W := W_T + \sum_{t = 1}^{T-1} (W_{+, t} + W_{-, t}),\]
    where $W_T, W_{+, t}, W_{-, t}$ are all determined using independent samples. We have that
\begin{align*}
\BE[W] &= g_T + O\left(\frac{\eps}{\gamma_1}\right) (g_T+q_T) + \sum_{t = 1}^{T-1} \left(g_{+, t}+g_{-, t} + O\left(\frac{\eps}{\gamma_1}\right) \cdot (g_{+, t}+g_{-, t}+q_{+, t} + q_{-, t})\right) \\
&= \left(1 + O\left(\frac{\eps}{\gamma_1}\right)\right) \cdot \left(g_T + \sum_{t = 1}^{T-1} (g_{+, t} + g_{-, t})\right) + O\left(\frac{\eps}{\gamma_1}\right) \cdot \left(q_T + \sum_{t = 1}^{T-1} (q_{+, t} + q_{-, t})\right).
\end{align*}
    Now, note that if $z$ is accepted by $\mathcal{O}$ and we draw a sample $\tilde{D}(z),$ then exactly one $\BI_{+, t}, \BI_{-, t}, \BI_T$ be a $1$, and all others be $0$. Therefore,
\[g_T + \sum_{t = 1}^{T-1} (g_{+, t} + g_{-, t}) = \BE_{z \sim \mathcal{U}} \left(\BI(\mathcal{O} \text{ accepts } z) \cdot |N \cdot D(z) - 1|\right) = \sum_{z = 1}^{N} s(z) \cdot \left|D(z) - \frac{1}{N}\right| = O(\gamma_1),\]
    since $s(z) > 0$ implies that $D(z) = O(1/N).$ Also,
\[q_T + \sum_{t = 1}^{T-1} (q_{+, t} + q_{-, t}) = \BE_{z \sim \mathcal{U}} \left(\BI(\mathcal{O} \text{ accepts } z)\right) = \BP_{z \sim \mathcal{U}} (\mathcal{O} \text{ accepts } z) = \sum_{z = 1}^{N} s(z) \cdot \frac{1}{N} = \gamma_1.\]

    Therefore, 
\[\BE[W] = \left(1 + O\left(\frac{\eps}{\gamma_1}\right)\right) \cdot \left(\sum\limits_{z = 1}^{N} s(z) \cdot \left|D(z) - \frac{1}{N}\right|\right) + O\left(\frac{\eps}{\gamma_1}\right) \cdot \gamma_1 = \left(\sum\limits_{z = 1}^{N} s(z) \cdot \left|D(z) - \frac{1}{N}\right|\right) + O(\eps).\]
    Moreover, since generating the $W_{+, t}$'s, $W_{-, t}$'s, and $W_T$ use independent queries, we have
\[Var(W) = Var(W_T) + \sum_{t = 1}^{T-1} Var(W_{+, t}) + \sum_{t = 1}^{T-1} Var(W_{-, t}) = O(\log \eps^{-1}) \cdot O\left(\frac{\eps^2}{\log \eps^{-1}}\right) = O(\eps^2).\]
    Therefore, with probability at least $0.97,$ we have that $W = O(\eps) + \sum\limits_{z = 1}^{N} s(z) \cdot \left|D(z) - \frac{1}{N}\right|.$
\end{proof}

\subsection{Creating the Oracle} \label{OracleCreate}

In this section, we create an oracle $\mathcal{O}'$ which is a slightly modified version of the oracle $\mathcal{O}$ used in Subsection \ref{OracleGiven}. In some cases for $\mathcal{D},$ we will not be able to create such an oracle, but we show in the next subsection that in both cases, regardless of whether we have found an oracle or not, we can still determine $\TV(\mathcal{D}, \mathcal{U})$ up to an additive error of $O(\eps)$.

We first prove the following lemma, which we note is very similar in idea to Proposition \ref{HighEltMass}.

\begin{lemma} \label{LowHighEltMass}
    Suppose $\TV(\mathcal{D}, \mathcal{U}) \le 1-3\eps.$ Then, there exists an integer $t$ such that $- \log_{1.01} \eps^{-1} \le t \le \log_{1.01} \eps^{-1} - 1$ and that
\[\BP_{x \sim \mathcal{D}} \left(D(x) \in \left[\frac{1.01^t}{N}, \frac{1.01^{t+1}}{N}\right]\right) \ge \frac{\eps}{210 \log \eps^{-1}}, \hspace{0.3cm} \BP_{x \sim \mathcal{U}} \left(D(x) \in \left[\frac{1.01^t}{N}, \frac{1.01^{t+1}}{N}\right]\right) \ge \frac{\eps}{210 \log \eps^{-1}}.\]
\end{lemma}

\begin{proof}
    Assume WLOG that $\eps^{-1}$ is an integer power of $1.01$. For each integer $t: - \log_{1.01} \eps^{-1} \le t \le \log_{1.01} \eps^{-1} - 1$, define $p_t = \BP_{x \sim \mathcal{D}}\left(D(x) \in \left[\frac{1.01^t}{N}, \frac{1.01^{t+1}}{N}\right]\right)$ and $q_t = \BP_{x \sim \mathcal{U}}\left(D(x) \in \left[\frac{1.01^t}{N}, \frac{1.01^{t+1}}{N}\right]\right)$. Also let $p_{-} = \BP_{x \sim \mathcal{D}}(D(x) \le \eps/N)$, $q_{-} = \BP_{x \sim \mathcal{U}}(D(x) \le \eps/N)$, $p_{+} = \BP_{x \sim \mathcal{D}}(D(x) \ge \eps^{-1}/N)$, and $q_{+} = \BP_{x \sim \mathcal{U}}(D(x) \ge \eps^{-1}/N)$.
    
    Now, note that by Proposition \ref{BasicTV},
\[1-\TV(\mathcal{D}, \mathcal{U}) = \sum_{i = 1}^{N} \min\left(D(i), \frac{1}{N}\right) \le \min(p_-, q_-) + \min(p_+, q_+) + \sum_{t = - \log_{1.01} \eps^{-1}}^{\log_{1.01} \eps^{-1} - 1} \min(p_t, q_t).\]
    However, note that $p_- \le \eps,$ since there are at most $N$ elements $x$ with $D(x) \le \frac{\eps}{N}.$ Likewise, $q_+ \le \eps,$ since there are at most $\eps \cdot N$ elements $x$ with $D(x) \ge \frac{\eps^{-1}}{N}.$ Therefore, if $\TV(\mathcal{D}, \mathcal{U}) \le 1-3 \eps,$ we have that $1-\TV(\mathcal{D}, \mathcal{U}) \ge 3 \eps,$ so there must exist some $t$ in the desired range such that $\min(p_t, q_t) \ge \frac{\eps}{2 \log_{1.01} \eps^{-1}} \ge \frac{\eps}{210 \log \eps^{-1}}.$
\end{proof}

Now, for any $x \in [N],$ recall that we can use \Call{Compare}{} to compare $x$ and some other element $w$. Using the notation of Proposition \ref{Compare}, if we set $\gamma = 0.01,$ \Call{Compare}{$w, x, \gamma$} returns $\alpha$ such that if $\frac{D(w)}{D(x)} \in [0.99, 1.01],$ then $\alpha \in [0.98, 1.02]$ with probability $1-\eps^{10}$, but if $\frac{D(w)}{D(x)} \not\in [0.97, 1.03],$ then $\alpha \not\in [0.98, 1.02]$ with probability $1-\eps^{10}.$ To make use of this observation, we first define the following.

\begin{defn} \label{dandu}
    For any element $x \in [n],$ let $d(x) = \BP_{w \sim \mathcal{D}}  \left(\Call{Compare}{w, x, 0.01}\right) \in [0.98, 1.02]$, and let $u(x) = \BP_{w \sim \mathcal{U}}  \left(\Call{Compare}{w, x, 0.01}\right) \in [0.98, 1.02]$.
\end{defn}

Next, we will find some element $x$ as well as a constant-factor approximation $\hat{D}(x)$ of $D(x)$.

\begin{lemma} \label{ConstantApprox}
    There exists an algorithm \Call{ConstantApprox}{} using $O(\eps^{-2} \log^5 \eps^{-1})$ queries to \Call{Pcond}{} and \Call{Samp}{} that returns a set $S$ with elements of the form $(x_r, \hat{D}(x_r))$ with the following guarantees.
\begin{enumerate}
    \item For all $(x, \hat{D}(x)) \in S,$ with probability at least $1-\eps^6,$ we have that $\frac{\hat{D}(x)}{D(x)} \in [0.9, 1.1]$ and $\hat{D}(x) \in \left[\frac{0.9 \eps}{N}, \frac{1.1 \eps^{-1}}{N}\right].$
    \item If $\TV(\mathcal{D}, \mathcal{U}) \le 1-3 \eps,$ then with probability at least $1-\eps^6,$ $S$ is nonempty.
    \item With probability at least $1-\eps^6,$ at least one of the following is true:
    \begin{enumerate}
        \item Some $(x, \hat{D}(x)) \in S$ satisfies $\hat{D}(x) \in \left[\frac{5}{9N}, \frac{9}{5N}\right]$.
        \item If some $(x, \hat{D}(x)) \in S$ satisfies $D(x) \ge \frac{1}{N},$ the $x$ with the smallest such $\hat{D}(x)$ satisfies $\BP_{w \sim \mathcal{U}} \left(\frac{1}{N} \le D(w) \le 0.8 \cdot D(x)\right) \le 2 \eps$. Likewise, if some $(x, \hat{D}(x)) \in S$ satisfies $D(x) \le \frac{1}{N},$ then the $x$ with the largest such $\hat{D}(x)$ satisfies $\BP_{w \sim \mathcal{D}} \left(\frac{1}{N} \ge D(w) \ge 1.25 \cdot D(x)\right) \le 2 \eps.$
    \end{enumerate}
\end{enumerate}
\end{lemma}

\begin{proof}
    First, choose $R = O(\eps^{-1} \log^2 \eps^{-1})$ and sample elements $x_1, \dots, x_R \leftarrow \mathcal{D}.$ Also, sample elements $w_1, \dots, w_R \leftarrow \mathcal{D}$ and $y_1, \dots, y_R \leftarrow \mathcal{U}.$
    
    For each pair $i, r$ with $1 \le i, r \le R$, we run $\Call{Compare}{w_i, x_r, 0.01}$ and $\Call{Compare}{y_i, x_r, 0.01}.$ Let $\tilde{d}(x_r) = \frac{1}{R} \cdot \#\{i \in [R]: \Call{Compare}{w_i, x_r, 0.01} \in [0.98, 1.02]\},$ and let $\tilde{u}(x_r) = \frac{1}{R} \cdot \#\{i \in [R]: \Call{Compare}{y_i, x_r, 0.01} \in [0.98, 1.02]\}.$ We note that $\tilde{d}(x_r)$ is distributed as $\frac{1}{R} \cdot \operatorname{Bin}(R, d(x_r))$ and $\tilde{u}(x_r)$ is distributed as $\frac{1}{R} \cdot \operatorname{Bin}(R, u(x_r)).$ Therefore, by a basic application of the Chernoff bound, if $d(x_r), u(x_r) \ge \frac{\eps}{400 \log \eps^{-1}},$ then assuming $R \ge C \eps^{-1} \log^2 \eps^{-1}$ for a sufficiently large constant $C$, we have that $\tilde{d}(x_r) \in [0.99 \cdot d(x_r), 1.01 \cdot d(x_r)]$ and $\tilde{u}(x_r) \in [0.99 \cdot u(x_r), 1.01 \cdot u(x_r)]$ with probability at least $1-\eps^8$. Moreover, if $d(x_r) \le \frac{\eps}{300 \log \eps^{-1}},$ we have that $\tilde{d}(x_r) \le \frac{\eps}{250 \log \eps^{-1}}$, and if $u(r) \le \frac{\eps}{300 \log \eps^{-1}},$ we have that $\tilde{u}(x_r) \le \frac{\eps}{250 \log \eps^{-1}}$.
    
    Now, the algorithm proceeds as follows. For each $r$, we check whether both $\tilde{d}(x_r) > \frac{\eps}{250 \log \eps^{-1}}$ and $\tilde{u}(x_r) > \frac{\eps}{250 \log \eps^{-1}}$. In this case, we let $\hat{D}(x_r) = \frac{\tilde{d}(x_r)}{\tilde{u}(x_r)} \cdot \frac{1}{N}.$ We will ignore any $\hat{D}(x_r) \not\in \left[\frac{0.9 \eps}{N}, \frac{1.1 \eps^{-1}}{N}\right].$ With probability at least $1-\eps^6$, for any $x_r$ with either $d(x_r)$ or $u(x_r)$ less than $\frac{\eps}{300 \log \eps^{-1}}$, we will output either $\tilde{d}(x_r) \le \frac{\eps}{250 \log \eps^{-1}}$ or $\tilde{u}(x_r) \le \frac{\eps}{250 \log \eps^{-1}}$. Therefore, for any $x_r$ with both $\tilde{d}(x_r), \tilde{u}(x_r) > \frac{\eps}{250 \log \eps^{-1}},$ we must have that $d(x_r), u(x_r) \ge \frac{\eps}{300 \log \eps^{-1}},$ which means that $\tilde{d}(x_r)$ and $\tilde{u}(x_r)$ are accurate up to a multiplicative error of $0.01$ by a simple application of Chernoff. So, for all $r$ such that we output some $\hat{D}(x_r),$ we have that $\hat{D}(x_r) \in \left[\frac{0.99}{1.01} \cdot \frac{1}{N} \cdot \frac{d(x_r)}{u(x_r)}, \frac{1.01}{0.99} \cdot \frac{1}{N} \cdot \frac{d(x_r)}{u(x_r)}\right]$.
    
    Next, we will look at the ratio $\frac{d(x_r)}{u(x_r)}.$ For each $w$, define $p(w, x_r)$ to be the probability that $\Call{Compare}{w, x_r, 0.01} \in [0.98, 1.02].$ Then, $d(x_r) = \sum_w D(w) \cdot p(w, x_r)$ and $u(x_r) = \sum_w \frac{1}{N} \cdot p(w, x_r)$. However, recalling that $p(w, x_r) \le \eps^{10}$ whenever $D(w) \not\in [0.97 D(x_r), 1.03 D(x_r)],$ we have that
\[d(x_r) = O(\eps^{10}) + \sum_{w: D(w) \in [0.97 D(x_r), 1.03 D(x_r)]} D(w) \cdot p(w, x_r)\]
    and
\[u(x_r) = O(\eps^{10}) + \sum_{w: D(w) \in [0.97 D(x_r), 1.03 D(x_r)]} \frac{1}{N} \cdot p(w, x_r).\]
    However, note that
\[\sum_{w: D(w) \in [0.97 D(x_r), 1.03 D(x_r)]} D(w) \cdot p(w, x_r) \in [0.97, 1.03] \cdot N \cdot D(x_r) \cdot \sum_{w: D(w) \in [0.97 D(x_r), 1.03 D(x_r)]} \frac{1}{N} \cdot p(w, x_r)\]
    due to our restriction of $D(w) \in [0.97 D(x_r), 1.03 D(x_r)]$. Therefore, if $d(x_r), u(x_r) \ge \frac{\eps}{300 \log \eps^{-1}},$ the $O(\eps^{10})$ additive errors are negligible, and we have that $d(x_r) \in [0.96, 1.04] \cdot N \cdot D(x_r) \cdot u(x_r)$. Thus, whenever we output $\hat{D}(x_r),$ we have that 
\[\hat{D}(x_r) \in [0.96, 1.04] \cdot N \cdot D(x_r) \cdot \left[\frac{0.99}{1.01} \cdot \frac{1}{N}, \frac{1.01}{0.99} \cdot \frac{1}{N}\right] \subseteq [0.9 \cdot D(x_r), 1.1 \cdot D(x_r)].\]

    Thus, we have proven that with probability at least $1-\eps^6,$ all returned $\hat{D}(x_r)$'s are accurate. Next, we prove the second condition, i.e., if $\TV(\mathcal{D}, \mathcal{U}) \le 1-3 \eps,$ then at least one $r$ will result in a $\hat{D}(x_r)$ being returned. To see why, if $\TV(\mathcal{D}, \mathcal{U}) \le 1-3 \eps,$ then by Lemma \ref{LowHighEltMass} there exists an integer $t$ such that $[1.01^t, 1.01^{t+1}] \subset [\eps, \eps^{-1}]$ and
\[\BP_{x \sim \mathcal{D}} \left(D(x) \in \left[\frac{1.01^t}{N}, \frac{1.01^{t+1}}{N}\right]\right) \ge \frac{\eps}{210 \log \eps^{-1}}, \hspace{0.3cm} \BP_{x \sim \mathcal{U}} \left(D(x) \in \left[\frac{1.01^t}{N}, \frac{1.01^{t+1}}{N}\right]\right) \ge \frac{\eps}{210 \log \eps^{-1}}.\]
    Therefore, with probability at least $1-\eps^7,$ some $x_r$ with $D(x_r) \in \left[\frac{1.01^t}{N}, \frac{1.01^{t+1}}{N}\right]$ will be sampled. It is clear that $d(x_r), u(x_r) \ge \frac{\eps}{210 \log \eps^{-1}},$ and we have seen that for all such $x_r$, we will output $\hat{D}(x_r) \in [0.9 D(x_r), 1.1 D(x_r)],$ with failure probability at most $\eps^6.$ Moreover, $\hat{D}(x_r) \in \left[0.9 \cdot \frac{\eps}{N}, 1.1 \cdot \frac{\eps^{-1}}{N}\right],$ since $[1.01^t, 1.01^{t+1}] \in [\eps, \eps^{-1}].$

    Finally, we verify the third condition.
    Suppose no $x_r$ has a returned $\hat{D}(x_r) \in \left[\frac{5}{9N}, \frac{9}{5N}\right],$ but that some $x_r$ has $\hat{D}(x_r) \ge \frac{9}{5N}$. Choose such an $x_r$ that minimizes $\hat{D}(x_r)$. Now, let $t_1$ be the smallest nonnegative integer $t \le \log_{1.01} \eps^{-1}-1$ such that $\BP_{w \sim \mathcal{U}} \left(D(w) \in \left[\frac{1.01^t}{N}, \frac{1.01^{t+1}}{N}\right]\right) \ge \frac{\eps}{210 \log \eps^{-1}}.$ Then, since $\frac{1.01^{t_1}}{N} \ge \frac{1}{N},$ we will also have that $\BP_{w \sim \mathcal{D}}\left(D(w) \in \left[\frac{1.01^{t_1}}{N}, \frac{1.01^{t_1+1}}{N}\right]\right) \ge \frac{\eps}{210 \log \eps^{-1}}.$ Thus, some $x$ with $D(x) \in \left[\frac{1.01^{t_1}}{N}, \frac{1.01^{t_1+1}}{N}\right]$ and some output $\hat{D}(x)$ of $D(x)$ will be returned, with probability at least $1-\eps^6$. Therefore, since $\hat{D}(x) \ge \frac{9}{5N},$ we have that
\[0.9 \cdot D(x_r) \le \hat{D}(x_r) \le \hat{D}(x) \le 1.1 \cdot D(x) \le 1.1 \cdot \frac{1.01^{t_1+1}}{N}.\]
    This implies that $0.8 \cdot D(x_r) \le 1.01^{t_1}/N,$ so $\BP\left(\frac{1}{N} \le D(w) \le 0.8 \cdot D(x_r)\right) \le t_1 \cdot \frac{\eps}{210 \log \eps^{-1}} \le \eps.$
    
    Now, if no such $t_1$ exists, then $\BP_{w \sim \mathcal{U}}\left(\frac{1}{N} \le D(w) \le \frac{\eps^{-1}}{N}\right) \le \eps$ and $\BP_{w \sim \mathcal{U}} \left(D(w) \ge \frac{\eps^{-1}}{N}\right) \le \eps,$ so we even have that $\BP_{w \sim \mathcal{U}}\left(\frac{1}{N} \le D(w)\right) \le 2 \eps.$ Finally, we note that the proof for the other direction, i.e., if $D(x) \le \frac{1}{N},$ is identical.
\end{proof}

We are now ready to prove the main lemma of this subsection. Informally, we show that as long as we found some $D(x)$ that is very close to $\frac{1}{N}$, we also find an oracle $\mathcal{O}'$ which essentially separates between elements with probabilities less than $\frac{1}{N}$ and probabilities greater than $\frac{1}{N},$ but allows for elements with probabilities close to $\frac{1}{N}$ to be ``unknown.''

\begin{lemma} \label{Oracle}
    Suppose that Lemma \ref{ConstantApprox} outputs some $(x, \hat{D}(x))$ such that $\hat{D}(x) \in \left[\frac{5}{9N}, \frac{9}{5N}\right]$ and $\hat{D}(x) \in [0.9 \cdot D(x), 1.1 \cdot D(x)].$ Then, there exists an algorithm $\Call{Oracle}{}$ using $O(\eps^{-2} \log^3 \eps^{-1})$ additional queries to \Call{Pcond}{} and \Call{Samp}{} that generates a randomized oracle $\mathcal{O}'$ which takes as input an element $z \in [N]$ and outputs either $0, 1,$ or $-1$. Moreover, the oracle satisfies the following four properties for all $z \in [N]$. 
\begin{enumerate}
    \item If $D(z) > \frac{5}{N},$ then $\mathcal{O}'(z) = 1$ with probability at least $1-O(\eps^6).$
    \item If $D(z) < \frac{1}{5N},$ then $\mathcal{O}'(z) = -1$ with probability at least $1-O(\eps^6).$ 
    \item If $\frac{1}{5N} \le D(z) \le \frac{1}{N},$ then $\mathcal{O}'(z)$ is either $0$ or $-1$ with probability at least $1-O(\eps^6).$
    \item If $\frac{1}{N} \le D(z) \le \frac{5}{N},$ then $\mathcal{O}'(z)$ is either $0$ or $1$ with probability at least $1-O(\eps^6)$.
\end{enumerate}
    Finally, calling the oracle $\mathcal{O}'$ requires $O(\log \eps^{-1})$ calls to \Call{Pcond}{}.
\end{lemma}

\begin{proof}
    By our assumption about what was returned by Lemma \ref{ConstantApprox}, we have that $D(x) \in \left[\frac{\eps}{2 N}, \frac{2 \eps^{-1}}{N}\right]$.

    The oracle works as follows. For any $z$, the oracle runs \Call{Compare}{$z, x, 0.01$} and returns some $\alpha.$ If $\alpha \in [0.45, 2.2],$ then we know that $\frac{D(z)}{D(x)} \in [0.4, 2.5],$ so $D(z) \in \left[\frac{1}{5N}, \frac{5}{N}\right],$ so $\mathcal{O}'$ returns $0$. If $\alpha < 0.45,$ then we know that $\frac{D(z)}{D(x)} < 0.5,$ so $D(z) < \frac{1}{N},$ so $\mathcal{O}'$ returns $-1$. Finally, if $\alpha > 2.2,$ then we know that $\frac{D(z)}{D(x)} > 2,$ so $D(z) > \frac{1}{N},$ so $\mathcal{O}'$ returns $1$. Finally, note that calling the oracle just requires calling \Call{Compare}{$z, x, 0.01$}, which needs $O(\log \eps^{-1})$ calls to \Call{Pcond}{}.
\end{proof}

\subsection{Finishing the Algorithm} \label{FinishAlgorithm}
In this section, we show how to combine subsections \ref{OracleGiven} and \ref{OracleCreate} to prove Theorem \ref{TolerantUnif}.

\begin{lemma} \label{GivenGoodElt}
    Suppose Lemma \ref{ConstantApprox} finds some $(x, \hat{D}(x))$ such that $\hat{D}(x) \in \left[\frac{5}{9N}, \frac{9}{5N}\right]$ and $\hat{D}(x) \in [0.9 \cdot D(x), 1.1 \cdot D(x)].$ Then, there is an algorithm \Call{GivenGoodElt}{} that uses $O(\eps^{-2} \log^3 \eps^{-1})$ additional queries to \Call{Pcond}{} and \Call{Samp}{} and  with probability at least $0.9$ returns $\TV(\mathcal{D}, \mathcal{U})$ with error $O(\eps)$.
\end{lemma}

\begin{proof}
    First, we use Lemma \ref{Oracle} to get the oracle $\mathcal{O}',$ and convert this to the oracle $\mathcal{O}$ of Subsection \ref{OracleGiven} by having the oracle $\mathcal{O}$ accept an input $i$ if and only if $\mathcal{O}'$ returns $0$ on input $i$. Note that $\mathcal{O}$ satisfies the requirements except that the probability that $\mathcal{O}$ accepts an element $i$ with $D(i) \not\in \left[\frac{1}{5N}, \frac{5}{N}\right]$ is at most $\eps^6$ rather than just never accepting. This, however, is fine, since we will only make at most $\tilde{O}(\eps^{-2})$ calls to $\mathcal{O}$. Importantly, note that if we sum $s(i) \cdot |D(i) - \frac{1}{N}|$ over the $i$ such that $D(i) \not\in \left[\frac{1}{5N}, \frac{5}{N}\right]$, where $s(i) = \BP(\mathcal{O}'(i) = 0)$, we get $O(\eps^6),$ since $\sum_i D(i), \sum_i \frac{1}{N} \le 1$ and $s(i) \le \eps^6$ for all such $i$. Therefore, using $O(\eps^{-2} \log^2 \eps^{-1})$ queries to \Call{Pcond}{} and our oracle $\mathcal{O}$ which we have created, we can determine $\sum_{i = 1}^{N} s(i) \cdot |D(i)-\frac{1}{N}|$ up to an $O(\eps)$ additive error using Lemma \ref{EstimateCloseTerms}, where $s(i)$ is the probability of $\mathcal{O}'$ returning $0$ on $i$. For Lemma \ref{EstimateCloseTerms} to work, we will need $\gamma_1 = \sum_{i = 1}^{N} s(i) \cdot \frac{1}{N}$ to be at least $10 \eps,$ but this can be checked easily, and if $\gamma_1 = O(\eps),$ then $\sum_{i = 1}^{N} s(i) \cdot |D(i) - \frac{1}{N}| = O(\eps)$ so we can estimate it as $0$.
    
    Now, let $r(i)$ be the probability that $\mathcal{O}'$ returns $-1$ on $i$ and $t(i)$ be the probability that $\mathcal{O}'$ returns $1$ on $i$. Note that $\BP_{y \sim \mathcal{U}}(\mathcal{O}'(y) = -1) = \sum_{i = 1}^{N} r(i) \cdot \frac{1}{N}$ and $\BP_{z \sim \mathcal{D}}(\mathcal{O}'(z) = -1) =  \sum_{i = 1}^{N} r(i) \cdot D(i)$. Likewise, $\BP_{y \sim \mathcal{U}}(\mathcal{O}'(y) = 1) = \sum_{i = 1}^{N} t(i) \cdot \frac{1}{N}$ and $\BP_{z \sim \mathcal{D}}(\mathcal{O}'(z) = 1) =  \sum_{i = 1}^{N} t(i) \cdot D(i)$.
    
    Now, we note that since $t(i) = O(\eps^6)$ for all $i$ such that $D(i) < \frac{1}{N}$ and $r(i) = O(\eps^6)$ for all $i$ such that $D(i) > \frac{1}{N},$ we have that $\sum_{i = 1}^{N} r(i) \cdot (\frac{1}{N} - D(i)) = O(\eps^6) + \sum_{i = 1}^{N} r(i) \cdot |\frac{1}{N} - D(i)|.$ Likewise, we have that $\sum_{i = 1}^{N} t(i) \cdot (D(i) - \frac{1}{N}) = O(\eps^6) + \sum_{i = 1}^{N} t(i) \cdot |D(i) - \frac{1}{N}|.$ Recall that we know $\sum_{i = 1}^{N} s(i) \cdot |D(i)-\frac{1}{N}|$ up to an $O(\eps)$ additive factor, and that we can compute $\BP_{y \sim \mathcal{U}}(\mathcal{O}'(y) = -1)$, $\BP_{z \sim \mathcal{D}}(\mathcal{O}'(z) = -1)$, $\BP_{y \sim \mathcal{U}}(\mathcal{O}'(y) = 1)$, and $\BP_{z \sim \mathcal{D}}(\mathcal{O}'(z) = 1)$ each up to an $O(\eps)$ error using $O(\eps^{-2})$ queries to $\mathcal{O}'$ with probability at least $0.9$, where we either sample from $\mathcal{D}$ or $\mathcal{U}$. Therefore, we can compute
\[\sum_{i = 1}^{N} s(i) \cdot \left|D(i)-\frac{1}{N}\right| + \sum_{i = 1}^{N} r(i) \cdot \left|D(i)-\frac{1}{N}\right| + \sum_{i = 1}^{N} t(i) \cdot \left|D(i)-\frac{1}{N}\right| = \sum_{i = 1}^{N} \left|D(i)-\frac{1}{N}\right| = 2 \cdot \TV(\mathcal{D}, \mathcal{U})\]
    up to an $O(\eps)$ factor, where we used the fact that $r(i)+s(i)+t(i) = 1$ for all $i$.
    
    In total, we used $O(\eps^{-2} \log^2 \eps^{-1})$ queries to \Call{Pcond}{}, \Call{Samp}{}, and $\mathcal{O}'.$ But since each call to $\mathcal{O}'$ uses $O(\log \eps^{-1})$ calls to \Call{Pcond}{}, the final query complexity of $O(\eps^{-2} \log^3 \eps^{-1}).$
\end{proof}

We are now ready to prove Theorem \ref{TolerantUnif}.

\begin{proof}[Proof of Theorem \ref{TolerantUnif}]
    We assume that the output of Lemma \ref{ConstantApprox} satisfies the guarantees, ignoring the $O(\eps^6)$ failure probability.

    First, suppose that in Lemma \ref{ConstantApprox}, we return $S = \{\}$. Then, we can say that $\TV(\mathcal{D}, \mathcal{U}) = 1$ and we are off by at most $3 \eps.$
    
    Next, suppose that in Lemma \ref{ConstantApprox}, we return some $(x, \hat{D}(x)) \in S$ with $\hat{D}(x) \in \left[\frac{5}{9N}, \frac{9}{5N}\right].$ Then, we can use Lemma \ref{GivenGoodElt} to finish the proof.

    Otherwise, we have that in Lemma \ref{ConstantApprox}, at least some pair $(x, \hat{D}(x)) \in S$ was found with $\hat{D}(x) \in \left[\frac{0.9 \eps}{N}, \frac{1.1 \eps^{-1}}{N}\right],$ but every such $(x_r, \hat{D}(x_r)) \in S$ satisfies $\hat{D}(x_r) \not\in \left[\frac{5}{9N}, \frac{9}{5N}\right].$ We will show first how to estimate $\BP_{y \sim \mathcal{U}} \left(D(y) \ge \frac{1}{N}\right)$ and then how to estimate $\BP_{z \sim \mathcal{D}} \left(D(z) < \frac{1}{N}\right).$ We combine these together to get the final estimate.
    
    Suppose there exists some $x_r$ returned by Lemma \ref{ConstantApprox} such that $\hat{D}(x_r) \ge \frac{1}{N}$. In that case, choose $x_r$ such that $\hat{D}(x_r) \le \hat{D}(x_{r'})$ for all $x_{r'}$ returned by Lemma \ref{ConstantApprox} with $\hat{D}(x_{r'}) \ge \frac{1}{N}.$ With probability at least $1-\eps^6,$ all $\hat{D}(x_r)$'s are accurate up to a $1 \pm 0.1$ factor, so $D(x_r) \ge \frac{3}{2N}$ since $\hat{D}(x_r) \ge \frac{9}{5N}$. We also know that
\[\BP_{y \sim \mathcal{U}} \left(\frac{1}{N} \le D(y) < 0.8 D(x_r)\right) \le 2\eps\]
    by Lemma \ref{ConstantApprox}. Therefore, the probability over $y \sim \mathcal{U}$ that \Call{Compare}{$y, x_r, 0.01$} is at least $0.78$ equals $\BP_{y \sim \mathcal{U}} \left(D(y) \ge \frac{1}{N}\right)$, up to a $3 \eps$ additive error. This is true because 
\[\BP_{y \sim \mathcal{U}}\left(\Call{Compare}{y, x_r, 0.01} \ge 0.78\right) \ge \BP_{y \sim \mathcal{U}}\left(\frac{D(y)}{D(x_r)} \ge 0.79 \right) - \eps \ge \BP_{y \sim \mathcal{U}} \left(D(y) \ge \frac{1}{N}\right) - 3 \eps,\]
    but 
\[\BP_{y \sim \mathcal{U}}\left(\Call{Compare}{y, x_r, 0.01} \ge 0.78 \right) \le \BP_{y \sim \mathcal{U}}\left(\frac{D(y)}{D(x_r)} \ge 0.77\right) + \eps \le \BP_{y \sim \mathcal{U}} \left(D(y) \ge \frac{1}{N}\right) + \eps.\]
    We can estimate $\BP_{y \sim \mathcal{U}}\left(\Call{Compare}{y, x_r, 0.01} \ge 0.78\right)$ up to a $2\eps$ error, using $O(\eps^{-2})$ samples of $y_i \leftarrow \mathcal{U}$ and computing $\Call{Compare}{y_i, x_r, 0.01}$ for each of them.
    Now, if no such $x_r$ with $\hat{D}(x_r) \ge \frac{1}{N}$ exists, then saw in the proof of Lemma \ref{ConstantApprox} that even
\[\BP_{y \sim \mathcal{U}} \left(\frac{1}{N} \le D(y)\right) \le 2\eps.\]
    Thus, we estimate $\BP_{x \sim \mathcal{U}} (D(x) \ge \frac{1}{N})$ as $0$, which is correct up to a $2 \eps$ additive error.
    
    Similarly, suppose there exists some $x_s$ such that $\hat{D}(x_s) < \frac{1}{N}$. In that case, choose $x_s$ such that $\hat{D}(x_s) \ge \hat{D}(x_{s'})$ for all $x_{s'}$ returned by Lemma \ref{ConstantApprox} with $\hat{D}(x_{s'}) < \frac{1}{N}.$ Also, with probability at least $1-\eps^6,$ all $\hat{D}(x_r)$'s are accurate up to a $1 \pm 0.1$ factor, so $D(x_s) \le \frac{2}{3N}$ since $\hat{D}(x_s) \le \frac{5}{9N}$. We also know that
\[\BP_{z \sim \mathcal{D}} \left(\frac{1}{N} > D(z) > 1.25 D(x_r)\right) < 2\eps\]
    by Lemma \ref{ConstantApprox}. Therefore, the probability over $z \sim \mathcal{D}$ that \Call{Compare}{$z, x_s, 0.01$} is at most $1.27$ equals $\BP_{z \sim \mathcal{D}} \left(z < \frac{1}{N}\right)$, up to a $3 \eps$ additive error. This is true because 
\[\BP_{z \sim \mathcal{D}}\left(\Call{Compare}{z, x_s, 0.01} \le 1.27\right) \ge \BP_{z \sim \mathcal{D}}\left(\frac{D(z)}{D(x_s)} \le 1.26\right) - \eps \ge \BP_{z \sim \mathcal{D}} \left(D(z) < \frac{1}{N}\right) - 3 \eps,\]
    but 
\[\BP_{z \sim \mathcal{D}}\left(\Call{Compare}{z, x_s, 0.01} \le 1.27\right) \le \BP_{z \sim \mathcal{D}}\left(\frac{D(z)}{D(x_s)} \le 1.28\right) + \eps \le \BP_{z \sim \mathcal{D}} \left(D(z) < \frac{1}{N}\right) + \eps.\]
    We can estimate $\BP_{z \sim \mathcal{D}}\left(\Call{Compare}{z, x_s, 0.01} \le 1.27\right)$ up to a $2\eps$ error, using $O(\eps^{-2})$ samples of $z_i \leftarrow \mathcal{D}$ and computing $\Call{Compare}{z_i, x_s, 0.01}$ for each of them.
    Now, if no $x_s$ exists, then we know that
\[\BP_{z \sim \mathcal{D}} \left(D(z) < \frac{1}{N}\right) \le 2\eps.\]
    by the same proof as in of Lemma \ref{ConstantApprox}. Namely, for any distribution $\mathcal{D}$, $\BP_{z \sim \mathcal{D}} \left(D(z) < \frac{\eps}{N}\right) \le \eps,$ and if $\BP_{z \sim \mathcal{D}} \left(\frac{1}{N} > D(z) \ge \frac{\eps}{N}\right) \ge \eps,$ then $\BP_{z \sim \mathcal{D}}\left(D(z) \in \left[\frac{1.01^{-(t+1)}}{N}, \frac{1.01^{-t}}{N}\right]\right) \ge \frac{\eps}{210 \log \eps^{-1}}$ for some $0 \le t < \log_{1.01} \eps^{-1}$, which also implies that $\BP_{z \sim \mathcal{U}}\left(D(z) \in \left[\frac{1.01^{-(t+1)}}{N}, \frac{1.01^{-t}}{N}\right]\right) \ge \frac{\eps}{210 \log \eps^{-1}}$. Thus, Lemma \ref{ConstantApprox} would find some $(x, \hat{D}(x))$ with $D(x)$ in the range $\left[\frac{1.01^{-(t+1)}}{N}, \frac{1.01^{-t}}{N}\right] \subset \left[\frac{\eps}{N}, \frac{1}{N}\right].$

    To finish, note that
\[\BP_{x \sim \mathcal{U}} \left(D(x) \ge \frac{1}{N}\right) + \BP_{x \sim \mathcal{D}}\left(D(x) < \frac{1}{N}\right) = \sum_{x: D(x) \ge \frac{1}{N}} \frac{1}{N} + \sum_{x: D(x) < \frac{1}{N}} D(x) = \sum_{x = 1}^{N} \min\left(D(x), \frac{1}{N}\right),\]
    which equals $1-\TV(\mathcal{D}, \mathcal{U})$ by Proposition \ref{BasicTV}. Therefore, we can estimate $\TV(\mathcal{D}, \mathcal{U})$ up to an $O(\eps)$ additive error, which concludes all cases.
\end{proof}

\section{An $\tilde{O}(\eps^{-4})$-query algorithm for Tolerant Identity Testing} \label{TolerantIdentity}

In this section, we present an algorithm that, given a known distribution $\mathcal{D}^*$ over $[N]$, makes $\tilde{O}(\eps^{-4})$ queries to \Call{Cond}{} with distribution $\mathcal{D}$ and determines $\TV(\mathcal{D}, \mathcal{D}^*)$ up to an $O(\eps)$ additive error. The dependence on the support size $N$ is optimal (i.e., no dependence), though it is possible that the dependence on $\eps$ can be improved to $\tilde{O}(\eps^{-2})$.

First, assume that $\mathcal{D}^*$ is ordered so that $D^*(1) \le D^*(2) \le \cdots \le D^*(N).$ We are allowed to permute the elements of both $\mathcal{D}$ and $\mathcal{D}^*$ with the same permutation, since this will not affect $\TV(\mathcal{D}, \mathcal{D}^*)$ and we can still make the same COND queries (after the same permutation is applied).

We will consider a slightly more general distribution problem. Suppose we have two distributions $\mathcal{P}$ and $\mathcal{P}^*$ over $[M],$ where $\mathcal{P}^*$ is known and $P^*(1) \le P^*(2) \le \cdots \le P^*(M)$, and we are given $\Call{Cond}{}$ access to $\mathcal{P}.$ Our goal is to determine $\sum_{i = 1}^{M} \min(c_1 P(i), c_2 P^*(i))$ up to an additive $O(\eps)$ error, where $c_1, c_2 \le 1$ are known constants. Since $\sum_{i = 1}^{M} P(i) = \sum_{i = 1}^{M} P^*(i) = 1,$ if either $c_1 = O(\eps)$ or $c_2 = O(\eps),$ then $\sum_{i = 1}^{M} \min(c_1 P(i), c_2 P^*(i)) = O(\eps)$ so we can just output $0$ as our estimate. For the case where $c_1, c_2 \gg \eps,$ we give an inductive approach. Namely, we show how to find some set $S \subset [M]$ such that either $P(S) \ge \frac{1}{3}$ or $P^*(S) \ge \frac{1}{3}$ and estimate $\sum_{i \in S} \min(c_1 P(i), c_2 P^*(i))$ up to an $O(\eps)$ additive error. To estimate $\sum_{i \not\in S} \min(c_1 P(i), c_2 P^*(i)),$ we can modify the distributions $P, P^*$ to be conditioned on $i \not\in S.$ Both $c_1$ and $c_2$ will either increase or stay the same, and either $c_1$ or $c_2$ will multiply by a factor of at most $\frac{2}{3},$ since either $P(S) \ge \frac{1}{3}$ or $P^*(S) \ge \frac{1}{3}.$ Therefore, we only need to repeat this process $O(\log \eps^{-1})$ times, until either $c_1$ or $c_2$ is $O(\eps).$ Our final error will be $O(\eps \cdot \poly \log \eps^{-1}),$ but we can fix this by replacing $\eps$ with $\eps' = \frac{\eps}{\poly \log \eps^{-1}}.$

\begin{theorem} \label{PartialDetermining}
    Suppose $\mathcal{P}, \mathcal{P}^*$ are distributions over $[M]$, where $\mathcal{P}^*$ is known, $P^*(1) \le P^*(2) \le \cdots \le P^*(M)$, and we have $\Call{Cond}{}$ access to $\mathcal{P}.$ Also, let $1 \ge c_1, c_2 \ge \eps$ be known constants. Then, there is an algorithm \Call{PartialDetermining}{} that uses $O(\eps^{-4} \log^6 \eps^{-1})$ queries to $\Call{Cond}{}_{\mathcal{P}}$, such that with probability at least $1-\eps$, the algorithm finds a set $S$ such that $P^*(S) \ge \frac{1}{3}$ as well as an estimate of $\sum_{i \in S} \min(c_1 P(i), c_2 P^*(i))$ which is accurate to an additive $O(\eps \log \eps^{-1})$ error.
\end{theorem}

    To begin the proof, we first let $L$ be the smallest integer such that $P^*([L]) = \BP_{x \sim \mathcal{P}^*} (x \le L) \ge \frac{1}{2}.$ Then, $P^*([L]) \ge \frac{1}{2}$ and $P^*([L:M]) \ge \frac{1}{2}$. We attempt to make $S = [L:M]$ or $S = [z]$ for some $z \ge L$. Note that for any set $S \subset [M]$,
\begin{align*}
\sum_{x \in S} \min(c_1 P(x), c_2 P^*(x)) &= \sum_{x \in S} P^*(x) \cdot \min\left(c_1 \cdot \frac{P(x)}{P^*(x)}, c_2\right) \\
&= \BE_{x \sim \mathcal{P}^*} \left(\BI(x \in S) \cdot \min\left(c_1 \cdot \frac{P(x)}{P^*(x)}, c_2\right)\right),
\end{align*}
    where $\BI(x \in S)$ is the indicator variable of $x \in S.$ Our rough goal will therefore be to provide estimates of $\frac{P(x)}{P^*(x)}$ for $x$ drawn from $\mathcal{P}^*$ with $x \in S.$

    First, suppose that $P^*(L) \ge \frac{1}{3}.$ Then, we let $S = \{L\}.$ To estimate $\sum_{x \in S} \min (c_1 P(x), c_2 P^*(x)) = \min(c_1 P^*(L), c_2 P(L)),$ we just need to estimate $P(L)$ up to an $\eps$ additive error, since $c_1, c_2 \le 1$ and we already know $P^*(z).$ This, however, can be done with $O(\eps^{-2} \log \eps^{-1})$ samples to $\mathcal{P}$ with failure probability $1 - \eps^{10}$ by a simple Chernoff bound argument. Otherwise, $P^*([z-1]) > \frac{1}{2} P^*(z),$ since if $z > L$ then $P^*([z-1]) \ge \frac{1}{2}$ and for $z = L,$ $P^*(z) < \frac{1}{3} \le \frac{2}{3} \cdot P^*([z]).$ Therefore, we can partition $[z-1]$ into sets $S_1, \dots, S_{k-1}$ such that $\frac{1}{2} P^*(z) \le P^*(S_i) \le P^*(z)$ for all $1 \le i \le k-1,$ using a simple greedy procedure \cite{FalahatgarJOPS15}. We finally let $S_k := \{z\}.$ Thus, $S_1, \dots, S_k$ partition $[z]$ so that $\frac{1}{2k} \le \frac{P^*(S_i)}{P([z])} \le \frac{2}{k}$ for all $1 \le i \le k$.
    
    To estimate $\frac{P(z)}{P^*(z)},$ write 
\[\frac{P(z)}{P^*(z)} = \frac{P([z])}{P^*(z)} \cdot \frac{P(S_j)}{P([z])} \cdot \frac{P(z)}{P(S_j)}\]
    for some $1 \le j \le k$ to be chosen later. Now, let $\mathcal{Q}, \mathcal{Q}^*$ be distributions over $[k]$ so that $Q(i) = \frac{P(S_i)}{P([z])}$ and $Q^*(i) = \frac{P^*(S_i)}{P^*([z])}.$ Then, recalling that $S_k = \{z\}$ and that $P^*([z]) = Q^*([k]) = 1$, we have that
\[\frac{P(z)}{P^*(z)} = \frac{1}{P^*(z)} \cdot P([z]) \cdot Q(j) \cdot \frac{Q(k)}{Q(j)}.\]
    
    We will  assume $\Call{Cond}{}$ access to both $\mathcal{P}$ and $\mathcal{Q}.$ We note that $\Call{Cond}{}$ access to $\mathcal{Q}$ can easily be simulated by $\Call{Cond}{}$ access to $\mathcal{P},$ since we just condition on a subset $T \subset [z]$ which is the union of some $S_i$'s and return which $S_i$ the $\Call{Cond}{}_{\mathcal{P}}(T)$ query outputs an element in. We will show how to output a $1 \pm \delta$ multiplicative approximation to each of $P([z])$, $Q_j$ for some $j$, and $\frac{Q(k)}{Q(j)}$ using a small number of queries. Moreover, our estimates will also be \emph{nearly unbiased}. These algorithms will not work under a few extreme cases, but we will deal with these cases accordingly.
    
\begin{lemma} \label{Est1}
    Fix $\delta \ge \eps.$ Then, there is an algorithm \Call{Est1}{} that uses $O(\eps^{-1} \cdot \log \eps^{-1} \cdot \delta^{-2})$ queries to $\Call{Cond}{}_{\mathcal{P}}$, such that conditioned on some event $E_1$ with $\BP(E_1) \ge 1-\eps^6,$ the following happens:
\begin{enumerate}
    \item If $P([z]) \ge \frac{\eps}{2},$ the algorithm outputs a random variable $\tilde{P}([z])$ such that $\frac{\tilde{P}([z])}{P([z])} \in [1 - \delta, 1 + \delta]$ whenever $E_1$ is true, and $\BE[\tilde{P}([z])|E_1] \in [1-\eps, 1+\eps] \cdot P([z])$.
    \item If $P([z]) < \frac{\eps}{2},$ the algorithm outputs $\tilde{P}([z]) \le \eps$ whenever $E_1$ is true.
\end{enumerate}
\end{lemma}

\begin{proof}
    Our algorithm will be quite straightforward. Namely, we sample $R = O(\eps^{-1} \cdot \log \eps^{-1} \cdot \delta^{-2})$ samples $x_1, \dots, x_R$ from $\mathcal{P}$ and let $\tilde{P}([z])$ denote the fraction of the $x_i$'s such that $x_i \le z.$ We let $E_1$ be the event $\frac{\tilde{P}([z])}{P([z])} \in [1-\delta, 1+\delta]$ in the case $P([z]) \ge \frac{\eps}{2}$ and the event $\tilde{P}([z]) \le \eps$ in the case $P([z]) < \frac{\eps}{2}.$ Since each $x_i$ has a $P([z])$ probability of being in $[z],$ the estimate $\tilde{P}([z])$ has distribution $\frac{1}{R} \cdot \operatorname{Bin}(R, P([z]))$, so $\BE[\tilde{P}([z])] = P([z])$. Moreover, if $P([z]) \ge \frac{\eps}{2},$ we have that $\frac{\tilde{P}([z])}{P([z])} \in [1 - \delta, 1 + \delta]$ with probability at least $1-\eps^6$ by the Chernoff bound. Even if we condition on $E_1$, the expectation of $\tilde{P}([z])$ changes by at most $\eps^6$ which is at most $\eps \cdot P([z])$. Likewise, if $P([z]) \le \frac{\eps}{2},$ then $\tilde{P}([z]) \le \eps$ with probability at least $1-\eps^6$, by the Chernoff bound.
\end{proof}

\begin{lemma} \label{Est2}
    Fix $\delta \ge \eps.$ Then, there is an algorithm \Call{Est2}{} that uses $O(\eps^{-2} \log^5 \eps^{-1} \cdot \delta^{-2})$ queries to $\Call{Cond}{}_{\mathcal{Q}}$, such that conditioned on some event $E_2$ with $\BP(E_2) \ge 1-O(\eps^6),$ the following happens:
\begin{enumerate}
    \item Suppose that $\TV(\mathcal{Q}, \mathcal{U}) < 1 - 3 \eps$, where $\mathcal{U}$ is the uniform distribution over $[k]$. Then, the algorithm outputs some pair $(j, \tilde{Q}(j))$ with $j \in [k]$ such that $Q(j) \in \left[\frac{2}{3} \cdot \frac{\eps}{k}, \frac{3}{2} \cdot \frac{\eps^{-1}}{k}\right]$ and $\frac{\tilde{Q}(j)}{Q(j)} \in [1 - \delta, 1+\delta]$ whenever $E_2$ is true. Moreover, conditioning on any fixed $j$ being returned, $\frac{\BE[\tilde{Q}(j)|E_2]}{Q(j)} \in [1-\eps, 1+\eps].$
    \item If $\TV(\mathcal{Q}, \mathcal{U}) \ge 1 - 3 \eps$, then conditioned on $E_2,$ the algorithm either outputs NULL or a pair $(j, \tilde{Q}(j))$ with $Q(j) \in \left[\frac{2}{3} \cdot \frac{\eps}{k}, \frac{3}{2} \cdot \frac{\eps^{-1}}{k}\right]$, $\frac{\tilde{Q}(j)}{Q(j)} \in [1 - \delta, 1+\delta]$, and conditioned on any fixed $j$ being returned, $\frac{\BE[\tilde{Q}(j)|E_2]}{Q(j)} \in [1-\eps, 1+\eps]$.
\end{enumerate}
\end{lemma}

\begin{proof}
    First, by running the procedure of Lemma \ref{ConstantApprox} with $\Call{Pcond}{}_\mathcal{Q}$ and $\Call{Samp}{}_{\mathcal{Q}}$ access, we find some $(j, \hat{Q}(j))$ with $\hat{Q}(j) \in \left[\frac{0.9 \eps}{k}, \frac{1.1 \eps^{-1}}{k}\right]$ and $\frac{\hat{Q}(j)}{Q(j)} \in [0.9, 1.1]$ if $\TV(\mathcal{Q}, \mathcal{U}) \le 1-3 \eps.$ If we don't find such a pair (i.e., Lemma \ref{ConstantApprox} returns $S = \{\}$) then we just return NULL.
    
    Else, let $j$ be the first $j$ returned by Lemma \ref{ConstantApprox}. Recall that in Lemma \ref{ConstantApprox}, we proved that $d(j), u(j) \ge \frac{\eps}{300 \log \eps^{-1}},$ where $d(j), u(j)$ are defined in Definition \ref{dandu} (we assume the guarantees of Lemma \ref{ConstantApprox} are met). To approximate $Q(j)$, we modify the approach of Lemma \ref{ConstantApprox}. We choose $R = O(\eps^{-2} \log^3 \eps^{-1} \cdot \delta^{-2})$ and sample $x_1, \dots, x_R \leftarrow \mathcal{Q}$ and $y_1, \dots, y_R \leftarrow \mathcal{U}$. Also, we define $\tilde{u}(j) = \frac{1}{R} \cdot \#\{i \in [R]: \Call{Compare}{}_{\mathcal{Q}}(y_i, j, 0.01) \in [0.98, 1.02]\},$ so $\tilde{u}_j$ has distribution $\frac{1}{R} \cdot \operatorname{Bin}(R, u(j)).$ Next, for each $i$, we create a random variable $X_i$ and run $\Call{Compare}{}_{\mathcal{Q}}(x_i, j, 0.01)$. If the returned value is between $0.98$ and $1.02,$ then we keep calling $\Call{Cond}{}_{\mathcal{Q}}(\{x_i, j\})$ until $x_i$ is returned, and define $X_i$ to be the number of times we see $j$ returned before the first time $x_i$ is returned. Otherwise, $X_i = 0.$ Finally, we let $\tilde{q}(j)$ be the average of $\min(X_1, C \log \eps^{-1}), \dots, \min(X_R, C \log \eps^{-1})$ for some sufficiently large constant $C$. In our implementation, for each $1 \le i \le R,$ we will stop calling $\Call{Cond}{\{x_i, j\}}$ once we have already made $C \log \eps^{-1}$ calls to $\Call{Cond}{\{x_i, j\}}$.
    
    If we define $p(x, j)$ to be the probability that $\Call{Compare}{}_{\mathcal{Q}}(x, j, 0.01) \in [0.98, 1.02],$ then
\[\BE[X_i] = \sum_{x = 1}^{k} Q(x) \cdot p(x, j) \cdot \frac{Q(j)}{Q(x)} = Q(j) \cdot \sum_{x = 1}^{k} p(x, j) = k \cdot Q(j) \cdot \sum_{x = 1}^{k} \frac{p(x, j)}{k} = k \cdot Q(j) \cdot u(j),\]
    To see why, recall that to create the variable $X_i,$ we first sample $x_i \leftarrow \mathcal{Q}$, then $X_i$ is only nonzero if $\Call{Compare}{}_{\mathcal{Q}}(x_i, j, 0.01) \in [0.98, 1.02],$ in which case $X_i$ is a Geometric random variable with parameter $p = \frac{Q(x_i)}{Q(x_i)+Q(j)}$ and thus has mean $\frac{Q(j)}{Q(x_i)}.$ The last equality is true since $u(j)$ is just the probability that $\Call{Compare}{}_{\mathcal{Q}}(x, j, 0.01) \in [0.98, 1.02],$ where $x$ is now uniformly distributed. Therefore, $\BE[\min(X_i, C \log \eps^{-1})] \le k \cdot Q(j) \cdot u(j).$ However, for any $x$ with $\frac{Q(x)}{Q(j)} \in [0.97, 1.03],$ if we choose $C$ large enough, then $\BE\left[\min\left(\operatorname{Geom}\left(\frac{Q(x_i)}{Q(x_i)+Q(j)}\right), C \log \eps^{-1}\right)\right] = (1 \pm O(\eps^6)) \cdot \frac{Q(j)}{Q(x_i)}$. Thus,
\begin{align*}
    \BE[\min(X_i, C \log \eps^{-1})] &\ge (1 - O(\eps^6)) \cdot \sum_{x: Q(x) \in [0.97, 1.03] \cdot Q(j)} Q(x) \cdot p(x, j) \cdot \frac{Q(j)}{Q(x)} \\
    &= (1 - O(\eps^6)) \cdot k \cdot Q(j) \cdot \left[u(j) - \sum_{x: Q(x) \not\in [0.97, 1.03] \cdot Q(j)} \frac{p(x,j)}{k}\right] \\
    &\ge (1-O(\eps^6)) \cdot k \cdot Q(j) \cdot [u(j) - O(\eps^{10})],
\end{align*}
    since for $x$ with $\frac{Q(x)}{Q(j)} \not\in [0.97, 1.03],$ the probability of $\Call{Compare}{}_{\mathcal{Q}}(x, j, 0.01) \in [0.98, 1.02]$ is at most $\eps^{10}.$ But since $u(j) \ge \frac{\eps}{300 \log \eps^{-1}}$ and $k \cdot Q(j) \ge \frac{2}{3} \eps,$ we have that $\BE[\max(X_i, C \log \eps^{-1})] = k \cdot Q(j) \cdot u(j) \cdot (1 \pm \frac{\eps}{10}) \ge \frac{\eps^2}{500 \log \eps^{-1}}.$ But since $\max(X_i, C \log \eps^{-1})$ is bounded by $C \log \eps^{-1},$ the Chernoff bound tells us that the average of $O(\eps^{-2} \log^3 \eps^{-1} \cdot \delta^{-2})$ samples $X_i$ will be within a $1 \pm \frac{\delta}{10}$ multiplicative factor of $\BE[X_i]$ with probability at least $\eps^{-6}.$ Thus, $\BE[\tilde{q}(j)] = k \cdot Q(j) \cdot u(j) \cdot (1 \pm \frac{\eps}{10})$, and with probability at least $1-\eps^6,$ $\tilde{q}(j) = k \cdot Q(j) \cdot u(j) \cdot (1 \pm \frac{\delta}{4}).$
    
    We output $\left(j, \frac{\tilde{d}(j)}{k \cdot \tilde{u}(j)}\right),$ unless no $(j, \hat{Q}(j))$ was found by Lemma \ref{ConstantApprox}, in which case we return NULL. We know that since $u(j) \ge \frac{\eps}{300 \log \eps^{-1}}$ and $\tilde{u}_j \sim \frac{1}{R} \cdot \operatorname{Bin}(R, u(j)),$ with probability $1 - \eps^6,$ $\tilde{u}(j) \in (1 \pm \frac{\delta \cdot \sqrt{\eps}}{4}) \cdot u(j).$ Therefore, $\frac{\tilde{w}(j)}{k \cdot \tilde{u}(j)} = Q(j) \cdot \frac{1 \pm \frac{\delta}{4}}{1 \pm \frac{\delta \sqrt{\eps}}{4}} = Q(j) \cdot (1 \pm \delta)$ with probability at least $1-3\eps^6.$ We let $E_2$ be the event that the claims in Lemma \ref{ConstantApprox} are satisfied, as well as that $\tilde{d}(j), \tilde{u}(j)$ are accurate up to a $1 \pm \frac{\delta}{4}$ and $1 \pm \frac{\delta \sqrt{\eps}}{4}$ multiplicative approximation, respectively, if Lemma \ref{ConstantApprox} doesn't return NULL.
    
    To compute $\BE[\tilde{Q}(j)|E_2]$ for a fixed $j$, first note that if $\tilde{u}(j) = (1 + \gamma) u(j)$ for $\gamma \in [-0.5, 0.5],$ then $\frac{1}{\tilde{u}(j)} = \frac{1}{u(j)} \cdot (1 - \gamma + O(\gamma^2)).$ Since $\tilde{u}(j) \in (1 \pm \frac{\delta \cdot \sqrt{\eps}}{4}) \cdot u(j)$ assuming $E_2,$ $\frac{1}{\tilde{u}(j)} = \frac{1}{u(j)} \cdot \left(2 - \frac{\tilde{u}(j)}{u(j)} + O(\eps)\right).$ As $\tilde{u}(j)$ has distribution $\frac{1}{R} \cdot \operatorname{Bin}(R, u(j)),$ and conditioning on an event with $O(\eps^6)$ failure probability will not change $\BE[\tilde{u}(j)]$ by more than $O(\eps^6)$, we have $\BE[\tilde{u}(j)|E_2] = u(j) \cdot (1 \pm \eps).$ Therefore, $\BE\left[\frac{1}{\tilde{u}(j)}\Big\vert E_2\right] = \frac{1}{u(j)} \cdot \left(2 - \frac{\BE[\tilde{u}(j)|E_2]}{u(j)} + O(\eps)\right) = \frac{1 + O(\eps)}{u(j)}.$ Also, since $\tilde{q}(j)$ is bounded by $O(\log \eps^{-1})$, $\BE[\tilde{q}(j)] \ge \Omega(\frac{\eps^{-2}}{\log \eps^{-1}}),$ and $E_2$ occurs with probability $1-O(\eps^6)$, conditioning on $E_2$ marginally affects the expectation of $\tilde{q}(j)$, and we will still have that $\BE[\tilde{q}(j)|E_2] = k \cdot Q(j) \cdot u(j) \cdot (1 \pm O(\eps))$. Thus, conditioned on $E_2$ (and $j$ being returned for some fixed $j$), the expected value of $\frac{\tilde{q}(j)}{k \cdot \tilde{u}(j)}$ is $\BE[\tilde{q}(j)|E_2] \cdot \BE\left[\frac{1}{\tilde{u}(j)}\big\vert E_2\right] \cdot \frac{1}{k} = Q(j) \cdot (1 \pm O(\eps)).$ We can split the expectation into a product because for any fixed $j$, $\tilde{q}(j)$ and $\tilde{u}(j)$ are independent conditioned on $E_2$ and $j$, as we used disjoint samples to compute $\tilde{q}(j)$ and $\tilde{u}(j)$ once we found $j$.
\end{proof}

\begin{lemma} \label{Est3}
    Fix $\delta \ge \eps$ and let $j$ be as returned in Lemma \ref{Est2}. Then, there is an algorithm \Call{Est3}{} that uses an expected $O(\eps^{-2} \log^2 \eps^{-1} \cdot \delta^{-2})$ queries to $\Call{Cond}{}_{\mathcal{Q}}$, such that conditioned on some event $E_3$ with $\BP(E_3) \ge 1-O(\eps^6),$ the following happens:
\begin{enumerate}
    \item If $\frac{Q(k)}{Q(j)} \in [0.1 \eps^2, 10 \eps^{-2}],$ then the algorithm finds an estimator $Y$ of $\frac{Q(k)}{Q(j)}$ such that $Y \in \left[(1-\delta) \cdot \frac{Q(k)}{Q(j)}, (1+\delta) \cdot \frac{Q(k)}{Q(j)}\right].$ Moreover, $\BE[Y|E_3] \in \left[(1-O(\eps)) \cdot \frac{Q(k)}{Q(j)}, (1+O(\eps)) \cdot \frac{Q(k)}{Q(j)}\right].$ 
    \item If $\frac{Q(k)}{Q(j)} \le 0.1 \eps^{2}$, then $Y \le 0.2 \eps^{-2}$, and if $\frac{Q(k)}{Q(j)} \ge 10 \eps^{-2},$ then $Y \ge 5 \eps^{-2}.$
\end{enumerate}
\end{lemma}

\begin{proof}
    First, we sample $R = O(\eps^{-2} \log \eps^{-1})$ queries of $\Call{Cond}{}_{\mathcal{Q}}(\{j, k\}).$ If $\frac{Q(k)}{Q(j)} \in [0.05 \eps^2, 20 \eps^{-2}].$ a simple application of the Chernoff bound tells us that with at least $1-\eps^6$ probability, the sample ratio of the number of times $k$ is returned to the number of times $j$ is returned will be correct up to a factor of $1 \pm 0.1$. Likewise, with $1-\eps^6$ probability, if $\frac{Q(k)}{Q(j)} < 0.05 \eps^2$, the sample ratio will be at most $0.06 \eps^2,$ and if $\frac{Q(k)}{Q(j)} > 20 \eps^{-2},$ the sample ratio will be at least $18 \eps^{-2}.$
    
    Let $E_3'$ be the event that the above sample ratio is sufficiently accurate. We know that $\BP(E_3') \ge 1-O(\eps^6).$ Now, assuming $E_3'$, if our estimate is not in the range $[0.06 \eps^2, 18 \eps^{-2}],$ we output the sample ratio as our estimate $Y$, and we know that the true ratio $\frac{Q(k)}{Q(j)}$ is not in the range $[0.1 \eps^2, 10 \eps^{-2}].$ Otherwise, we know that our sample ratio, which we will call $\alpha,$ is accurate up to a factor of $1 \pm 0.1,$ and that $\frac{Q(k)}{Q(j)} \in [0.05 \eps^2, 20 \eps^{-2}].$
    
    Now, consider the following algorithm. If $\alpha \ge 1,$ we create random variables $X_1, \dots, X_T$ where $T = O(\log^2 \eps^{-1} \cdot \delta^{-2}).$ For each $1 \le t \le T,$ we create $X_t$ by sampling from $\Call{Cond}{}_{\mathcal{Q}}(\{j, k\})$ until $k$ is returned, and letting $X_t$ be the number of times we saw $j$ returned before $k$ was returned. We know this is a Geometric random variable with parameter $p = \frac{Q(j)}{Q(j)+Q(k)}$ and thus has mean $\frac{Q(k)}{Q(j)}$. We will truncate this random variable at $O(\alpha \cdot \log \eps^{-1}),$ i.e., we will really let $X_t = \min\left(C \alpha \cdot \log \eps^{-1}, \operatorname{Geom}\left(\frac{Q(j)}{Q(j)+Q(k)}\right)\right)$ for some large constant $C$. That way, since $\alpha$ is within a $1 \pm 0.1$ factor of $\frac{Q(k)}{Q(j)}$ we will still have $\BE[X_t] = \frac{Q(k)}{Q(j)} \cdot (1 \pm \eps),$ conditioned on $E_3'$. However, the Chernoff bound tells us that $Y := \frac{1}{T} (X_1+X_2+\cdots+X_T) \in [1-\frac{\delta}{2}, 1+\frac{\delta}{2}] \cdot \BE[X_t]$ with probability at least $1-\eps^6,$ conditioned on $E_3.$ Moreover, the number of calls to $\Call{Cond}{}_{\mathcal{Q}}$ in expectation is $O(\alpha \cdot T) = O(\eps^{-2} \log^2 \eps^{-1} \cdot \delta^{-2}).$

    Likewise, if $\alpha < 1,$ we create random variables $X_1, \dots, X_T$ where $T = O(\alpha^{-1} \cdot \log^2 \eps^{-1} \cdot \delta^{-2}).$ For each $1 \le t \le T,$ we create $X_t$ by sampling from $\Call{Cond}{}_{\mathcal{Q}}(\{j, k\})$ until $k$ is returned, and letting $X_t$ be the number of times we saw $k$ returned before $j$ was returned, but we also truncate this random variable at $O(\log \eps^{-1}).$ Thus, $X_t = \min\left(C \cdot \log \eps^{-1}, \operatorname{Geom}\left(\frac{Q(j)}{Q(j)+Q(k)}\right)\right)$ for some large constant $C$. Since $\alpha \le 1$, this means $\frac{Q(k)}{Q(j)} \le 1.2$, so we still have $\BE[X_t] = \frac{Q(k)}{Q(j)} \cdot (1 \pm \eps),$ conditioned on $E_3'$. However, the Chernoff bound tells us that $Y := \frac{1}{T} (X_1+X_2+\cdots+X_T) \in [1-\frac{\delta}{2}, 1+\frac{\delta}{2}] \cdot \BE[X_t]$ with probability at least $1-\eps^6,$ conditioned on $E_3'$. Moreover, the expected number of calls to $\Call{Cond}{}_{\mathcal{Q}}$ is $O(T) = O(\eps^{-2} \log^2 \eps^{-1} \cdot \delta^{-2}).$

    Finally, let $E_3$ be the event that $E_3'$ is true and that if our initial estimate $\alpha$ is in the range $[0.06 \eps^2, 18 \eps^{-2}],$ then $Y \in [1 - \frac{\delta}{2}, 1 + \frac{\delta}{2}] \cdot \BE[X_t|E_3'].$ Clearly, $\BP(E_3) \ge 1-O(\eps^6).$ If we condition on $E_3$ instead of $E_3'$ we trivially achieve all the guarantees by our previous analysis, except possibly the bound on $\BE[Y|E_3].$ However, note that $Y$ is uniformly bounded by $O(\eps^{-2} \log \eps^{-1})$ assuming $E_3'$ and $\BE[Y|E_3'] = \Omega(\eps^2),$ so conditioning on $E_3$ instead of $E_3'$, where $\BP(E_3|E_3') \ge 1-O(\eps^6),$ can only change the expectation of $Y$ by a $1 \pm \eps$ multiplicative error. Thus, $\BE[Y|E_3] = (1 \pm O(\eps)) \cdot \frac{Q(k)}{Q(j)}$.
\end{proof}

Next, we show how to combine Lemmas \ref{Est1}, \ref{Est2}, and \ref{Est3} to get a good estimator for $c_1 \cdot \frac{P(z)}{P^*(z)}.$

\begin{lemma} \label{Est}
    Suppose $z \ge L$ and $\eps \le \delta \le \frac{1}{10}$ are fixed. Then, there is an algorithm \Call{Est}{} using $O(\eps^{-2} \log^5 \eps^{-1} \cdot \delta^2)$ queries to $\Call{Cond}{}_{\mathcal{P}}$ and $\Call{Cond}{}_{\mathcal{Q}}$, such that conditioned on some event $E_4 := E_4(z, \delta)$ with $\BP(E_4) \ge 1-O(\eps^6),$ the following happens.
\begin{enumerate}
    \item If $\sum_{i \le z} \min(c_1 P(i), c_2 P^*(i)) > 6 \eps,$ then assuming $E_4$:
    \begin{enumerate}
        \item If $c_1 \cdot \frac{P(z)}{P^*(z)} < \eps,$ we return an estimator $X := X(z, \delta)$ such that $X \le 2 \eps.$
        \item If $c_1 \cdot \frac{P(z)}{P^*(z)} > \frac{5}{3},$ we return an estimator $X$ such that $X \ge \frac{3}{2}.$
        \item If $\eps \le c_1 \cdot \frac{P(z)}{P^*(z)} \le \frac{5}{3},$ we return an estimator $X$ such that $X \in [1-4\delta, 1+4\delta] \cdot \frac{P(z)}{P^*(z)}$ and $\BE[X|E_4] \in [1-4\eps, 1+4\eps] \cdot \frac{P(z)}{P^*(z)}$.
    \end{enumerate}
    \item If $\sum_{i \le z} \min(c_1 P(i), c_2 P^*(i)) \le 6 \eps,$ then assuming $E_4$, we either will return $([z], 0)$ or return some $X$ with the same guarantees as above.
\end{enumerate}
\end{lemma}

\begin{proof}
    We use Lemma \ref{Est1} to find an estimate $\tilde{P}([z]),$ then Lemma \ref{Est2} to find some pair $(j, \tilde{Q}(j)),$ and then Lemma \ref{Est3} to find $Y$ (assuming Lemma \ref{Est2} did not return NULL). We let $E_4$ indicate that $E_1$ from Lemma \ref{Est1}, $E_2$ from Lemma \ref{Est2}, and $E_3$ from Lemma \ref{Est3} (assuming Lemma \ref{Est2} didn't return NULL) are all true. Clearly, $\BP(E_4) \ge 1-O(\eps^6).$
    
    Now, suppose Lemma \ref{Est1} returns $\tilde{P}([z]) \le \frac{\eps}{c_1}.$ Then, we know that $P([z]) \le 2 \frac{\eps}{c_1},$ which means that 
\[\sum_{i \le z} \min(c_1 P(i), c_2 P^*(i)) \le \sum_{i \le z} c_1 P(i) \le c_1 \cdot 2 \frac{\eps}{c_1} = 2 \eps.\]
    Therefore, we can output $S = [z]$ and our estimate as $0$, i.e., we output $([z], 0).$
    
    Otherwise, suppose that Lemma \ref{Est2} returns NULL. Then, we have that $\TV(\mathcal{Q}, \mathcal{U}) \ge 1-3 \eps,$ so by Proposition \ref{BasicTV}, $\sum_{j = 1}^{k} \min\left(\frac{1}{k}, Q(j)\right) \le 3 \eps.$ Therefore,
\[\sum_{i \le z} \min(c_1 P(i), c_2 P^*(i)) \le \sum_{i \le z} \min(P(i), P^*(i)) \le \sum_{j \le k} \min\left(Q(j), Q^*(j)\right) \le \sum_{j \le k} \min\left(Q(j), \frac{2}{k}\right) \le 6 \eps,\]
    where we used the facts that $c_1, c_2 \le 1,$ $\sum_{i \in S_j} P(i) = \frac{Q(j)}{P([z])} \le Q(j),$ $\sum_{i \in S_j} P^*(i) = \frac{Q^*(j)}{P^*([z])} \le Q^*(j),$ and $Q^*(j) \le \frac{2}{k}$ for all $1 \le j \le k,$ due to how we chose the sets $S_1, \dots, S_{k-1}, S_k.$ Therefore, we can again output $S = [z]$ and our estimate as $0$, i.e., we output $([z], 0)$.
    
    Otherwise, since we conditioned on $E_4 = E_1 \cap E_2 \cap E_3$, we get estimates $\tilde{P}([z]) \in [1-\delta, 1+\delta] \cdot P([z])$ such that $P([z]) \ge \frac{\eps}{2 c_1}$, $\tilde{Q}(j) \in [1-\delta, 1+\delta] \cdot Q(j)$ such that $Q(j) \in \left[\frac{2}{3} \cdot \frac{\eps}{k}, \frac{3}{2} \cdot \frac{\eps^{-1}}{k}\right],$ and $Y \in [1-\delta, 1+\delta] \cdot \frac{Q(k)}{Q(j)},$ unless $\frac{Q(k)}{Q(j)} \not\in \left[0.1 c_1 \eps^2, 10 c_1^{-1} \eps^{-2}\right].$ We let $X := c_1 \cdot \frac{1}{P^*(z)} \cdot \tilde{P}([z]) \cdot \tilde{Q}(j) \cdot Y$ be our estimate for $c_1 \frac{P(z)}{P^*(z)}.$ Note that since $\frac{1}{2} \le P^*([z]) \le 1,$ and by our partitioning of $[z]$ into $S_1, \dots, S_k,$ we have that $\frac{1}{2k} \le \frac{P^*(z)}{P^*([z])} \le \frac{2}{k}.$ Therefore, $\frac{1}{4k} \le P^*(z) \le \frac{2}{k}.$
    
    Now, if $\frac{Q(k)}{Q(j)} < 0.1 \eps^2,$ then 
\[c_1 \frac{P(z)}{P^*(z)} = c_1 \cdot \frac{1}{P^*(z)} \cdot P([z])\cdot Q(j) \cdot \frac{Q(k)}{Q(j)} \le c_1 \cdot 4k \cdot 1 \cdot \frac{3}{2} \cdot \frac{\eps^{-1}}{k} \cdot 0.1 \eps^2 \le 0.6 \eps,\]
    and conditioned on $E_4,$
\[X = c_1 \cdot \frac{1}{P^*(z)} \cdot \tilde{P}([z])\cdot \tilde{Q}(j) \cdot Y \le c_1 \cdot 4k \cdot 1 \cdot (1+\delta) \cdot \frac{3}{2} \cdot \frac{\eps^{-1}}{k} \cdot 0.2 \eps^2 \le 2 \eps,\]
    for $\delta \le \frac{1}{10}.$ Likewise, if $\frac{Q(k)}{Q(j)} > 10 \eps^2,$ then 
\[c_1 \frac{P(z)}{P^*(z)} = c_1 \cdot \frac{1}{P^*(z)} \cdot P([z])\cdot Q(j) \cdot \frac{Q(k)}{Q(j)} \ge c_1 \cdot \frac{k}{2} \cdot \frac{\eps}{2 c_1} \cdot \frac{2}{3} \cdot \frac{\eps}{k} \cdot 10 \eps^{-2} \ge \frac{5}{3},\]
    and conditioned on $E_4,$
\[X = c_1 \cdot \frac{1}{P^*(z)} \cdot \tilde{P}([z])\cdot \tilde{Q}(j) \cdot Y \ge c_1 \cdot \frac{k}{2} \cdot \frac{\eps}{c_1} \cdot (1-\delta) \cdot \frac{2}{3} \cdot \frac{\eps}{k} \cdot 5 \eps^{-2} \ge 1.5,\]
    for $\delta \le \frac{1}{10}.$
    
    Finally, if $\frac{Q(k)}{Q(j)} \in [0.1 \eps^2, 10 \eps^{-2}],$ then conditioned on $E_4,$ $\tilde{P}([z]) \in [1-\delta, 1+\delta] \cdot P([z]),$ $\tilde{Q}(j) \in [1-\delta, 1+\delta] \cdot Q(j)$, and $Y \in [1-\delta, 1+\delta] \cdot \frac{Q(k)}{Q(j)}.$ Moreover, conditioned on $E_1, E_2, E_3,$ we have that $\BE[\tilde{P}(z)] \in [1-\eps, 1+\eps] \cdot P([z]),$ $\BE[\tilde{Q}(j)] \in [1-\eps, 1+\eps] \cdot Q(j)$, and $\BE[Y] \in [1-\eps, 1+\eps] \cdot \frac{Q(k)}{Q(j)}.$ Due to using independent samples, this means that $X \in [1-4 \delta, 1+4 \delta] \cdot c_1 \cdot \frac{P(i)}{P^*(i)},$ and $\BE[X_{z, \delta}|E_1, E_2, E_3] \in [1-4\eps, 1+4\eps] \cdot c_1 \cdot \frac{P(i)}{P^*(i)}.$
    
    Overall, assuming we have not already returned $([z], 0)$ we have that, assuming $E_4,$ if $c_1 \cdot \frac{P(z)}{P^*(z)} \ge \frac{5}{3},$ we will output an estimate $X \ge \frac{3}{2},$ if $c_1 \cdot \frac{P(z)}{P^*(z)} \le \eps,$ then $X \le 2 \eps,$ and otherwise, $X \in [1-4 \delta, 1+4 \delta] \cdot c_1 \cdot \frac{P(z)}{P^*(z)}$ and $\BE[X|E_4] \in [1-4\eps, 1+4\eps] \cdot c_1 \cdot \frac{P(z)}{P^*(z)}.$
\end{proof}

We are now ready to finish the proof of Theorem \ref{PartialDetermining}. The rest of the proof will be similar to that of Lemmas \ref{ZEstimate} and \ref{EstimateCloseTerms}.

\begin{proof}[Proof of Theorem \ref{PartialDetermining}]
    Fix $L \le z \le M$, where $L$ is the smallest integer such that $\BP_{x \sim \mathcal{P}^*} (x \le L) \ge \frac{1}{2}$, and $T$ such that $2^{-T} = c \eps$ for some small constant $c$. We show a method for estimating $c_1 \cdot \frac{P(z)}{P^*(z)}$ and create indicator random variables $\BI_{+, t}(z), \BI_{-, t}(z)$ to go along with this for each $1 \le t \le T-1$, as well as $\BI_T(z)$. Recall that Lemma \ref{Est}, conditioned on some event $E_4 =: E_4(z, \delta)$ occurring with probability $1-O(\eps^6),$ either outputs an estimate $X(z, \delta))$ that is a $1 \pm \delta$ multiplicative approximation of $c_1 \frac{P(z)}{P^*(z)},$ or returns $([z], 0),$ meaning that $\sum_{i \le z} \min(c_1 P(i), c_2 P^*(i)) \le 6 \eps$. We set $\delta_i = \frac{1}{20} \cdot 2^{-i}$, beginning with $i = 1.$ At any step, we compute $X(z, \delta_i)$ using Lemma \ref{Est} and check if $X(z, \delta_i) \in [(1-2^{-i}) c_2, (1+2^{-i}) c_2].$ If so, we increment $i$. We repeat this process of incrementing $i$ (i.e., dividing $\delta$ by $2$) and running Lemma \ref{Est} until one of the following three things occurs: either $X(z, \delta_i) \not\in [(1-2^{-i}) c_2, (1+2^{-i}) c_2],$ or $([z], 0)$ is returned instead of an estimate $X$, or $i \ge T$. If $([z], 0)$ is ever returned, we will simply output $S = [z]$ and $\sum_{x \in S} \min(c_1 P(x), c_2 P^*(x)) = 0.$ Else, for some $1 \le t \le T-1,$ if $t$ is the first value such that $X(z, \delta_t) \not\in [(1-2^{-t}) c_2, (1+2^{-t}) c_2]$, then if $X(z, \delta_t) \le (1-2^{-t}) c_2,$ then we set $\BI_{-, t}(z) := 1$ and if $X(z, \delta_t) \ge (1+2^{-t}) c_2,$ then we set $\BI_{+, t}(z) := 1$. However, if we reach $i = T,$ we just set $\BI_T(z) = 1.$ All variables not set to $1$ will be $0$.

    For the rest of the proof, we implicitly condition on the events $E_4 = E_4(z; \delta)$ being true for every call to the algorithm of Lemma \ref{Est}. We will only make $\tilde{O}(\eps^{-2})$ calls to the algorithm \ref{Est}, and as $\BP(E_4) \ge 1-\eps^6,$ the event we are conditioning on will happen with probability $1-O(\eps^3).$
    
    Now, when running the above procedure on some element $z$, we will always have that exactly one of the indicator variables is $1$ (unless $([z], 0)$ is returned). If $c_1 \cdot \frac{P(z)}{P^*(z)} \ge c_2,$ the nonzero indicator is either $\BI_{+, t}(z)$ for some $t$ or $\BI_{T}(z)$. Next, if $c_1 \cdot \frac{P(z)}{P^*(z)} = c_2 (1 - \gamma)$ for some $\gamma \ge \eps,$ then exactly one value $\BI_{-, t}$ will be nonzero for some $t = \log_2 \gamma^{-1} \pm O(1).$ Finally, if $c_1 (1 - \eps) \le c_1 \cdot \frac{P(z)}{P^*(z)} \le c_2,$ either $\BI_T(z)$ or some $\BI_{-, t}(z)$ will be nonzero: in the latter case, $t = \log_2 \eps^{-1} \pm O(1).$ Define $q_{+, t}(z) := \BP(\BI_{+, t}(z)),$ $q_{-, t}(z) := \BP(\BI_{-, t}(z)),$ and $q_T(z) := \BP(\BI_T(z)),$ where we implicitly condition on $E_4$ being true for all calls to the algorithm of Lemma \ref{Est} and that $([z], 0)$ is never returned by the algorithm of Lemma \ref{Est}. Therefore, if $c_1 \cdot \frac{P(z)}{P^*(z)} \le c_2,$ then $\sum_{t = 1}^{T-1} q_{+, t}(z) = 0;$ and if $c_1 \cdot \frac{P(z)}{P^*(z)} \ge c_2,$ then $\sum_{t = 1}^{T-1} q_{-, t}(z) = 0.$
    This means that for any $L \le z \le M,$ 
\begin{align*}
\min\left(c_1 \cdot \frac{P(z)}{P^*(z)}, c_2\right) &= \sum_{t = 1}^{T-1} q_{+, t}(z) \cdot c_2 + q_T(z) \cdot \min\left(c_1 \cdot \frac{P(z)}{P^*(z)}, c_2\right) + \sum_{t = 1}^{T-1} q_{-, t}(z) \cdot \left(c_1 \cdot \frac{P(z)}{P^*(z)}\right) \\
&= c_2 + O(\eps) - \sum_{t = 1}^{t-1} q_{-, t}(z) \cdot \left(c_2 - c_1 \cdot \frac{P(z)}{P^*(z)}\right),
\end{align*}
    where we used the fact that $\sum_{t = 1}^{T-1} q_{+, t}(z) + q_T(z) + \sum_{t = 1}^{T-1} q_{-, t}(z)$, and if $q_T(z) > 0$ then $c_1 \cdot \frac{P(z)}{P^*(z)} = c_2 (1 \pm O(\eps)).$

    Since we want to compute
\[\sum_{z = L}^{M} \min(c_1 P(z), c_2 P^*(z)) = \BE_{z \sim \mathcal{P}^*} \left[\BI(z \ge L) \min\left(c_1 \cdot \frac{P(z)}{P^*(z)}, c_2\right)\right],\]
where $\BI(z \ge L)$ is the indicator function of $z \ge L$, it suffices to approximate
\[\BE_{z \sim \mathcal{P}^*} \left[\BI(z \ge L) \cdot q_{-, t}(z) \cdot \left(c_2 - c_1 \cdot \frac{P(z)}{P^*(z)}\right)\right] = \BE_{z \sim \mathcal{P}^*} \left[\BI(z \ge L) \cdot \BI_{-, t}(z) \cdot \left(c_2 - c_1 \cdot \frac{P(z)}{P^*(z)}\right)\right]\]
    for each $1 \le t \le T-1,$ where the final expectation is over $z$ drawn from $\mathcal{P}^*$ and the randomness in the algorithm determining $\BI_{-, t}(z)$, and is implicitly conditioned on all events $E_4$ being true and $([z], 0)$ never being returned for all calls to Lemma \ref{Est}. To approximate this quantity, we set $\delta := \delta_t = \frac{1}{20} \cdot 2^{-t}$ and we sample $S = O((\delta/\eps)^2 \log \eps^{-1})$ samples $z_1, z_2, \dots, z_S \leftarrow \mathcal{P}^*.$ For each sample $z_s$, we determine $\BI_{-, t}(z_s),$ which can be done using $O(\eps^{-2} \log^5 \eps^{-1} \cdot \delta^{-2})$ calls to $\Call{Cond}{}_{\mathcal{P}}$ (since we run Lemma \ref{Est} for $\delta_1, \delta_2, \dots, \delta_{t+O(1)}$ and then we can stop). If $z_s < L$ or $\BI_{-, t}(z_s) = 0,$ then we can just set some variable $Z_s = 0.$ Otherwise, we again run Lemma \ref{Est} on $z_s$ but with fresh randomness and with $\delta = \delta_t$, which we do to get an estimate $X_s$ such that $X_s \in [1 - 4\delta, 1 + 4\delta] \cdot c_1 \frac{P(z_s)}{P^*(z_s)},$ and we set $Z_s = c_2 - X_s.$ Importantly, since we are using fresh randomness, if $z_s \ge L,$ then $\frac{P(z_s)}{P^*(z_s)} \le 1$ and $\BE[X_s|z_s, \BI_{-, t}(z_s) = 1] \in [1-4\eps, 1+4\eps] \cdot c_1 \cdot \frac{P(z_s)}{P^*(z_s)},$ unless $\frac{P(z_s)}{P^*(z_s)} < \eps,$ in which case we have $X = c_1 \cdot \frac{P(z_s)}{P^*(z_s)} \pm O(\eps)$ uniformly, since $c_1 \cdot \frac{P(z_s)}{P^*(z_s)} \le \eps$ and $X \le 2 \eps$ always by Condition 1a) of Lemma \ref{Est}. Therefore, $\BE[Z_s|z_s, \BI_{-, t}(z_s) = 1] = c_2 - c_1 \cdot \frac{P(z_s)}{P^*(z_s)} + O(\eps).$ Moreover, since $c_1 \frac{P(z_s)}{P^*(z_s)} = c_2 \cdot (1-O(\delta)),$ we always have that $Z_s = O(\delta)$ (even if $z_s < L$ or $\BI_{-, t}(z_s) = 0$). Therefore,
\[\BE_{z \sim \mathcal{P}^*} \left[\BI(z \ge L) \cdot \BI_{-, t}(z) \cdot \left(c_2 - c_1 \cdot \frac{P(z)}{P^*(z)}\right)\right] = \BE[Z_s] + O(\eps)\]
    and $Z_s$ is uniformly bounded in magnitude by $O(\delta)$. Therefore, by sampling $S = O((\delta/\eps)^2 \log \eps^{-1})$ samples of $z_1, \dots, z_S$ and computing $Z_s$ for each $1 \le s \le S$, with probability at least $1-\eps^2,$ we will have that the average of the $Z_s$'s is within $O(\eps)$ of $\BE_{z \sim \mathcal{P}^*} \left[\BI(z \ge L) \cdot \BI_{-, t}(z) \cdot \left(c_2 - c_1 \cdot \frac{P(z)}{P^*(z)}\right)\right].$ (The exception is if Lemma \ref{Est} ever returns $([z_s], 0)$ for some sampled $z_s,$ but by our assumption that $E_4(z_s, \delta)$ is always true, in this case we can instead return $([z_s], 0)$.) Therefore, the overall error will be $O(\eps \log \eps^{-1})$, since we need to compute this for all $1 \le t \le T-1 = O(\log \eps^{-1})$. Moreover, each $t$ will require $O((\delta/\eps)^2 \cdot \eps^{-2} \log^5 \eps^{-1} \cdot \delta^{-2}) = O(\eps^{-4} \log^5 \eps^{-1})$ queries to $\Call{Cond}{}_{\mathcal{P}},$ so the total number of queries to $\Call{Cond}{}_{\mathcal{P}}$ is $O(\eps^{-4} \log^6 \eps^{-1})$.
\end{proof}

We are now ready to prove Theorem \ref{TolerantId}.

\begin{proof}[Proof of Theorem \ref{TolerantId}]
    It suffices to show that we can estimate $\TV(\mathcal{D}, \mathcal{D}^*)$ up to an $O(\eps^{-1} \log^2 \eps^{-1})$ additive error using $\tilde{O}(\eps^{-4})$ queries, since we can then replace $\eps$ with $\eps' = \frac{\eps}{\log^2 \eps^{-1}}$.
    
    Let $T_0 \supset T_1 \supset T_2 \supset \cdots \supset T_r$ be subsets of $[N]$ with $T_0 = [N]$ - we will decide the remaining sets later. Also, let $\mathcal{P}_k$ be the distribution $\mathcal{D}$ conditioned on being in $T_k,$ and let $\mathcal{P}_k^*$ be the distribution $\mathcal{D}^*$ conditioned on being in $T_k$. Also, let $c_{1, k} = D(T_k)$ and $c_{2, k} = D^*(T_k).$ Importantly, note that $\mathcal{P}_0 = \mathcal{D},$ $\mathcal{P}_0^* = \mathcal{D}^*,$ and for any $i \in T_k$, $c_{1, k} P_k(i) = D(i)$ and $c_{2, k} P_k^*(i) = D^*(i).$
    
    Now, our algorithm proceeds as follows. Suppose we have determined $T_k$. Then, we can compute $c_{2, k}$ as $\sum_{i \in T_k} D^*(i),$ and we can estimate $c_{1, k}$ up to an additive $\eps$ error using $\tilde{O}(\eps^{-2})$ samples from the distribution $\mathcal{D}$ and determining what fraction of the samples are in $T_k$.
    
    If either our estimate $\tilde{c}_{1, k}$ (for $c_{1, k}$) or $c_{2, k}$ is less than $2 \eps,$ then either $c_{1, k}$ or $c_{2, k}$ is at most $3 \eps,$ so
\[\sum_{i \in T_k} \min(D(i), D^*(i)) \le 3 \eps.\]
    In this case, set $S = T_k.$ Else, $c_{1, k}, \tilde{c}_{1, k}, c_{2, k} \ge \eps,$ so we can find a set $S \subset T_k$ such that $D^*(S) \ge \frac{D^*(T_k)}{3}$ and determine
\[\sum_{i \in S} \min(\tilde{c}_{1, k} P_k(i), c_{2, k} P_k^*(i))\]
    up to an additive $O(\eps \log \eps^{-1})$ error by Theorem \ref{PartialDetermining}, since we can conditionally sample subsets in $T_k$. But noting that $\sum_{i \in S} P_k(i), \sum_{i \in S} P_k^*(i) \le 1$ and that $|\tilde{c}_{1, k} - c_{1, k}| \le \eps,$ we therefore have that our estimate is an $O(\eps \log \eps^{-1})$ additive error approximation to
\[\sum_{i \in S} \min(c_{1, k} P_k(i), c_{2, k} P_k^*(i)) = \sum_{i \in S} \min(D(i), D^*(i)).\]
    
    We now let $T_{k+1} = T_k \backslash S$ for each $k$. At each stage, we can determine $\sum_{i \in T_{k+1} \backslash T_k} \min(D(i), D^*(i))$ up to an additive $O(\eps \log \eps^{-1})$ error. Moreover, $D^*(T_{k+1}) \le \frac{2}{3} D^*(T_k)$, so this process will only continue $O(\log \eps^{-1})$ times until we reach some $r$ such that either $D(T_r) \le \eps$ or $D^*(T_r) \le O(\eps),$ in which case we can just estimate $\sum_{i \in T_r} \min(D(i), D^*(i))$ as $0$. Therefore, by adding our estimates for $\sum_{i \in T_k \backslash T_{k+1}} \min(\tilde{c}_{1, k} P_k(i), \tilde{c}_{2, k} P_k^*(i))$ for all $0 \le k \le r$ and $T_{r+1} = \emptyset,$ we get an $O(\eps \log^2 \eps^{-1})$ additive error for $\sum_{i = 1}^{N} \min(D(i), D^*(i)),$ which equals $1 - \TV(\mathcal{D}, \mathcal{D}^*)$ by Proposition \ref{BasicTV}.
\end{proof}

\section{An $\tilde{O}(\eps^{-4})$-Query Algorithm for Monotonicity Testing} \label{Monotonicity}

In this section, we improve upon Canonne's $\tilde{O}(\eps^{-22})$-query algorithm \cite{Canonne15} by showing a $\tilde{O}(\eps^{-4})$-query algorithm. The improved algorithm follows Canonne's overall structure, but uses our improved tolerant uniformity algorithm and some additional observations. Recall that $\mathcal{D}$ over $[N]$ is defined to be \emph{monotone} if $D(1) \ge D(2) \ge \cdots \ge D(N)$, and that we are trying to distinguish between $\mathcal{D}$ being monotone or $\eps$-far from all monotone distributions in Total Variation Distance..

\subsection{Preliminaries}

In this section, we explain the definitions and prior results used by Canonne \cite{Canonne15}, as well as the results that Canonne proved that we can use as a black box.

First, we need to describe the oblivious decomposition, as well as the reduced and flattened distributions.

\begin{defn} \label{oblivious}
    Given some parameter $\alpha \ge \frac{1}{\sqrt{N}}$ and a domain size $N$, the \emph{oblivious decomposition} of $[N]$ is the partition $\mathcal{I}_\alpha = (I_1, I_2, \dots, I_\ell),$ where $\ell = \Theta\left(\frac{\ln (\eps N + 1)}{\eps}\right) = \Theta\left(\frac{\log N}{\eps}\right)$, $|I_k| = \lfloor (1+\eps)^k \rfloor$ for $1 \le k \le \ell.$ This decomposition is called \emph{oblivious} because it only depends on $N$, not on any distribution $\mathcal{D}$.
\end{defn}

\begin{defn} \label{redflat}
    Recall the notation in Definition \ref{oblivious}. Now, given a distribution $\mathcal{D}$ and parameter $\alpha,$ define $\mathcal{D}_{\alpha}^{\text{red}}$ to be the \emph{reduced distribution} on $[\ell]$, meaning that for all $1 \le k \le \ell$, $\mathcal{D}_{\alpha}^{\text{red}}(k) = D(I_k)$. Define $\Phi_{\alpha}(\mathcal{D})$ to be the \emph{flattened distribution} with relation to $\mathcal{I}_{\alpha}$, meaning that for all $1 \le k \le \ell$ and for all $i \in I_k,$ $\Phi_{\alpha}(\mathcal{D})(i) = \frac{D(I_k)}{|I_k|}.$
\end{defn}

For the rest of this section, we implicitly define $I_k$ and $\ell$ depending on $N, \alpha$.

Note that given conditional sample access to $\mathcal{D},$ it is easy to simulate conditional samples from $\mathcal{D}_{\alpha}^{\text{red}}$. If we wish to sample $k \leftarrow \mathcal{D}_{\alpha}^{\text{red}}(T),$ we just sample $i \leftarrow \mathcal{D}(S)$ for $S = \bigcup_{x \in I_k: k \in T} x,$ and then find the $k$ such that $i \in I_k.$ We can also sample from $\Phi_{\alpha}(D)$ by sampling $k \leftarrow \mathcal{D}_{\alpha}^{\text{red}}$ and then uniformly sampling from $I_k$: we will not need to conditionally sample from $\Phi_{\alpha}(D)$.

We also have the following theorem due to Birg\'{e}.

\begin{theorem} \cite{Birge87} \label{Birge}
    If $\mathcal{D}$ is monotone, then $\TV(\mathcal{D}, \Phi_{\alpha}(\mathcal{D})) \le \alpha.$
\end{theorem}


We next define the ``exponential property'' of a distribution.

\begin{defn} \label{exponential}
    For fixed $N, \alpha, \ell$ as in Definition \ref{oblivious}, we say that a distribution $\mathcal{Q}$ over $[\ell]$ satisfies the \emph{exponential property} if and only if $\mathcal{Q}(k+1) \le (1+\alpha) Q(k)$ for all $1 \le k < \ell.$
\end{defn}

From now on, we write $\mathcal{M}$ as the set of monotone distributions over $[N],$ and $\mathcal{P}_{\alpha}$ as the set of distributions over $[\ell]$ with the exponential property. We also write $\TV(\mathcal{D}, \mathcal{M}) := \min_{\mathcal{D}' \in \mathcal{M}} \TV(\mathcal{D}, \mathcal{D}'),$ and likewise, for $\mathcal{Q}$ a distribution over $[\ell],$ we write $\TV(\mathcal{Q}, \mathcal{P}_{\alpha}) := \min_{\mathcal{Q}' \in \mathcal{P}_{\alpha}} \TV(\mathcal{Q}, \mathcal{Q}').$

We have the following two facts, noted by Canonne \cite{Canonne15}.

\begin{proposition} \label{MonoProp1}
    If $\mathcal{D}$ is monotone over $[N],$ then $\mathcal{D}_{\alpha}^{\text{red}}$ satisfies the exponential property.
\end{proposition}

\begin{proposition} \label{MonoProp2}
    For any distribution $\mathcal{D},$ $\TV(\Phi_{\alpha}(\mathcal{D}), \mathcal{M}) = \TV(\mathcal{D}_{\alpha}^{\text{red}}, \mathcal{P}_{\alpha}).$
\end{proposition}

Finally, we note the following folklore result about nonnegative random variables.

\begin{proposition} \label{ReverseMarkov}
    Suppose $X$ is a nonnegative random variable bounded by $1$, such that $\BE X \ge \eps.$ Then, there exist $\alpha = 2^{-a}, \beta = 2^{-b}$ with $a, b$ nonnegative integers such that $\alpha \ge \Theta(\eps)$, $\alpha \cdot \beta \ge \Theta(\eps/\log \eps^{-1}),$ and $\BP(X \ge \alpha) \ge \beta$.
\end{proposition}

\subsection{The Algorithm}

Canonne \cite{Canonne15} develops the algorithm in two steps, as follows:
\begin{enumerate}
    \item First, he approximates $\TV(\mathcal{D}, \Phi_{\alpha}(\mathcal{D}))$ using $O(\eps^{-2})$ samples and queries to a tolerant uniformity testing algorithm. With the worse bound of $\tilde{O}(\eps^{-20})$ queries for tolerant uniformity, this took $\tilde{O}(\eps^{-22})$ total queries.
    \item Second, he shows how to test whether $\mathcal{D}_{\alpha}^{\text{red}}$ is in $\mathcal{P}_\alpha$ or far from $\mathcal{P}_\alpha$, using $\tilde{O}(\eps^{-8})$ queries.
\end{enumerate} 
These two combined, for $\alpha = \eps/4,$ are sufficient to test monotonicity, as we briefly explain. Since we can do tolerant uniformity testing in $\tilde{O}(\eps^{-2})$ samples, the first step now only takes $\tilde{O}(\eps^{-4})$ queries. For the second step, we make some modifications to do the test in $\tilde{O}(\eps^{-4})$ queries.

Formally, we have the following result, due to Canonne.

\begin{lemma} \cite[Lemma 4.10]{Canonne15} \label{MonoPart1}
    Given \Call{Cond}{} access to $\mathcal{D},$ there is an algorithm \Call{DistToFlat}{} that, for any $0 < \eps, \alpha \le \frac{1}{2}$, makes $O(\eps^{-2})$ samples from $\mathcal{D}$ and $O(\eps^{-2})$ calls to \Call{TolerantUnif}{}$_{\mathcal{D}[I_k]},$ for some choices of $I_k$ and returns an $\eps/4$-additive approximation to $\TV(\mathcal{D}, \Phi_{\alpha}(\mathcal{D}))$ with failure probability at most $\frac{1}{10}.$ By $\mathcal{D}[I_k],$ we mean the distribution $\mathcal{D}$ conditioned in $I_k$, which is a distribution with support size $|I_k| \approx (1 + \alpha)^k$.
\end{lemma}

Moreover, we prove the following lemma, improving on \cite[Theorem 4.15]{Canonne15}.

\begin{lemma} \label{MonoPart2}
    Let $\alpha = \eps/4.$ Then, there is an algorithm \Call{ExpoTester}{} that, given any distribution $\mathcal{Q}$ over $[\ell]$ and access to $\Call{Cond}{}_{\mathcal{Q}}$, if $\mathcal{Q} \in \mathcal{P}_{\alpha}$, the algorithm outputs ACCEPT with probability at least $9/10,$ and if $\TV(\mathcal{Q}, \mathcal{P}_{\alpha}) \ge \eps/4$, the algorithm outputs REJECT with probability at least $9/10.$ Moreover, the algorithm uses $\tilde{O}(\eps^{-4})$ queries.
\end{lemma}

The overall algorithm, given these two pieces, works the same way as in \cite{Canonne15}. First, we estimate $\TV(\mathcal{D}, \Phi_{\alpha}(D))$ up to error $\eps/4$: by Lemma \ref{MonoPart1} and Theorem \ref{TolerantUnif}, this uses $\tilde{O}(\eps^{-4})$ queries since we can simulate $\Call{Cond}{}_{\mathcal{D}[I_k]}$ with calls to $\Call{Cond}{}_{\mathcal{D}}$. Then, run Lemma \ref{MonoPart2} on $\mathcal{Q} = \mathcal{D}_{\alpha}^{\text{red}}$, which also uses $\tilde{O}(\eps^{-4})$ queries. If Lemma \ref{MonoPart1} returns a total variation distance more than $\eps/2$ or if Lemma \ref{MonoPart2} returns REJECT, the overall algorithm outputs REJECT. Otherwise, we output ACCEPT.

To prove that this works, first suppose $\mathcal{D}$ is monotone. Then, by Theorem \ref{Birge}, $\TV(\mathcal{D}, \Phi_{\alpha}(\mathcal{D})) \le \eps/4,$ so our output from Lemma \ref{MonoPart1} is at most $\eps/2$ with probability at least $9/10$. Moreover, since $\mathcal{D}$ is monotone, by Proposition \ref{MonoProp1}, $\mathcal{D}_{\alpha}^{\text{red}} \in \mathcal{P}_{\alpha},$ so Lemma \ref{MonoPart2} will output ACCEPT with probability at least $9/10$. Therefore, the overall algorithm outputs ACCEPT with probability at least $4/5$.

Now, suppose that $\TV(\mathcal{D}, \mathcal{M}) \ge \eps.$ Then, by the Triangle Inequality, either $\TV(\mathcal{D}, \Phi_{\alpha}(\mathcal{D})) > 3 \eps/4$ or $\TV(\Phi_{\alpha}(\mathcal{D}), \mathcal{M}) = \TV(\mathcal{D}_{\alpha}^{\text{red}}, \mathcal{P}_{\alpha}) \ge \eps/4$, where we used Proposition \ref{MonoProp2}. If $\TV(\mathcal{D}, \Phi_{\alpha}(\mathcal{D})) > 3 \eps/4$, then our output from Lemma \ref{MonoPart1} is more than $\eps/2$ with probability at least $9/10$. Also, if $\TV(\mathcal{D}_{\alpha}^{\text{red}}, \mathcal{P}_{\alpha}) \ge \eps/4$, then Lemma \ref{MonoPart2} will output REJECT with probability at least $9/10$. Therefore, the overall algorithm outputs REJECT with probability at least $9/10$.

Thus, assuming Lemma \ref{MonoPart1} and \ref{MonoPart2}, we have the full monotonicity tester. So, we just need to prove Lemma \ref{MonoPart2}.

\subsection{Testing the Exponential Property: Proof of Lemma \ref{MonoPart2}}

First, we prove a lemma which leads to the algorithm. This lemma improves on \cite[Lemma 4.13]{Canonne15}, which is Canonne's main lemma for testing the exponential property.

\begin{lemma}
    Fix $\eps, \alpha \le 1/2$ and $\ell \ge 2.$ Suppose that $\mathcal{Q}$ is a distribution over $[\ell]$ with $\TV(\mathcal{Q}, \mathcal{P}_{\alpha}) \ge \eps.$ Then,
\[\BE_{i \sim \mathcal{Q}} \left[\max\left(0, 1 - (1+\alpha) \cdot \frac{q_{i-1}}{q_i}\right)\right] \ge \frac{\alpha \cdot \eps}{2}.\]
    For $i = 1,$ we define $\max\left(0, 1 - (1+\alpha) \cdot \frac{q_{i-1}}{q_i}\right)$ to be $0$.
\end{lemma}

\begin{proof}
    Let $q_i := Q(i)$ be the probability of selecting $i$ from a sample of $\mathcal{Q}$, and let $S = \{s_1, s_2, \dots, s_r\}$ be the set of right-to-left maxima of $q_i/(1+\alpha)^i.$ By this, we mean that $s \in S$ if and only if $q_s/(1+\alpha)^s > (q_t)/(1+\alpha)^t$ for all $t > s.$ Note that $\ell \in S$ always. Now, for any $i \in [\ell],$ let $q_i' = \max_{s \ge i} q_s/(1+\alpha)^{s-i}$. Note that if $s \in S,$ then $q_s' = q_s,$ and that $q_i' \le (1+\alpha) q_{i-1}'$ for all $2 \le i \le \ell.$
    
    For ease of notation, we define $s_0 = 0$. Now, for any $1 \le j \le r,$ if $s_j - s_{j-1} = 1,$ then $q_i = q_i'$ for all $s_{j-1} + 1 \le i \le s_j.$ Else, let $a = s_{j-1}+1$ and $b = s_j$. Note that $q_i \le q_i'$ for all $a \le i \le b,$ with equality when $i = b.$ We have that
\begin{equation} \label{Eq1} \sum_{i = a+1}^{b} (q_i - (1+\alpha) q_{i-1}) = (q_b - q_a) - \alpha \cdot \sum_{i = a}^{b-1} q_i. \end{equation}
    However, we have that for all $i$ between $a$ and $b$ (inclusive) that $q_i' = q_b/(1+\alpha)^{b-i}.$ Therefore,
\[\sum_{i = a}^{b} (q_i' - q_i) = \sum_{i = a}^{b} q_i' - \sum_{i = a}^b q_i = \sum_{i' = 0}^{\infty} \frac{q_b}{(1+\alpha)^{i'}} - \sum_{i'= 1}^{\infty} \frac{q_a'}{(1+\alpha)^{i'}} - \sum_{i = a}^{b} q_i\]
\begin{equation} \label{Eq2} = q_b \cdot \frac{1+\alpha}{\alpha} - q_a' \cdot \frac{1}{\alpha} - \sum_{i = a}^{b} q_i = \frac{1}{\alpha} \cdot (q_b-q_a') - \sum_{i = a}^{b-1} q_i \le \frac{1}{\alpha} \cdot (q_b-q_a) - \sum_{i = a}^{b-1} q_i. \end{equation}
    where we used the fact that $q_a \le q_a'$. Combining Equations \eqref{Eq1} and \eqref{Eq2}, we have that
\[\alpha \cdot \sum_{i = a}^{b} (q_i' - q_i) \le \sum_{i = a+1}^{b} (q_i-(1+\alpha) q_{i-1}).\]
    By replacing $(q_i-(1+\alpha) q_{i-1})$ with $\max(q_i-(1+\alpha) q_{i-1}, 0),$ and adding this over all $a = s_{j-1}+1, b = s_j$ from $j = 1$ to $r,$ we get that
\[\alpha \cdot \sum_{i = 1}^{\ell} (q_i'-q_i) \le \sum_{i = 2}^{\ell} \max(q_i-(1+\alpha) q_{i-1}, 0) = \BE_{i \sim \mathcal{Q}} \max\left(0, 1 - (1+\alpha) \cdot \frac{q_{i-1}}{q_i}\right).\]

    Now, note that $\sum_{i = 1}^{\ell} q_i = 1,$ so if $\sum_{i = 1}^{\ell} q_i' = 1 + \rho,$ then $\sum_{i = 1}^{\ell} |q_i' - q_i| = \rho$ and $\sum_{i = 1}^{\ell} \left|q_i' - \frac{q_i'}{1+\rho}\right| = \rho.$ Thus, $\sum_{i - 1}^{\ell} \left|q_i - \frac{q_i'}{1+\rho}\right| \le 2 \rho,$ Since $\sum_{i = 1}^{\ell} \frac{q_i'}{1+\rho} = 1$ and $q_i' \le (1+\alpha) q_{i-1}'$ for all $2 \le i \le \ell,$ the distribution with probability $\frac{q_i'}{1+\rho}$ is in $\mathcal{P}_{\alpha}$. Therefore $\TV(\mathcal{Q}, \mathcal{P}_\alpha) \le 2 \rho.$ Summarizing, we have that
\[\BE_{i \sim \mathcal{Q}} \max\left(0, 1 - (1+\alpha) \cdot \frac{q_{i-1}}{q_i}\right) \ge \alpha \cdot \sum_{i = 1}^{\ell} (q_i'-q_i) \ge \alpha \cdot \frac{\TV(\mathcal{Q}, \mathcal{P}_{\alpha})}{2} \ge \frac{\alpha \cdot \eps}{2}. \qedhere\]
\end{proof}

Now, we are ready to prove Lemma \ref{MonoPart2}.

\begin{proof}[Proof of Lemma \ref{MonoPart2}]
    Let $X$ be the random variable representing $\max\left(0, 1 - (1+\alpha) \cdot \frac{q_{i-1}}{q_i}\right)$ (which we say equals $0$ if $i = 1$). Note that $X$ is nonnegative, bounded by $1$, and $\BE[X] \ge \frac{\alpha \cdot \eps}{2}.$ Thus, by Proposition \ref{ReverseMarkov}, there exist $\tau = 2^{-t}, \beta = 2^{-b}$ such that $\tau \ge \Theta(\eps \cdot \alpha) = \Theta(\eps^2),$ $\tau \cdot \beta \ge \Theta(\eps^2/\log \eps^{-1})$, and $\BP(X \ge \tau) \ge \beta.$ But, if $1 - (1+\alpha) \cdot \frac{q_{i-1}}{q_i} \ge \tau,$ then $\frac{q_{i-1}}{q_i} \le \frac{1-\tau}{1+\alpha},$ so $\frac{q_i}{q_{i-1}} \ge \frac{1+\alpha}{1-\tau} \ge 1 + \alpha + \tau.$
    
    The algorithm now works as follows. Recall that $\alpha = \frac{\eps}{4}.$ Now, for all $\tau = 2^{-t}, \beta = 2^{-b}$ with $b, t$ nonnegative integers, such that $\tau \ge \Theta(\eps \cdot \alpha) = \Theta(\eps^2),$ $\tau \cdot \beta \ge \Theta(\eps^2/\log \eps^{-1})$, we sample $R = O(\beta^{-1} \log \eps^{-1})$ random $i_1, i_2, \dots, i_R \leftarrow \mathcal{Q}.$ For each such $i_r$, if $i_r \ge 2$, we run \Call{Compare}{$i_r, i_r-1, \frac{\tau}{3}$}. This procedure uses $O(\tau^{-2} \cdot \log \eps^{-1})$ calls to \Call{Cond}{}, and with failure probability $\eps^{10},$ returns a value strictly more than $1+\alpha + \frac{\tau}{2}$ if $\frac{q_i}{q_{i-1}} \ge 1 + \alpha + \tau$ but strictly less than $1+\alpha + \frac{\tau}{2}$ if $\frac{q_i}{q_{i-1}} \le 1 + \alpha.$ If for any $(\tau, \beta, i_r)$, the \Call{Compare}{} procedure returns a value more than $1 + \alpha + \frac{\tau}{2},$ we output REJECT; else, we output ACCEPT.
    
    To see why this works, first note that we make at most one call to \Call{Compare}{} for each triple $(\tau, \beta, i_r).$ However, there are only $O(\log \eps^{-1})$ choices for $\tau = 2^{-t}, \beta = 2^{-b},$ and $R \le O(\beta^{-1} \log \eps^{-1}) = \tilde{O}(\eps^{-2}),$ so we only call \Call{Compare}{} at most $\tilde{O}(\eps^{-2})$ times. Now, if $\mathcal{Q}$ satisfies the exponential property, then for all $i_r,$ we have that $\frac{q_i}{q_{i-1}} \le 1 + \alpha.$ Therefore, since each call to \Call{Compare}{$i_r, i_r-1, \frac{\tau}{3}$} returns a value less than $1 + \alpha + \frac{\tau}{2}$ with probability at least $1-\eps^{10}$, the algorithm successfully outputs ACCEPT with probability at least $1-O(\eps^7)$. Conversely, if $\TV(\mathcal{Q}, \mathcal{P}_{\alpha}) \ge \eps,$ then for some choice $\tau = 2^{-t}, \beta = 2^{-b},$ we have that $\BP\left(\frac{q_i}{q_{i-1}} \ge 1 + \alpha + \tau\right) \ge \beta.$ Therefore, for that choice of $\tau, \beta,$ since we sample $O(\beta^{-1} \log \eps^{-1})$ values $i_1, \dots, i_R,$ at least one of those $i$ values will satisfy $\frac{q_i}{q_{i-1}} \ge 1 + \alpha + \tau$ with probability at least $0.99.$ But for that value $i = i_r$, we call \Call{Compare}{$i_r, i_r-1, \frac{\tau}{3}$} and get a value that is more than $1 + \alpha + \frac{\tau}{2}$ with probability at least $1-\eps^{10}.$ Thus, with probability at least $0.99 - \eps^{10} \ge 9/10,$ the algorithm successfully outputs REJECT.
    
    Finally, the total number of queries to \Call{Cond}{} per choice $\tau, \beta$ is $O(\tau^{-2} \cdot \log \eps^{-1}) \cdot O(\beta^{-1} \cdot \log \eps^{-1}) \ge O(\eps^{-4} \cdot \log^4 \eps^{-1}),$ since $\tau \cdot \beta \ge \eps^2/\log \eps^{-1}.$ But there are only $O(\log^2 \eps^{-1})$ choices for $\tau, \beta,$ so in total, we make $O(\eps^{-4} \cdot \log^6 \eps^{-1}) = \tilde{O}(\eps^{-4})$ queries to \Call{Cond}{}.
\end{proof}

\section{Identity Testing in the PAIRCOND model} \label{PcondIdentity}

In this section, we prove Theorem \ref{IdPcondUpper}. In other words, we show that in the PAIRCOND model, we can test whether an unknown distribution $\mathcal{D}$ equals $\mathcal{D}^*$ or is $\eps$-far in Total Variation Distance from $\mathcal{D}^*$, using $\tilde{O}\left(\frac{\sqrt{\log N}}{\eps^2}\right)$ samples.

We first prove the following lemma.

\begin{lemma} \label{ExpectationComputation}
Suppose that $\mathcal{P}$ is an unknown probability distribution over $[m]$ with probabilities $P(1), \dots, P(m),$ and suppose $\mathcal{P}^*$ is a known probability distribution with probabilities $P^*(1), \dots, P^*(m),$ with $P^*(i) \in \left[\frac{1}{2m}, \frac{2}{m}\right]$. Finally, suppose that $\TV(\mathcal{P}, \mathcal{P}^*) \ge \eps.$ Then, we have that if $i$ is drawn from $\mathcal{P}$ and $j$ is drawn uniformly from $\mathcal{U},$ the uniform distribution over $[m]$, then
\[\mathop{\BE}_{\substack{i \sim \mathcal{P} \\ j \sim \mathcal{U}}} \left|\frac{P(i)}{P(i)+P(j)} - \frac{P^*(i)}{P^*(i)+P^*(j)}\right| \ge \frac{\eps}{16}.\]
\end{lemma}

\begin{proof}
    First, let $r_i = \frac{P(i)}{P^*(i)}.$ Then, note that
\begin{align*}
    P(i) \cdot \left|\frac{P(i)}{P(i)+P(j)} - \frac{P^*(i)}{P^*(i)+P^*(j)}\right| &= \frac{P(i)}{P(i)+P(j)} \cdot \left|P(i) - \frac{P^*(i) (P(i)+P(j))}{P^*(i)+P^*(j)}\right| \\
    &= \frac{P(i)}{P(i)+P(j)} \cdot \frac{|P(i) P^*(j) - P^*(i) P(j)|}{P^*(i)+P^*(j)} \\
    &= \frac{P(i)}{P(i)+P(j)} \cdot \frac{P^*(i) P^*(j) \cdot |r_i-r_j|}{P^*(i)+P^*(j)} \\
    &\ge \frac{1}{4m} \cdot \frac{P(i)}{P(i)+P(j)} \cdot |r_i-r_j|,
\end{align*}
    since $\frac{P^*(i) P^*(j)}{P^*(i)+P^*(j)} \ge \frac{1}{4m}$.
    
    Therefore, 
\begin{align*}
    \mathop{\BE}_{\substack{i \sim \mathcal{P} \\ j \sim \mathcal{U}}} \left|\frac{P(i)}{P(i)+P(j)} - \frac{P^*(i)}{P^*(i)+P^*(j)}\right| &= \sum\limits_{i, j = 1}^{m} P(i) \cdot \frac{1}{m} \cdot \left|\frac{P(i)}{P(i)+P(j)} - \frac{P^*(i)}{P^*(i)+P^*(j)}\right| \\
    &\ge \frac{1}{4m^2} \cdot \sum_{i, j = 1}^{m} \frac{P(i)}{P(i)+P(j)} \cdot |r_i-r_j| \\
    &= \frac{1}{4m^2} \cdot \sum\limits_{1 \le i < j \le m} |r_i-r_j| \\
    &= \frac{1}{8m^2} \cdot \sum_{i, j = 1}^{m} |r_i-r_j|.
\end{align*}

    Next, since $\sum P(i) = \sum P^*(i) = 1$ but $\sum |P(i)-P^*(i)| \ge 2\eps$ since $\TV(\mathcal{P}, \mathcal{P}^*) \ge \eps,$ we have that $\sum_{i: r_i \ge 1} (r_i-1) P^*(i) = \sum_{i: r_i \le 1} (1-r_i) P^*(i) \ge \eps.$ But since $P^*(i) \le \frac{2}{m}$ for all $i$, this means that $\sum_{i: r_i \ge 1} (r_i-1) \ge m \cdot \frac{\eps}{2}$ and $\sum_{i: r_i \le 1} (1-r_i) \ge m \cdot \frac{\eps}{2}.$ Thus, for any $r_i,$ if $r_i \ge 1$ then
\[\sum_{j = 1}^{m} |r_i-r_j| \ge \sum_{r_j \le 1} (r_i-r_j) \ge \sum_{r_j \le 1} (1-r_j) \ge \frac{\eps}{2} \cdot m,\]
    and if $r_i \le 1$ then 
\[\sum_{j = 1}^{m} |r_i-r_j| \ge \sum_{r_j \ge 1} (r_j - r_i) \ge \sum_{r_j \ge 1} (r_j - 1) \ge \frac{\eps}{2} \cdot m.\]
    Therefore,
\[\sum_{i, j = 1}^{m} |r_i-r_j| \ge m \cdot \frac{\eps}{2} \cdot m \ge \frac{\eps}{2} \cdot m^2,\]
    so we get the final bound of $\frac{1}{8m^2} \cdot \frac{\eps}{2} m^2 = \frac{\eps}{16}$.
\end{proof}

    Recall we are trying to determine if $\mathcal{D} = \mathcal{D}^*$ or $\TV(\mathcal{D}, \mathcal{D}^*) \ge \eps.$ The algorithm proceeds as follows. First, we split $[N]$ into sets $S_1, S_2, \dots, S_{\log (10 N/\eps)}$ where $i \in S_k$ if and only if $2^{-k} < D_i^* \le 2 \cdot 2^{-k}.$ Let $K = \log \frac{10 N}{\eps}$, and define $S_{K + 1} := [n=N] \backslash \left(\bigcup_{i = 1}^{K} S_i\right)$, so that $S_1, \dots, S_{K+1}$ partition $[N]$. We define $\mathcal{S}$ as the distribution over $[K + 1]$ with $\BP_{x \sim \mathcal{S}}(x = k) := D(S_k) = \BP_{x \sim \mathcal{D}}(x \in S_k).$ Likewise, we define $\mathcal{S}^*$ as the distribution over $[K + 1]$ with $\BP_{x \sim \mathcal{S}^*}(x = k) := D^*(S_k) = \BP_{x \sim \mathcal{D}^*}(x \in S_k).$ Also, for all $1 \le k \le K+1,$ we define $\mathcal{P}_k$ as the conditional distribution of $x \sim \mathcal{D}$ conditioned on $x \in S_k,$ and $\mathcal{P}_k^*$ as the conditional distribution of $x \sim \mathcal{D}^*$ conditioned on $x \in S_k.$ We finally define $s_k = \BP_{x \sim \mathcal{S}}(x = k)$, $s_k^* = \BP_{x \sim \mathcal{S}^*}(x = k)$, $P_{k}(i) = \BP_{x \sim \mathcal{P}_k}(x = i),$ and $P_{k}^*(i) = \BP_{x \sim \mathcal{P}_k^*}(x = i).$ (Note: we use $s_k$ and $s_k^*$ instead of $S_k$ and $S_k^*$ to avoid confusion with the sets $S_k$.)
    
    First, we wish to compare the distributions $\mathcal{S}$ and $\mathcal{S}^*.$ Since $\mathcal{S}^*$ is a known distribution, we can test whether $\mathcal{S} = \mathcal{S}^*$ or $d_{TV}(\mathcal{S}, \mathcal{S}^*) \ge \frac{\eps}{10}$ using $O\left(\frac{\sqrt{\log (N/\eps)}}{\eps^2}\right)$ samples \cite{ValiantV17}, since the support sizes of $\mathcal{S}$ and $\mathcal{S}^*$ equal $K = O(\log (N/\eps))$. However, since a sample from $\mathcal{S}$ can be simulated by a single call to \Call{Samp}{} for $\mathcal{D},$ we only need $O\left(\frac{\sqrt{\log (N/\eps)}}{\eps^2}\right)$ calls to \Call{Samp}{}. If the test tells us that $d_{TV}(\mathcal{S}, \mathcal{S}^*) \ge \frac{\eps}{10}$ then we know $\mathcal{S} \neq \mathcal{S}^*$, so $\mathcal{D} \neq \mathcal{D}^*$ and we can output NO. Otherwise, we know that $d_{TV}(\mathcal{S}, \mathcal{S}^*) \le \frac{\eps}{10}.$
    
    Assuming we haven't yet output NO, we let $\mathcal{D}'$ be the distribution where we draw $k$ according to $\mathcal{S},$ and then draw $i$ according to $\mathcal{P}_k^*$. It is simple to see that $d_{TV}(\mathcal{D}^*, \mathcal{D}') = d_{TV}(\mathcal{S}^*, \mathcal{S}) \le \frac{\eps}{10}.$
    
    Now, suppose $X$ is a random variable constructed as follows. First, draw $k \sim \mathcal{S}$, and draw $i \sim \mathcal{P}_k$ and $j$ uniformly from $S_k$. Then, $X$ is defined as $\left|\frac{D(i)}{D(i)+D(j)} - \frac{D^*(i)}{D^*(i)+D^*(j)}\right|.$ Note that $0 \le X \le 1.$ Moreover, 
\[\BE X = \mathop{\BE}_{k \sim \mathcal{S}} \mathop{\BE}_{\substack{i \sim \mathcal{P}_k \\j \sim Unif[S_k]}} \left|\frac{D(i)}{D(i)+D(j)} - \frac{D^*(i)}{D^*(i)+D^*(j)}\right|.\]

    We next note the following proposition.
    
\begin{proposition} \label{ExpectationLowerBound}
    We have that $\BE X \ge \frac{1}{16} \cdot \left(d_{TV}(\mathcal{D}', \mathcal{D}) - \frac{3 \eps}{10}\right)$.
\end{proposition}

\begin{proof}
    First, we condition on $k$ for $1 \le k \le K$. In this case, note that for any $i, j \in S_k$, $\frac{D^*(i)}{D^*(i)+D^*(j)} = \frac{P_{k}^*(i)}{P_{k}^*(i)+P_{k}^*(j)}$ and $\frac{D(i)}{D(i)+D(j)} = \frac{P_{k}(i)}{P_{k}(i)+P_{k}(j)}$. Moreover, since $2^{-k} \le D^*(i) \le 2 \cdot 2^{-k}$ for all $i \in S_k,$ we have that $P_{k}^*(i) \in \left[\frac{1}{2 |S_k|}, \frac{2}{|S_k|}\right]$ for all $i \in S_k$. Therefore, by Lemma \ref{ExpectationComputation}, we have that 
\[\mathop{\BE}_{\substack{i \sim \mathcal{P}_k \\j \sim Unif[S_k]}} \left|\frac{D(i)}{D(i)+D(j)} - \frac{D^*(i)}{D^*(i)+D^*(j)}\right| \ge \frac{1}{16} \cdot d_{TV}(\mathcal{P}_k^*, \mathcal{P}_k)\]
    for some constant $c > 0$. Taking the expected value over $k$, we get that
\[\BE X \ge \sum\limits_{k = 1}^{K} \frac{1}{16} \cdot s_k \cdot d_{TV}(\mathcal{P}_k^*, P_k) \ge \frac{1}{16} \cdot \left(\sum\limits_{k = 1}^{K + 1} s_k \cdot d_{TV}(\mathcal{P}_k^*, P_k)\right) - \frac{1}{16} \cdot s_{K+1},\]
    since $d_{TV}$ is always in the range $[0, 1].$ Now, note that $\sum_{k = 1}^{K+1} s_k \cdot d_{TV}(\mathcal{P}_k^*, \mathcal{P}_k) = d_{TV}(\mathcal{D}', \mathcal{D})$ by definition of $\mathcal{D}'.$ Also, $s_{K+1}^* \le \frac{\eps}{10},$ since $|S_{K+1}| \le n$ and each element $i \in S_{K+1}$ satisfies $P^*(i) \le \frac{\eps}{10 n}$. But since $d_{TV}(\mathcal{S}, \mathcal{S}^*) \le \frac{\eps}{10},$ this means that $s_{K+1} \le \frac{\eps}{10} + 2 \cdot \frac{\eps}{10} \le \frac{3\eps}{10}.$ Thus, we have that $\BE X \ge \frac{1}{16} \cdot \left(d_{TV}(\mathcal{D}', \mathcal{D}) - \frac{3\eps}{10}\right),$ as desired.
\end{proof}

    Therefore, if $d_{TV}(\mathcal{D}, \mathcal{D}^*) \ge \eps,$ we have that $d_{TV}(\mathcal{D}', \mathcal{D}) \ge \frac{9 \eps}{10}$ by the triangle inequality, so $\BE X \ge \frac{\eps}{32}$ by Proposition \ref{ExpectationLowerBound}. Therefore, by Proposition \ref{ReverseMarkov}, there exist constants $\alpha = 2^{-a}, \beta = 2^{-b}$ for $a, b$ nonnegative integers, such that $\alpha = \Omega(\eps),$ $\alpha \cdot \beta = \Omega\left(\frac{\eps}{\log 1/\eps}\right)$ and $\BP(X \ge \alpha) \ge \beta.$ However, if $\mathcal{D} = \mathcal{D}',$ then $X$ is uniformly $0$.
    
    Now, for each pair $(\alpha, \beta),$ we let $R = O(\beta^{-1} \cdot \log \eps^{-1})$, and we choose $R$ samples $i_1, i_2, \dots, i_{R}$ from the distribution $\mathcal{D}.$ For each $i_r$ for $1 \le r \le R,$ we let $k_r$ denote the index of the set $S_k$ that contains $i_r$ and draw $j_r$ uniformly from $S_{k_r}$. This is indeed equivalent to drawing $k_r \sim \mathcal{S}, i_r \sim \mathcal{P}_{k_r},$ and $j_r \sim Unif[S_{k_r}]$. Therefore, if $d_{TV}(\mathcal{D}, \mathcal{D}^*) \ge \eps$, with at least $9/10$ probability, some $\alpha, \beta, i_r, j_r$ satisfies 
\[\left|\frac{D(i_r)}{D(i_r)+D(j_r)} - \frac{D^*(i_r)}{D^*(i_r)+D^*(j_r)}\right| \ge \alpha.\]
    However, if $\mathcal{D} = \mathcal{D}^*$, then for all $\alpha, \beta, i_r, j_r,$ we have $\frac{D(i_r)}{D(i_r)+D(j_r)} - \frac{D^*(i_r)}{D^*(i_r)+D^*(j_r)} = 0.$
    
    Now, for each $(\alpha, \beta)$ and $1 \le r \le R$, we use $O(\alpha^{-2} \log \eps^{-1})$ Pairwise Conditional samples from $\mathcal{D}_{i_r, j_r}$ to determine $\frac{D(i_r)}{D(i_r)+D(j_r)}$ up to an $\frac{\alpha}{3}$ additive factor. If $d_{TV}(\mathcal{D}, \mathcal{D}^*) \ge \eps,$ then there will be some $(\alpha, \beta)$ and some $r \le R$ such that our estimate for $\frac{D(i_r)}{D(i_r)+D(j_r)}$ differs from $\frac{D^*(i_r)}{D^*(i_r)+D^*(j_r)}$ by at least $\frac{2 \alpha}{3}$ with probability at least $1-\eps^{10}$ by the Chernoff bound. However, if $\mathcal{D} = \mathcal{D}^*,$ then for all $(\alpha, \beta)$ and all $r \le R,$ our estimate for $\frac{D(i_r)}{D(i_r)+D(j_r)}$ differs from $\frac{D^*(i_r)}{D^*(i_r)+D^*(j_r)}$ by at most $\frac{\alpha}{3}$ with probability at least $1-\eps^{10}$ by the Chernoff bound. Therefore, the total number of samples from the total distribution and from Pair conditional samples is at most
\[\sum_{(\alpha, \beta)} O(\alpha^{-2} \log \eps^{-1} \cdot \beta^{-1} \cdot \log \eps^{-1}) = O(\log^2 \eps^{-1}) \cdot \sum_{(\alpha, \beta)} \alpha^{-2} \beta^{-1}.\]
    But since we only need to look at $\alpha = \Omega(\eps),$ $\alpha \cdot \beta = \Omega(\eps/\log \eps^{-1})$, and $\alpha, \beta$ as negative powers of $2$, the sum on the right hand side is $O(\eps^{-2} \cdot \log \eps^{-1}),$ so the total number of calls to \Call{Samp}{} and \Call{Pcond}{} in this step is at most $O(\eps^{-2} \cdot \log^3 \eps^{-1}).$ Adding this to the initial $O\left(\frac{\sqrt{\log (N/\eps)}}{\eps^2}\right)$ calls to \Call{Samp}{}, we have made a total of $\tilde{O}\left(\frac{\sqrt{\log N}}{\eps^2}\right)$ queries.
    
\section*{Acknowledgments}

I would like to thank Piotr Indyk for many helpful discussions, reading and editing drafts of this paper, and pointing me to useful references. I would also like to thank Ted Pyne for reading and editing a draft of this paper. I would also like to thank Cl\'{e}ment Canonne for answering several questions about the state of the art in the conditional sampling model. Finally, I would like to thank Ronitt Rubinfeld, Talya Eden, Sandeep Silwal, and Tal Wagner for helpful discussions.


\newcommand{\etalchar}[1]{$^{#1}$}

\newpage
\appendix 

\section{A Near-Tight Lower Bound for Identity Testing in PAIRCOND} \label{Lower}

In this section, we prove Theorem \ref{IdPcondLower}. As we noted in the introduction, the proof is very similar to \cite[Theorem 8]{CanonneRS15}, and to maintain consistency, we will adopt a similar proof structure.

\subsection{The Distribution $\mathcal{D}^*$ and Proof Intuition} \label{LowerOverview}

Let $K = \Theta\left(\eps^{-2} \log N\right)$ and let $R = \Theta\left(\frac{\log N}{\log K}\right) = \Theta\left(\frac{\log N}{\log (\eps^{-1} \log N)}\right),$ so that $N = K + K^2 + K^3 + \dots + K^{2R}.$ Now, for each $1 \le r \le 2R,$ let $B_r$ be the interval of integers starting from $1 + \sum_{i = 1}^{r-1} K^i$ and ending with $\sum_{i = 1}^{r} K^i,$ so that $|B_r| = K^r,$ and the $B_r$'s partition $[N].$ We now define $\mathcal{D}^*$ as follows. For $1 \le i \le N,$ if $i \in B_r,$ then $D^*(i) = \frac{1}{2R \cdot K^r}.$ This way, each bucket has equal probability, and for each fixed bucket, the elements all have the same probability.

We will show that it is difficult to distinguish between this distribution and a distribution randomly selected from $\mathcal{P},$ where $\mathcal{P}$ is a collection of distributions each having total variation distance at least $\frac{\eps}{2}$ from $\mathcal{D}^*$. We select a distribution $\mathcal{D} \leftarrow \mathcal{P}$ based on a random string $s \in \{0, 1\}^R.$ For each $1 \le r \le R,$ if $s = 0,$ then for all $i \in B_{2r-1},$ we choose $D(i) = \frac{1-\eps}{2R \cdot K^{2r-1}}$ and for all $i \in B_{2r},$ we choose $D(i) = \frac{1+\eps}{2R \cdot K^{2r}}$. This way, it is clear that for all strings $s$, $\sum_{i = 1}^{N} D(i) = 1$, so $\mathcal{D}$ is in fact a distribution, and that $\TV(\mathcal{D}, \mathcal{D}^*) = \frac{\eps}{2}$ for all $\mathcal{D}$.

For intuition as to why it is difficult to distinguish between $\mathcal{D}^*$ and $\mathcal{D} \leftarrow \mathcal{P},$ first note that intuitively, \Call{Pcond}{} is useless. This is because if we ever call \Call{Pcond}{$x, y$} and $x$ and $y$ are not in the same bucket, we will almost always get the one in the smaller bucket, but if $x$ and $y$ are in the same bucket, \Call{Pcond}{$x, y$} is equivalent to choosing a random element in $\{x, y\}$ regardless of whether our distribution is $\mathcal{D}^*$ or $\mathcal{D}.$ Thus, the only useful information we get is from \Call{Samp}{}. However, the only real information we get from \Call{Samp}{} is which bucket the sampled element is in. This is because beyond that, we are just sampling a uniformly random element in the bucket regardless of whether our distribution is $\mathcal{D}^*$ or $\mathcal{D}.$ However, we have $2R$ buckets, and it is known that in the sampling model, at least $\sqrt{R} \cdot \eps^{-2}$ samples are needed to test uniformity \cite{Paninski08}. And indeed, $\mathcal{D}^*$ which is uniform on the buckets and $\mathcal{D}$, when restricted to the buckets, has half of its elements with probability $\frac{1+\eps}{2 R}$ and  half of its elements with probability $\frac{1-\eps}{2 R},$ and thus has total variation distance $\frac{\eps}{2}$ from uniform. This suggests that we need $\Omega\left(\sqrt{\frac{\log N}{\log (\eps^{-1} \log N)}} \cdot \eps^{-2}\right)$ queries.

\subsection{Preliminaries}

First, we need the following well-known lemma, known as the Data Processing Inequality for Total Variation Distance.

\begin{lemma} \label{DataProcessing}
    Let $\mathcal{D}, \mathcal{D}'$ be two distributions over some probability space $\Omega.$ Let $F$ be a randomized function over $\Omega,$ which can be thought of as a distribution over functions $f$ on $\Omega.$ In other words, $F(\mathcal{D})$ is the distribution of $f(x)$ where $x \leftarrow \mathcal{D}$ and $f \leftarrow F$ (and likewise for $F(\mathcal{D}')$). Then, we have that
\[\TV(F(\mathcal{D}), F(\mathcal{D}')) \le \TV(\mathcal{D}, \mathcal{D}').\]
\end{lemma}

We will also need the following result, which is used to prove a lower bound for uniformity testing in the sampling model.

\begin{theorem} \cite[rephrased]{Paninski08} \label{SampLowerBound}
    Let $m \ge 1$ be a positive integer, and let $\mathcal{U}$ be the uniform distribution over $[2m]$. Next, draw random $s_1, \dots, s_m \leftarrow \{0, 1\}$ and let $\mathcal{Q}_s$ be a distribution over $[2m]$ such that $Q_s(2i-1) = \frac{1-\eps}{2m}$ and $Q_s(2i) = \frac{1+\eps}{2m}$ if $s_i = 0$, and $Q_s(2i-1) = \frac{1+\eps}{2m}$ and $Q_s(2i) = \frac{1-\eps}{2m}$ if $s_i = 1.$
    
    Then, if $q = o(\sqrt{m} \cdot \eps^{-2}),$ no algorithm can distinguish between $q$ samples drawn from $\mathcal{U}$ and $q$ samples drawn from $\mathcal{Q}_s$ with advantage $\Omega(1)$, where $s$ is unknown and drawn beforehand.
\end{theorem}

\subsection{The Proof}

In this subsection, we prove the following variant of Theorem \ref{IdPcondLower}:

\begin{theorem}
    Let $A$ be any (adaptive) algorithm, which makes $q = o(\sqrt{R} \cdot \eps^{-2})$ calls to \Call{Samp}{}, followed by $q$ calls to \Call{Pcond}{}. Then,
\[\left|\BP_{\mathcal{D} \leftarrow \mathcal{P}} \Big[A^{\mathcal{D}} \text{ outputs ACCEPT}\Big] - \BP \Big[A^{\mathcal{D}^*} \text{ outputs ACCEPT}\Big]\right| < \frac{1}{3},\]
    where $\mathcal{A}^{\mathcal{D}}$ implies the algorithm has has been given \Call{Samp}{} and \Call{Pcond}{} access to the distribution $\mathcal{D}$ (and likewise for $\mathcal{A}^{\mathcal{D}^*}$).
\end{theorem}

\begin{remark}
    To see why this implies Theorem \ref{IdPcondLower}, note that if an algorithm can use $q$ queries to \Call{Samp}{} and \Call{Cond}{}, it can make all the \Call{Samp}{} queries first, since they are nonadaptive queries and do not depend on any previous input or output. This implication was also used in \cite{CanonneRS15}.
\end{remark}

\begin{proof}
    We shall use the same setup used in the proof of \cite[Theorem 9]{CanonneRS15}. Namely, we first fix such an algorithm $A$, and define a \emph{transcript} for $A$ to be a pair $(Y, Z)$, such that $Y = (s_1, \dots, s_q) \in [N]^q$ and $Z = ((\{x_1, y_1\}, p_1), \dots, (\{x_q, y_q\}), p_q)),$ where each $x_i, y_i \in [N]$ and each $p_i$ is either $x_i$ or $y_i$. The $Y$ represents the $q$ samples drawn from \Call{Samp}{}, and $Z$ represents the queries $\{x_i, y_i\}$ and the output $p_i = \Call{Pcond}{x_i, y_i}$.
    
    Let $\Tr(\mathcal{D}^*)$ denote the distribution of transcripts generated by running $A$ on distribution $\mathcal{D}^*,$ and let $\Tr(\mathcal{P})$ denote the distribution of transcripts generated by first sampling $\mathcal{D} \leftarrow \mathcal{P}$ and then running $A$ on distribution $\mathcal{D}.$ Our goal will be to show that the total variation distance between the transcript distribution $\Tr(\mathcal{D}^*)$ and the transcript distribution $\Tr(\mathcal{P})$ is less than $\frac{1}{3}$ for $q = o(\sqrt{R} \cdot \eps^{-2}).$
    
    To do this, we consider the following modified algorithm $A^{(k)}$ for $0 \le k \le q.$
\begin{enumerate}
    \item $A^{(k)}$ simulates $A$ by making $q$ calls to \Call{Samp}{} and then simulates the first $k$ calls to \Call{Pcond}{}. 
    \item For each $k' > k,$ for the $(k')^{\text{th}}$ call to \Call{Pcond}{}, $A^{(k)}$ generates $(x_{k'}, y_{k'})$ as $A$ would given the output $Y$ and the output $(\{x_1, y_1\}, p_1), \dots, (\{x_{k'-1}, y_{k'-1}\}, p_{k'-1})$ that it has already seen. However, instead of calling \Call{Pcond}{}, $A^{(k)}$ does the following:
    \begin{enumerate}
        \item If $x_{k'}$ and $y_{k'}$ belong to the same block $B_\ell,$ then $p_{k'}$ is chosen uniformly from $\{x_{k'}, y_{k'}\}.$
        \item If $x_{k'}$ and $y_{k'}$ belong to different blocks, then $p_{k'}$ will just be the smaller of $x_{k'}$ and $y_{k'}$ (since the smaller element is in the smaller block).
    \end{enumerate}
\end{enumerate}
    
    We now define $\Tr^{(k)}(\mathcal{D}^*)$ as the distribution of transcripts generated by running $A^{(k)}$ on $\mathcal{D}^*$ and $\Tr^{(k)}(\mathcal{P})$ as the distribution of transcripts generated sampling $\mathcal{D} \leftarrow \mathcal{P}$ and running $A^{(k)}$ on $\mathcal{D}$. Note that $\Tr^{(q)}(\mathcal{D}^*)$ is just $\Tr(\mathcal{D}^*)$ and $\Tr^{(q)}(\mathcal{P})$ is just $\Tr(\mathcal{P})$. By an immediate application of the Triangle Inequality, the proof follows from the following two lemmas.
    
\begin{lemma} \label{15Orig}
    $\TV\left(\Tr^{(0)}(\mathcal{D}^*), \Tr^{(0)}(\mathcal{P})\right) = o(1).$
\end{lemma}

\begin{lemma} \label{16Orig}
    For all $0 \le k \le q-1,$ if $q = o(\sqrt{R} \cdot \eps^{-2})$ we have that $\TV\left(\Tr^{(k)}(\mathcal{D}^*), \Tr^{(k+1)}(\mathcal{D}^*)\right) \le \frac{1}{20 q}$ and $\TV\left(\Tr^{(k)}(\mathcal{P}), \Tr^{(k+1)}(\mathcal{P})\right) \le \frac{1}{20 q}.$
\end{lemma}

    The theorem follows since by the triangle inequality, $\TV(Tr(\mathcal{D}^*), Tr(\mathcal{P})) \le o(1) + 2q \cdot \frac{1}{20 q} = o(1) + \frac{1}{10} < \frac{1}{3},$ so the algorithm cannot even statistically distinguish between $\mathcal{D}^*$ and $\mathcal{D} \leftarrow \mathcal{P}$ using $q$ queries with advantage at least $\frac{1}{3}$.
    
    Thus, we just need to prove Lemmas \ref{15Orig} and \ref{16Orig}. In fact, the proof of Lemma \ref{16Orig} doesn't need to be changed from the corresponding proof in Canonne et al. (\cite[Lemma 16]{CanonneRS15}) at all. So, we just prove Lemma \ref{15Orig} and give an outline of the proof of Lemma \ref{16Orig}.

\begin{proof}[Proof of Lemma \ref{15Orig}]
    Note that when $A^{(0)}$ runs on $\mathcal{D}^*$ or $\mathcal{D}$, it does not ever call from the \Call{Pcond}{} oracle, and only calls upon \Call{Samp}{} $q$ times, followed by generating $\{x_k, y_k\}, p_k$ for all $1 \le k \le q,$ which is done only using the return values of the previous calls to \Call{Samp}{} and the randomness which can be generated by the algorithm $A^{(0)}$ itself. Therefore, if  $\TV\left(\Tr^{(0)}(\mathcal{D}^*), \Tr^{(0)}(\mathcal{P})\right) = \Omega(1),$ we would have a way of distinguishing between $\mathcal{D}^*$ and $\mathcal{D} \leftarrow \mathcal{P}$ with $\Omega(1)$ advantage using only $q$ queries.
    
    This, however, would give us a way of distinguishing between a uniform distribution $\mathcal{U}$ over $[2R]$ and $\mathcal{Q}_s$ for $s \leftarrow \{0, 1\}^R,$ which contradicts Theorem \ref{SampLowerBound}. To see why, suppose we were trying to distinguish between $\mathcal{U}$ and $\mathcal{Q}_s.$ Then, we simply choose $N = K + K^2 + \cdots + K^{2R}.$ For $1 \le k \le q,$ we sample $i$ from either $\mathcal{U}$ or $\mathcal{Q}_s,$ and output a random element in bucket $B_i \subset [N]$. This will have the distribution $\mathcal{D}^*$ if the original distribution were $\mathcal{U},$ and would have the distribution $\mathcal{D} \leftarrow \mathcal{P}$ with $s$ being the randomness drawn. Therefore, if we could distinguish between $\mathcal{D}^*$ and $\mathcal{D} \leftarrow \mathcal{P},$ we could also distinguish between $\mathcal{U}$ and $\mathcal{Q}_s.$ This proves the lemma.
\end{proof} 

\begin{proof}[Proof Sketch of Lemma \ref{16Orig}]
    For simplicity, let's see how to show $\TV\left(\Tr^{(k)}(\mathcal{D}^*), \Tr^{(k+1)}(\mathcal{D}^*)\right) \le \frac{1}{20 q}$. Note that for both $\Tr^{(k)}(\mathcal{D}^*)$ and $\Tr^{(k+1)}(\mathcal{D}^*)$, the algorithms $A^{(k)}$ and $A^{(k+1)}$ operate identically until the point of calling \Call{Pcond}{$x_{k+1}, y_{k+1}$}. They still generate $x_{k+1}, y_{k+1}$ in the same way, which means that the transcripts until this point have the same distribution.
    
    Now, if $x_{k+1}, y_{k+1}$ are in the same bucket $B_\ell$, the output $p_{k+1}$ will be a random element in $\{x_{k+1}, y_{k+1}\}$ for both $A^{(k+1)}$ and $A^{(k)},$ since $A^{(k)}$ just chooses a random element, and $A^{(k+1)}$ actually calls \Call{Pcond}{$x_{k+1}, y_{k+1}$}, which will give a random element. However, if $x_{k+1}$ is in bucket $\ell$ and $y_{k+1}$ is in bucket $\ell'$ for $\ell < \ell',$ the transcript of $A^{(k)}$ will always choose $p_{k+1} = x_{k+1}$ and the transcript of $A^{(k+1)}$ will choose $p_{k+1} = x_{k+1}$ with probability at least $1 - O\left(\frac{1}{K}\right)$ and will choose $y_{k+1}$ otherwise. Thus, only the rare event that the transcript chooses $p_{k+1} = y_{k+1}$ will cause the transcripts to deviate in distribution, so up to this point, the total variation distance is at most $O\left(\frac{1}{K}\right)$.
    
    Now, the rest of the protocol is the same for both $A^{(k)}$ and $A^{(k+1)}$ - namely, the rest of the transcript is just a randomized function of the current transcript. Thus, we can use Lemma \ref{DataProcessing} to say that the overall total variation distance is at most $O\left(\frac{1}{K}\right) \le \frac{1}{20 q},$ since $K = \Theta(\eps^{-2} \log N)$ and $q = o(\sqrt{R} \cdot \eps^{-2}) = o(K)$.
\end{proof}
    
\end{proof}

\newpage
\section{Pseudocode} \label{Pseudocode}

\subsection{Algorithms for Section \ref{TolerantUniformity}}

In this subsection, we write the pseudocode for all algorithms in Section \ref{TolerantUniformity}, leading to the final algorithm for Theorem \ref{TolerantUnif}. We assume \Call{Samp}{} and \Call{Pcond}{} access to $\mathcal{D}$ and that we already know the size $N$. Recall that $\mathcal{D}$ is a distribution over $[N]$ and we are trying to determine $\TV(\mathcal{D}, \mathcal{U})$ where $\mathcal{U}$ is uniform over $[N]$.
    
\begin{algorithm}
\caption{Lemma \ref{SingleElement}: Determines the probability of a given element $x$, based on the oracle.}\label{SingleElementAlg}
\begin{algorithmic}[1]
\Procedure{SingleElement}{$\eps, x, \hat{D}(x)$} \Comment{The oracle $\mathcal{O}$ returns ACCEPT on $z$ if and only if $\Call{Oracle}{\eps, x, \hat{D}(x), z} = 0$. The $\hat{D}(x)$ estimate comes from Lemma \ref{ConstantApprox}, and will be a weaker estimate than $\tilde{D}(x).$ See Algorithm \ref{Oracle} for the \Call{Oracle}{} procedure.}
\State $K = O(\eps^{-2})$
\For{$k = 1$ to $K$}
    \State $y_k \leftarrow \mathcal{U}$
    \State $z_k \leftarrow \mathcal{D}$
\EndFor
\State $\tilde{\gamma}_1, \tilde{\gamma}_2, \gamma_3 = 0$
\For{$k = 1$ to $K$}
    \If{$\Call{Oracle}{\eps, x, \hat{D}(x), y_k} = 0$}
        \State $\tilde{\gamma}_1 \leftarrow \tilde{\gamma}_1 + \frac{1}{K}$
        \While{$\Call{Pcond}{x, y_k} = y_k$}
            \State $\gamma_3 \leftarrow \gamma_3 + \frac{1}{K}$
        \EndWhile
    \EndIf
    \If{$\Call{Oracle}{\eps, x, \hat{D}(x), z_k} = 0$}
        \State $\tilde{\gamma}_2 \leftarrow \tilde{\gamma_2} + \frac{1}{K}$ 
    \EndIf
\EndFor
\State \textbf{Return} $\tilde{\gamma}_1, \frac{\tilde{\gamma_2}}{\gamma_3}$ \Comment{$\frac{\tilde{\gamma}_2}{\gamma_3}$ is our improved estimate $\tilde{D}(x).$}
\EndProcedure
\end{algorithmic}
\end{algorithm}

\begin{algorithm}
\caption{Lemma \ref{ZEstimate}: Estimates the probability of an element $z$ based on distance to $\frac{1}{N}$.}\label{ZEstimateAlg}
\begin{algorithmic}[1]
\Procedure{ZEstimate}{$\beta, x, \tilde{D}(x), z$}
\State Initialize $\tilde{D}(z)$
\For{$i = 1$ to $\log_2 \beta^{-1} + O(1)$}
    \State $\alpha = \Call{Compare}{z, x, 2^{-i}/20}$
    \State $\tilde{D}(z) = \alpha \cdot \tilde{D}(x)$
    \If{$\frac{1}{N} \cdot (1 - 2^{-i}) > \tilde{D}(z)$ or $\frac{1}{N} \cdot (1 + 2^{-i}) < \tilde{D}(z)$}
        \State \textbf{Return} $\tilde{D}(z)$
    \EndIf
\EndFor
\State \textbf{Return} $\tilde{D}(z)$
\EndProcedure
\end{algorithmic}
\end{algorithm}

\begin{algorithm}
\caption{Lemma \ref{EstimateCloseTerms}: Estimates average distance to $\frac{1}{N}$ among elements accepted by the oracle.}\label{EstimateCloseTermsAlg}
\begin{algorithmic}[1]
\Procedure{EstimateCloseTerms}{$\eps, x, \hat{D}(x)$}
\State $\tilde{\gamma}_1, \tilde{D}(x) = \Call{SingleElement}{\eps,  x, \hat{D}(x)}$
\State $T = \log_2 (\tilde{\gamma}_1/\eps) + O(1)$
\For{$t = 1$ to $T-1$}
    \State $W_{+, t} = 0$
    \State $\delta = 2^{-(t+1)}$
    \State $C = O((\delta/\eps)^2 \cdot \log \eps^{-1})$
    \For{$i = 1$ to $C$}
        \State $z \sim \mathcal{U}$
        \If{$\Call{Oracle}{\eps, x, \hat{D}(x), z} = 0$ and \Call{ZEstimate}{$c \cdot \delta, x, \tilde{D}(x), z$} $\in \left[\frac{1 + \delta}{N}, \frac{1 + 2\delta}{N}\right]$} \Comment{We can replace $\beta = O(\eps/\gamma_1)$ with $c \cdot \delta$ for a small constant $c$ since we just need to know if \Call{ZEstimate}{} returns a value in the right range of $[(1+\delta)/N, (1+2\delta)/N]$}
            \State $C' = O(\delta^{-2})$
            \State $\bar{V} = 0$
            \For{$j = 1$ to $C'$}
                \While{$\Call{Pcond}{z, x} = z$}
                    \State $\bar{V} \leftarrow \bar{V} + \frac{1}{C'}$
                \EndWhile
            \EndFor
            \State $W_{+, t} \leftarrow W_{+, t} + \frac{N \cdot \tilde{D}(x) \cdot \bar{V} - 1}{C}$
        \EndIf
    \EndFor
\EndFor
\State Similarly create $W_{-, 1}, \dots, W_{-, T-1}, W_{T}.$
\State \textbf{Return} $\sum_{t = 1}^{T-1} W_{+, t} + \sum_{t = 1}^{T-1} W_{-, t} + W_{T}$
\EndProcedure
\end{algorithmic}
\end{algorithm}

\begin{algorithm}
\caption{Lemma \ref{ConstantApprox}: Determines the probability of a set of elements up to a $1 \pm 0.1$ factor.}\label{ConstantApproxAlg}
\begin{algorithmic}[1]
\Procedure{ConstantApprox}{$\eps$}
\State $R = O(\eps^{-1} \log^2 \eps^{-1})$
\For{$r = 1$ to $R$}
    \State $x_r \leftarrow \mathcal{D}$
    \State $w_r \leftarrow \mathcal{D}$
    \State $y_r \leftarrow \mathcal{U}$
\EndFor
\For{$r = 1$ to $R$}
    \State $\tilde{d}(x_r) = 0$
    \State $\tilde{u}(x_r) = 0$
    \State $\hat{D}(x_r) = \text{NULL}$
    \For{$i = 1$ to $R$}
        \If{$0.98 < \Call{Compare}{x_r, w_i, 0.01} < 1.02$}
            \State $\tilde{d}(x_r) \leftarrow \tilde{d}(x_r) + \frac{1}{R}$
        \EndIf
        \If{$0.98 < \Call{Compare}{x_r, y_i, 0.01} < 1.02$}
            \State $\tilde{u}(x_r) \leftarrow \tilde{U}(x_r) + \frac{1}{R}$
        \EndIf
    \EndFor
    \If{$\tilde{d}(x_r) > \frac{\eps}{250 \log \eps^{-1}}$ and $\tilde{u}(x_r) > \frac{\eps}{250 \log \eps^{-1}}$}
        \State $\hat{D}(x_r) = \frac{1}{N} \cdot \frac{\tilde{d}(x_r)}{\tilde{u}(x_r)}$
    \EndIf
\EndFor 
\State \textbf{Return} $\left\{(x_r, \hat{D}(x_r)): \hat{D}(x_r) \neq \text{NULL}, \hat{D}(x_r) \in \left[\frac{0.9 \cdot \eps}{N}, \frac{1.1 \cdot \eps^{-1}}{N}\right]\right\}$
\EndProcedure
\end{algorithmic}
\end{algorithm}

\begin{algorithm}
\caption{Lemma \ref{Oracle}: Generates the oracle used in Subsection \ref{OracleGiven}}\label{OracleAlg}
\begin{algorithmic}[1]
\Procedure{Oracle}{$\eps, x, \hat{D}(x), z$} \Comment{\Call{Oracle}{$\eps, x, \hat{D}(x), z$} is the same procedure as $\mathcal{O}'(z)$. We assume that $\hat{D}(x) \in \left[\frac{5}{9N}, \frac{9}{5N}\right].$}
\State $\alpha = \Call{Compare}{z, x, 0.01}$
\If{$\alpha < 0.45$}
    \State \textbf{Return} $-1$
\ElsIf{$\alpha \le 2.2$}
    \State \textbf{Return} $0$
\Else
    \State \textbf{Return} $1$
\EndIf
\EndProcedure
\end{algorithmic}
\end{algorithm}

\begin{algorithm}
\caption{Lemma \ref{GivenGoodElt}: Solves tolerant uniformity in COND assuming we found an element $x$ with $\hat{D}(x) \approx \frac{1}{N}$.}\label{GivenGoodEltAlg}
\begin{algorithmic}[1]
\Procedure{GivenGoodElt}{$\eps, x, \hat{D}(x)$} \Comment{We assume that $\hat{D}(x) \in \left[\frac{5}{9N}, \frac{9}{5N}\right]$}
\State $K = O(\eps^{-2})$
\State $a, b, c, d, e = 0$
\State $\tilde{\gamma}_1 = 0$ \Comment{Must verify that $\gamma_1 \ge 10 \eps$ to use \Call{SingleElement}{} properly}
\For{$i = 1$ to $K$}
    \State $z \leftarrow \mathcal{D}$
    \If{$\Call{Oracle}{\eps, x, \hat{D}(x), z} = -1$}
        \State $a \leftarrow a + \frac{1}{K}$
    \ElsIf{$\Call{Oracle}{\eps, x, \hat{D}(x), z} = -1$}
        \State $c \leftarrow c + \frac{1}{K}$
    \Else
        \State $\tilde{\gamma}_1 \leftarrow \tilde{\gamma}_1 + \frac{1}{K}$
    \EndIf
    $y \leftarrow \mathcal{U}$
    \If{$\Call{Oracle}{\eps, x, \hat{D}(x), y} = -1$}
        \State $b \leftarrow b + \frac{1}{K}$
    \ElsIf{$\Call{Oracle}{\eps, x, \hat{D}(x), y} = -1$}
        \State $d \leftarrow d + \frac{1}{K}$
    \EndIf
\EndFor
\If{$\tilde{\gamma}_1 \ge 11 \cdot \eps$}
    \State $e = \Call{EstimateCloseTerms}{\eps, x, \hat{D}(x)}$ \Comment{If not, we know this value will be $O(\eps)$ so we can just have $e = 0$}
\EndIf
\State \textbf{Return} $0.5 \cdot (e + a-b-c+d)$

\EndProcedure
\end{algorithmic}
\end{algorithm}

\begin{algorithm}
\caption{Theorem \ref{TolerantUnif}: Main algorithm for tolerant uniformity testing in COND/PAIRCOND.}\label{TolerantUnifAlg}
\begin{algorithmic}[1]
\Procedure{TolerantUnif}{$\eps$}
\State $S = \Call{ConstantApprox}{\eps}$
\If{$S = \{\}$}
    \State \textbf{Return} $1$
\ElsIf{$\exists (x, \hat{D}(x)) \in S$ such that $\hat{D}(x) \in \left[\frac{5}{9N}, \frac{9}{5N}\right]$}
    \State \textbf{Return} \Call{GivenGoodElt}{$\eps, x, \hat{D}(x)$}
\Else
    \State $a, b = 0$
    \State $K = O(\eps^{-2})$
    \If{$\exists (x, \hat{D}(x)) \in S$ such that $\hat{D}(x) \ge \frac{1}{N}$}
        \For{$i = 1$ to $K$}
            \State $y \leftarrow \mathcal{U}$
            \If{$\Call{Compare}{y, x, 0.01} \ge 0.78$}
                \State $a \leftarrow a + \frac{1}{K}$
            \EndIf
        \EndFor
    \EndIf
    \If{$\exists (x, \hat{D}(x)) \in S$ such that $\hat{D}(x) \le \frac{1}{N}$}
        \For{$i = 1$ to $K$}
            \State $z \leftarrow \mathcal{D}$
            \If{$\Call{Compare}{z, x, 0.01} \le 1.27$}
                \State $b \leftarrow b + \frac{1}{K}$
            \EndIf
        \EndFor
    \EndIf
    \State \textbf{Return} $(1-a-b)$
\EndIf
\EndProcedure
\end{algorithmic}
\end{algorithm}

\newpage

\subsection{Algorithms for Section \ref{TolerantIdentity}}

In this subsection, we write the pseudocode for all algorithms in Section \ref{TolerantIdentity}, leading to the final algorithm for Theorem \ref{TolerantId}. We assume \Call{Cond}{} access to $\mathcal{D}$ and that we already know the size $N$. Recall that $\mathcal{D}$ is a distribution over $[N]$ and we are trying to determine $\TV(\mathcal{D}, \mathcal{D}^*)$ where $\mathcal{D}^*$ is a known distribution over $[N]$.

\begin{algorithm}
\caption{Lemma \ref{Est1}: Estimates $P([z])$ for given $z$.}\label{Est1Alg}
\begin{algorithmic}[1]
\Procedure{Est1}{$\eps, \delta, z, \mathcal{P}$} \Comment{$\mathcal{P}$ is a distribution over $[M]$ that we assume we can conditionally sample from.}
\State $R = O(\eps^{-1} \cdot \log \eps^{-1} \cdot \delta^{-2})$
\State $\tilde{P}([z]) = 0$
\For{$i = 1$ to $R$}
    \State $x_i \leftarrow \mathcal{P}$
    \If{$x_i \le z$}
        \State $\tilde{P}([z]) \leftarrow \tilde{P}([z]) + \frac{1}{R}.$
    \EndIf
\EndFor
\State \textbf{Return} $\tilde{P}([z])$
\EndProcedure
\end{algorithmic}
\end{algorithm}

\begin{algorithm}
\caption{Lemma \ref{Est2}: Finds some $j \le z$ along with an estimate $\tilde{Q}(j)$ of $Q(j)$.}\label{Est2Alg}
\begin{algorithmic}[1]
\Procedure{Est2}{$\eps, \delta, z, \mathcal{Q}$} \Comment{$\mathcal{Q}$ is a distribution over $[k]$ that we assume we can conditionally sample from.}
\State Run Algorithm \Call{ConstantApprox}{} with $\Call{Pcond}{}_\mathcal{Q}$ and $\Call{Samp}{}_{\mathcal{Q}}$ access: outputs a set $S$.
\If{$S = \{\}$}
    \State \textbf{Return} NULL
\EndIf
\State Else $S$ is nonempty: let $j$ such that $(j, \hat{Q}(j))$ is the first pair returned in $S$.
\State $R = O(\eps^{-2} \cdot \log^3 \eps^{-1} \cdot \delta^{-2})$
\State $\tilde{u}(j), \tilde{q}(j) = 0$
\For{$i = 1$ to $R$}
    \State $y_i \leftarrow \mathcal{U}$
    \If{$\Call{Compare}{}_{\mathcal{Q}}(y_i, j, 0.01) \in [0.98, 1.02]$}
        \State $\tilde{u}(j) \leftarrow \tilde{u}(j) + \frac{1}{R}.$
    \EndIf
    \State $x_i \leftarrow \mathcal{Q}$
    \State $X_i = 0$
    \If{$\Call{Compare}{}_{\mathcal{Q}}(x_i, j, 0.01) \in [0.98, 1.02]$}
        \State $ctr = 0$ \Comment{Counter to make sure we don't call COND more than $C \log \eps^{-1}$ times}
        \While{$\Call{Cond}{}_{\mathcal{Q}}(\{x_i, j\}) = j$ and $ctr \le O(\log \eps^{-1})$}
            \State $X_i \leftarrow X_i + 1$
            \State $ctr \leftarrow ctr+1$
        \EndWhile
    \EndIf
    \State $\tilde{q}(j) \leftarrow \tilde{q}(j) + \frac{X_i}{R}$ \Comment{$\tilde{q}_j$ is the average of the $X_i$'s}
\EndFor
\State \textbf{Return} $\left(j, \frac{\tilde{d}(j)}{k \cdot \tilde{u}(j)}\right)$
\EndProcedure
\end{algorithmic}
\end{algorithm}

\begin{algorithm}
\caption{Lemma \ref{Est3}: Estimates $Q(k)/Q(j)$.}\label{Est3Alg}
\begin{algorithmic}[1]
\Procedure{Est3}{$\eps, \delta, j, \mathcal{Q}$} \Comment{$\mathcal{Q}$ is a distribution over $[k]$ that we assume we can conditionally sample from.}
\State $R = O(\eps^{-2} \log \eps^{-1})$
\State $ctr = 0$
\For{$i = 1$ to $R$}
    \If{$\Call{Cond}{}_{\mathcal{Q}}(\{j, k\}) = k$}
        \State $ctr \leftarrow ctr + 1$
    \EndIf
\EndFor
\State $\alpha = \frac{ctr}{R-ctr}$
\If{$\alpha < 0.06 \eps^2$ or $\alpha > 18 \eps^{-2}$}
    \State \textbf{Return} $\alpha$
\Else \Comment{We will combine the $\alpha \ge 1$ and $\alpha < 1$ case here}
    \State $Y = 0$
    \State $T = O(\max(\alpha^{-1}, 1) \cdot \log^2 \eps^{-1} \cdot \delta^{-2})$
    \For{$t = 1$ to $T$}
        \State $X_t = 0$
        \While{$\Call{Cond}{}_{\mathcal{Q}}(\{j, k\}) = k$}
            \State $X_t \leftarrow X_t+1$
        \EndWhile
        \State $X_t = \min(X_t, O(\max(\alpha, 1) \cdot \log \eps^{-1}))$
    \EndFor
    \State $Y \leftarrow Y + \frac{X_t}{T}$ \Comment{$Y$ is the average of the $X_t$'s}
    \State \textbf{Return} $Y$  
\EndIf
\EndProcedure
\end{algorithmic}
\end{algorithm}

\begin{algorithm}
\caption{Lemma \ref{Est}: Estimates either $\frac{P(z)}{P^*(z)}$ or tells us that $\sum_{i \le z} \min(c_1 P(i), c_2 P^*(i))$ is small. Assumed that $\mathcal{P}, \mathcal{P}^*$ have support size $M$ with $P^*(1) \le P^*(2) \le \cdots \le P^*(M).$}\label{EstAlg}
\begin{algorithmic}[1]
\Procedure{Est}{$\eps, \delta, z, \mathcal{P}, \mathcal{P}^*, c_1, c_2$}
\State $\tilde{P}([z]) \leftarrow \Call{Est1}{\eps, \delta, z, \mathcal{P}}$
\If{$\tilde{P}(z) \le \frac{\eps}{c_1}$}
    \State \textbf{Return} $([z], 0)$
\Else
    \State Partition $[z]$ into sets $S_1, S_2, \dots, S_k$ such that $S_k = \{z\}$ and $\max P^*(S_i) \le 2 \min P^*(S_i).$ \Comment{Can be done with a simple greedy procedure}
    \State Let $\mathcal{Q}$ be the distribution over $[k]$ with $Q(i) = P(S_i)/P([z])$ \Comment{Easy to conditionally sample from $\mathcal{Q}$ if we can conditionally sample from $\mathcal{P}$}
    \State $(j, \tilde{Q}(j)) \leftarrow \Call{Est2}{\eps, \delta, z, \mathcal{Q}}$
    \If{$(j, \tilde{Q}(j)) = \text{NULL}$} \Comment{i.e., if Lemma \ref{Est2} returned NULL}
        \State \textbf{Return} $([z], 0)$
    \EndIf
    \State $Y \leftarrow \Call{Est3}{\eps, \delta, j, \mathcal{Q}}$
    \State \textbf{Return} $\frac{c_1}{P^*(z)} \cdot \tilde{P}([z]) \cdot \tilde{Q}(j) \cdot Y.$
\EndIf
\EndProcedure
\end{algorithmic}
\end{algorithm}

\begin{algorithm}
\caption{Theorem \ref{PartialDetermining}: Returns $(S, \beta)$, where $S$ is a subset with $P^*(S) \ge \frac{1}{3}$, and $\beta$ is an estimate for $\sum_{x \in S} \min(c_1 D(x), c_2 D^*(x)) = O(\eps^{-1} \log \eps^{-1})$. Assumed that $\mathcal{P}, \mathcal{P}^*$ have support size $M$ with $P^*(1) \le P^*(2) \le \cdots \le P^*(M).$}\label{PartialDeterminingAlg}
\begin{algorithmic}[1]
\Procedure{PartialDetermining}{$\eps, \mathcal{P}, \mathcal{P}^*, c_1, c_2$}
\State Let $L$ be the smallest integer such that $P^*([L]) \ge \frac{1}{2}$
\If{$P^*(L) \ge \frac{1}{3}$} \Comment{Deal with this edge case first}
    \State Compute an estimate $\tilde{P}(L) = P(L) \pm O(\eps)$ using $\tilde{O}(\eps^{-2})$ samples to $\mathcal{P}.$
    \State \textbf{Return} $(\{L\}, \min(c_1 \tilde{P}(L), c_2 P^*(L)).$
\EndIf
\State $T = \log_2 \eps^{-1} + O(1)$
\State $Y_1, \dots, Y_{T-1} = 0$ \Comment{$Y_t$ will be the average of the $Z_t$ estimates}
\For{$t = 1$ to $T-1$}
    \State $\delta = 2^{-t}/20$
    \State $S = O((\delta/\eps)^2 \cdot \log \eps^{-1})$
    \For{$s = 1$ to $S$}
        \State $z_s \leftarrow \mathcal{P}^*$
        \State $Z_s = 0$
        \If{$z_s \ge L$}
            \State $\BI_{-, t} = 0$ \Comment{Lines 15 to 24 are solely for determining $\BI_{-, t}(z)$}
            \For{$i = 1$ to $t$}
                \State $\delta' = 2^{-i}/20$
                \State $X = \Call{Est}{\eps, \delta', z_s, \mathcal{P}, \mathcal{P}^*, c_1, c_2}$
                \If{$X = ([z_s], 0)$}
                    \State \textbf{Return} $([z_s], 0)$ \Comment{Whenever Algorithm \Call{Est}{} returns $([z], 0)$, this algorithm can also automatically return $([z], 0).$}
                \EndIf
                \If{$i = t$ and $X < c_2 (1-2^{-i})$}
                    \State $\BI_{-, t} = 1$
                \ElsIf{$X \not\in [c_2 (1-2^{-i}), c_2 (1+2^{-i})]$}
                    \State \textbf{break} (out of the \textbf{for} loop starting at line 15)
                \EndIf
            \EndFor
            \If{$\BI_{-, t} = 1$}
                \State $X_s = \Call{Est}{\eps, \delta, z_s, \mathcal{P}, \mathcal{P}^*, c_1, c_2}$
                \If{$X_s = ([z_s], 0)$}
                    \State \textbf{Return} $([z_s], 0)$
                \EndIf
                \State $Z_s = c_2 - X_s$ \Comment{In all other cases, we set $Z_s = 0$}
            \EndIf
        \EndIf
        \State $Y_t \leftarrow Y_t + \frac{Z_s}{S}$
    \EndFor
    \State \textbf{Return} $([L:M], Y_1 + \dots + Y_{t-1})$.
\EndFor
\EndProcedure
\end{algorithmic}
\end{algorithm}

\begin{algorithm}
\caption{Theorem \ref{TolerantId}: Main algorithm for tolerant identity testing in COND}\label{TolerantIdAlg}
\begin{algorithmic}[1]
\Procedure{TolerantId}{$\eps$}
\State $\eps' = O(\eps/\log^2 \eps^{-1})$
\State $T_0 = [N]$.
\State $\tilde{c}_{1, 0}, c_{2, 0} = 1.$
\State $k, \gamma = 0$ \Comment{$\gamma$ will be our approximation to $1-\TV(\mathcal{D}, \mathcal{D}^*)$}
\While{$\tilde{c}_{1, k}, c_{2, k} \ge 2 \eps'$}
    \State $\mathcal{P}_k, \mathcal{P}^*_k$ are the distributions of $\mathcal{D}, \mathcal{D}^*$ conditioned on $T_k$. \Comment{We can conditionally sample from $\mathcal{P}_k$ assuming we can conditionally sample from $\mathcal{D}$. We treat $\mathcal{P}_k, \mathcal{P}^*_k$ as distributions over $[M_k]$, where $M_k= |T_k|$.}
    \State $(S, \beta) \leftarrow \Call{PartialDetermining}{\eps, \mathcal{P}, \mathcal{P}^*, \tilde{c}_{1, k}, c_{2, k}}$ 
    \State $\gamma \leftarrow \gamma+\beta$
    \State $T_{k+1} = T_k \backslash S$
    \State $c_{2, k+1} = D^*(T_{k+1})$
    \State Compute $\tilde{c}_{1, k+1} = D(T_{k+1}) \pm \eps'$ using $\tilde{O}(\eps'^{-2})$ samples from $\mathcal{D}$ and computing the fraction of samples that land in $T_{k+1}$.
    \State $k \leftarrow k+1$
\EndWhile
\State \textbf{Return} $1-\gamma$
\EndProcedure
\end{algorithmic}
\end{algorithm}

\newpage

\subsection{Algorithms for Section \ref{Monotonicity}}

In this subsection, we write the pseudocode for the algorithms of Lemma \ref{MonoPart2} and Theorem \ref{TestMonotone}. We assume \Call{Cond}{} access to $\mathcal{D}$ and that we already know the size $N$. Recall that $\mathcal{D}$ is a distribution over $[N]$ and we are trying to distinguish between $\mathcal{D}$ being monotone or $\eps$-far from monotone. We use the algorithm \Call{DistToFlat}{} of Lemma \ref{MonoPart1} as a black box.

\begin{algorithm}
\caption{Lemma \ref{MonoPart2}: Algorithm for testing exponential property}\label{ExpoTesterAlg}
\begin{algorithmic}[1]
\Procedure{ExpoTester}{$\eps$}
\State $\alpha = \eps/4$
\For{$t, b \ge 0,$ $2^{-t} = \Omega(\eps^2),$ $2^{-(a+b)} = \Omega(\eps^2/\log \eps^{-1})$}
    \State $\tau = 2^{-t}, \beta = 2^{-b}$
    \State $R = O(\beta^{-1} \log \eps^{-1})$
    \For{$r = 1$ to $R$}
        \State $i_r \leftarrow \mathcal{D}$
        \If{$i_r \ge 2$ and $\Call{Compare}{i_r, i_r-1, \frac{\tau}{3}} \ge 1 + \alpha + \frac{\tau}{2}$}
            \State \textbf{Return} REJECT
        \EndIf
    \EndFor
\EndFor
\State \textbf{Return} ACCEPT
\EndProcedure
\end{algorithmic}
\end{algorithm}

\begin{algorithm}
\caption{Theorem \ref{TestMonotone}: Main algorithm for testing monotonicity}\label{TestMonotoneAlg}
\begin{algorithmic}[1]
\Procedure{TestMonotone}{$\eps$}
\State $\alpha = \eps/4$
\State Compute Birg\'{e} Decomposition $\mathcal{I} = (I_1, I_2, \dots, I_\ell)$ of $[N]$
\State Compute $\hat{d} := \Call{DistToFlat}{\eps, \alpha}$ \Comment{approximation to $\TV(\mathcal{D}, \Phi_{\alpha}(\mathcal{D}))$}
\State Run \Call{ExpoTester}{} on reduced distribution $\mathcal{D}_{\alpha}^{\text{red}}$ by simulating conditional samples of $\mathcal{D}_{\alpha}^{\text{red}}$ with conditional samples from $\mathcal{D}$.
\If{$\hat{d} > \eps/2$ or \Call{ExpoTester}{} subroutine returns REJECT}
    \State \textbf{Return} REJECT
\EndIf
\State \textbf{Return} ACCEPT
\EndProcedure
\end{algorithmic}
\end{algorithm}

\newpage

\subsection{Algorithm for Section \ref{PcondIdentity}}

In this subsection, we write the pseudocode for Theorem \ref{IdPcondUpper}. 
We assume \Call{Samp}{} and \Call{Pcond}{} access to $\mathcal{D}$ and that we know the size $N$. Recall that $\mathcal{D}$ is a distribution over $[N]$ and we are trying to distinguish between $\mathcal{D} = \mathcal{D}^*$ and $\TV(\mathcal{D}, \mathcal{D}^*) \ge \eps$ where $\mathcal{D}^*$ is a known distribution over $[N]$.

\begin{algorithm}
\caption{Theorem \ref{IdPcondUpper}: Full algorithm for identity testing in PAIRCOND}\label{PcondIdAlg}
\begin{algorithmic}[1]
\Procedure{PcondId}{$\eps$}
\State $K = \log (10N/\eps)$
\For{$k = 1$ to $K$}
    \State $S_k = \{i: 2^{-k} < D_i^* \le 2 \cdot 2^{-k}\}$
\EndFor
\State $S_{K+1} = [N] \backslash \left(\bigcup_{i = 1}^{K} S_i\right)$
\State $\mathcal{S}$ is the distribution over $[K+1]$ where we sample $k \leftarrow \mathcal{S}$ if we sample $i \leftarrow \mathcal{D}$ and $i \in S_k$.
\State $\mathcal{S}^*$ is the distribution over $[K+1]$ where $\BP_{x \sim \mathcal{S}^*} (x = k) = \BP_{x \sim \mathcal{D}^*} (x \in S_k)$.
\State Use \cite{ValiantV17} to determine if $\mathcal{S} = \mathcal{S}^*$. Outputs ACCEPT with at least $9/10$ probability if $\mathcal{S} = \mathcal{S}^*$ and REJECT with at least $9/10$ probability if $\TV(\mathcal{S}, \mathcal{S}^*) \ge \eps.$
\If{[VV17] outputs REJECT}
    \State \textbf{Return} REJECT
\EndIf
\For{$a, b \ge 0,$ $2^{-a} = \Omega(\eps),$ $2^{-(a+b)} = \Omega(\eps/\log \eps^{-1})$}
    \State $\alpha = 2^{-a}, \beta = 2^{-b}$
    \State $R = O(\beta^{-1} \log \eps^{-1})$
    \For{$r = 1$ to $R$}
        \State $i_r \leftarrow \mathcal{D}$
        \State $k_r := $ set index such that $S_{k_r}$ contains $i_r$
        \State $j_r \leftarrow Unif\left[S_{k_r}\right]$
        \State Let $c$ be the approximation to $\frac{D(i_r)}{D(i_r)+D(j_r)}$ formed by making $O(\alpha^{-2} \log \eps^{-1})$ calls to $\Call{Pcond}{i_r, j_r}$
        \If{$\left|c - \frac{D^*(i_r)}{D^*(i_r)+D^*(j_r)}\right| \ge \frac{2\alpha}{3}$}
            \State \textbf{Return} REJECT
        \EndIf
    \EndFor
\EndFor
\State \textbf{Return} ACCEPT
\EndProcedure
\end{algorithmic}
\end{algorithm}

\end{document}